\documentclass{article}
\usepackage[utf8]{inputenc}

\newcommand{\modif}[1]{{\color{black} #1}}
\newcommand{\modifbis}[1]{{\color{black} #1}}
\newcommand{\sout}[1]{}

\usepackage{comment}

\usepackage[left=2cm,right=2cm,top=2cm,bottom=2cm]{geometry}
\usepackage{tikz}
\usetikzlibrary{fit,positioning}
\usepackage[authoryear,round,colon]{natbib}
\usepackage{here}
\usepackage{wrapfig}
\usepackage[noend]{algpseudocode}
\usepackage[titlenumbered,ruled,noend,onelanguage]{algorithm2e}
\usepackage{afterpage}

\usepackage[%
    font={small,sf},
    labelfont=bf,
    format=hang,    
    format=plain,
    margin=0pt,
]{caption}
\usepackage[list=true]{subcaption}

\usepackage{listings}
\usepackage{color}

\definecolor{dkgreen}{rgb}{0,0.6,0}
\definecolor{gray}{rgb}{0.5,0.5,0.5}
\definecolor{mauve}{rgb}{0.58,0,0.82}

\usepackage{fontawesome}

\newcommand{\pr}{P} 

\newcommand{\cov}{\mathrm{Cov}}

\newcommand{\R}{\mathbb{R}}

\newcommand{\ipsw}{\mathrm{IPSW}}
\newcommand{\aipsw}{\mathrm{AIPSW}}

\newcommand{\E}{\mathbb{E}}
\newcommand{\EE}[1][]{\mathbb{E}\left[#1\right]}
\newcommand\independent{\protect\mathpalette{\protect\independenT}{\perp}}

\def\independenT#1#2{\mathrel{\rlap{$#1#2$}\mkern2mu{#1#2}}}

\usepackage[colorinlistoftodos,textwidth=2.1cm]{todonotes}

\newcommand{\es}[1]{\todo[color=magenta, size=\tiny]{ES: #1}}

\usepackage{caption}
\usepackage{hyperref}
\hypersetup{
     colorlinks   = true,
     citecolor    = blue
}

\title{Causal effect on a target population: a sensitivity analysis to handle missing covariates}

\usepackage{amsthm, amsmath, amssymb}

\newtheorem{assumption}{Assumption}
\newtheorem{theorem}{Theorem}
\newtheorem*{theorem-non}{Theorem}
\newtheorem{corollary}{Corollary}

\newtheorem{proposition}{Proposition}
\newtheorem{definition}{Definition}
\newtheorem{lemma}{Lemma}

\author{B\'{e}n\'{e}dicte Colnet \thanks{Soda project-team, INRIA Saclay, Centre de Math\'{e}mathiques Appliqu\'{e}es, Institut Polytechnique de Paris, Palaiseau, France. (email: benedicte.colnet@inria.fr).}
  \and
Julie Josse\thanks{INRIA Sophia-Antipolis, Montpellier, France.}
  \and
Ga\"{e}l Varoquaux\thanks{Soda project-team, INRIA Saclay, France.}
     \and 
Erwan Scornet \thanks{Centre de Math\'{e}mathiques Appliqu\'{e}es, UMR 7641, \'{E}cole polytechnique, CNRS, Institut Polytechnique de Paris, Palaiseau, France.}}
\date{\today}

\begin{document}

\maketitle

\begin{center}
    {\footnotesize This article has been accepted for publication in \textit{Journal of Causal Inference}. \\Updated version is available on \href{https://www.degruyter.com/document/doi/10.1515/jci-2021-0059/html}{degruyter website}.}
\end{center}

\begin{abstract}
Randomized Controlled Trials (RCTs) are often considered the
gold standard for estimating causal effect, but they may lack external validity when the population
eligible to the RCT is substantially different from the target population.
Having at hand a sample of the target population of interest allows \modifbis{us} to generalize the causal effect.
Identifying the treatment effect in the target population requires covariates to capture all
treatment effect modifiers that are shifted between the two sets.
Standard estimators then use either weighting (IPSW), outcome modeling (G-formula), or combine the two in doubly robust approaches (AIPSW).
However such covariates are often not available in both sets.
In this paper, after \sout{completing existing proofs on the complete-case
consistency of these three estimators}\modifbis{proving $L^1$-consistency of these three estimators}, we compute the expected bias
induced by a missing covariate, assuming a Gaussian distribution, a continuous outcome,  and a
semi-parametric model. 
\modifbis{Under this setting, we perform a} sensitivity analysis for each missing
covariate pattern and compute the sign of the expected bias. We also show
that there is no gain in \modifbis{linearly imputing} \sout{imputing }a partially-unobserved covariate. Finally
we study the \sout{replacement}\modifbis{substitution} of a missing covariate by a proxy. We illustrate
all these results on simulations, as well as semi-synthetic benchmarks
using data from the Tennessee Student/Teacher Achievement Ratio (STAR),
and a real-world example from critical care medicine.
    \vspace{12pt}
\end{abstract}
\section{Introduction}

\paragraph*{Context}

Randomized Controlled Trials (RCTs) are often considered the
gold standard for estimating causal effects \citep{imbens2015causal}. 
Yet, they may lack external validity, when the population
eligible to the RCT is substantially different from the target population of
the intervention policy \citep{rothwell2005external}. Indeed, 
if there are treatment effect modifiers with a different distribution in the
target population than that in the trial,
some form of adjustment of the causal effects measured on the RCT is necessary to estimate the causal effect in the target population. 
Using covariates present in both RCT and an observational sample of the
target population, this target population average treatment effect (ATE) can be identified and estimated with a
variety of methods \citep{Hotz2005training, cole2010generalizing,
stuart2011use, bareinboim2011formaltransportability, bareinboim2013general, tipton2013improving,
bareinboim2014recovering, pearl2014externalvalidity,
kern2016generalization, bareinboim2016causalfusion,
buchanan2018generalizing, stuart2018generalizability, dong2020integrative}, reviewed in
\citep{colnet2020causal} and \citep{degtiar2021reviewgeneralizability}.

In this context, two main approaches exist to estimate the target population ATE from a RCT. 
The \textit{Inverse Probability of Sampling Weighting} (IPSW)
reweights the RCT sample so that it resembles the target
population with respect to the necessary covariates for generalization,
while the \textit{G-formula}
models the outcome, using the RCT sample, with and without treatment
conditionally \sout{to}\modifbis{on} the same covariates, and then \sout{to extend}\modif{marginalizes} the model to the
target population of interest. These two methods can be combined in
a doubly-robust
approach --\textit{Augmented Inverse Probability of Sampling Weighting}
(AIPSW)-- \sout{that}\modifbis{which} enjoys better statistical properties.
These methods rely on covariates to capture the heterogeneity of the
treatment and the population distributional shift. But the datasets describing the RCT and the target population
are seldom acquired as part of a homogeneous effort and as a result they
come with different covariates \citep{bareinboim2011formaltransportability, susukida2016representativeness,lesko2016generalizingHIV, stuart2017generalizing,
egami2021covariate, li2021generalizing}. Restricting the analysis to the covariates in common
raises the risk of omitting an important one leading to \sout{identifiabilities}\modifbis{identifiability} issues. 
Controlling biases due to unobserved covariates is of
crucial importance \sout{to}\modifbis{for} causal inference, where it is known as
\emph{sensitivity analysis} \citep{cornfield1959smoking, imbens2003sensitivity, rosenbaum2005sensitivity}. 

\paragraph*{Prior work}
 
The problem of missing covariates is central in causal inference as, in 
an observational study, one can never prove that there is no hidden
confounding. 
In that setting, sensitivity analysis strives to assess how far confounding would affect
the conclusion of a study (for example, would the ATE be of a different
sign with such a hidden confounder). Such approaches date back to a study
on the effect of smoking on lung cancer \citep{cornfield1959smoking}, and
have been further developed for both parametric
\citep{imbens2003sensitivity, rosenbaum2005sensitivity, dorie2016flexiblesensitivity,
ichino2008sensitivity, cinelli2020sensitivityextending} and
semi-parametric situations \citep{franks2019sensitivity,
veitch2020sense}. Typically, the analysis translates expert judgment
into mathematical expression of how much the confounding affects
treatment assignment and the outcome, and finally how much the estimated treatment effect is biased\sout{(giving for example bounds)}. 
In practice the expert must usually provide \textit{sensitivity parameters} 
that reflect plausible properties of the missing confounder.
Classic sensitivity analysis, dedicated to ATE estimation from observational data, 
use as sensitivity parameters the \textit{impact of the missing covariate on treatment assignment probability} 
along with the \textit{strength on the outcome} of the missing confounder. 
However, given that these quantities are hardly directly transposable when it comes to generalization,
these approaches cannot be directly applied to \sout{estimating} estimate the
population treatment effect.
These parameters have to be respectively replaced by the \sout{\textit{sampling bias}}\modif{\textit{covariate shift}} and the \textit{strength of a treatment effect modifier}\sout{in adequate sensitivity analysis method}.
\modifbis{Existing sensitivity analysis methods for generalization usually consider a completely unobserved covariate.}\sout{For sensitivity analysis of the generalization shift\footnote{Note that the term \textit{sampling selection bias} or \textit{selection bias} is also used in the literature, but for this work we prefer the term generalization shift as the term bias is used later on for estimator bias.},}
\citep{Andrews2018ValidityBias} \modifbis{rely on a logistic model for sampling probability and a linear generative model of the outcome.} \sout{study the case of a
totally unobserved covariate and model the strength of the missing covariate as a treatment effect modifier and on the sampling mechanism in a linear generative model.} 
\citep{dahabreh2019sensitivity} propose a sensitivity analysis \sout{that models}\modifbis{assuming a model on} the \modifbis{identification} bias of the conditional average treatment effect.\sout{ without the missing covariate.}
Very recent works propose two other approaches: \textit{(i)} \citep{nie2021covariate} rely on the IPSW estimator and bound the error on the the density ratio and then derive the bias on the ATE following the spirit of \citep{rosenbaum2005sensitivity}; \textit{(ii)} \citep{huang2021leveraging} present a method with very few assumptions on the data generative process leading to three sensitivity parameters, including the variance of the treatment effect.
As the analysis starts from two data sets, the missing covariate can also be partially observed in one of the two data set, which opens the door to new dedicated methods, in addition to sensitivity
methods for totally-missing covariates.
Following this observation, \citep{nguyen2017sensitivity, nguyen2018sensitivitybis} \sout{propose a sensitivity analysis dedicated to this specific purpose. In particular they design a method to} handle the case where a covariate is present in the RCT but not in the observational data set, and establish a sensitivity analysis under the hypothesis of a linear generative model for the outcome. 
When the missing covariate is partially observed, practitioners sometimes
impute missing values based on other observed covariates, though 
this approach is poorly documented. For example,
\citep{lesko2016generalizingHIV} impute \modif{a partially-observed covariate in a clinical study}
using a range of plausible distributions.
Imputation
has also been used in the context of individual participant data in
meta-analysis \citep{RescheRigon2013IDEmissing, jolani2015miceIPD}. 

\paragraph*{Contributions}
In this work we investigate the problem of a missing covariate that affects the identifiability of the target population average treatment effect (ATE), a common situation when combining different data sources. This work comes after the identifiability assessment, that is we consider that the necessary set of covariates to generalize is known, but a necessary covariate is totally or partially missing.
\sout{We focus our analysis on three main estimators IPSW, G-formula, and AIPSW.} 
Section~\ref{sec:notations} \sout{we recall the definitions of these three
estimators in the complete case -- that is when all
covariates needed to ensure identifiability are observed. This section also
establish the consistency of  IPSW, G-formula, and AIPSW estimators.}\modifbis{recalls the context along with the generic notations and assumptions used when coming to generalization.}
\modifbis{In Section~\ref{sec:linear-causal-model}, we quantify the bias due to unobserved covariates under the assumption of a semi-parametric generative process, considering a linear conditional average treatment effect (CATE), and under a transportability assumption of links between covariates in both populations.
This bias is not estimator-specific and remains valid for the IPSW, G-formula, and AIPSW estimators.} We also prove that a linear imputation of a partially missing covariate can not replace a sensitivity analysis.
As mentioned in the introduction, and unlike classic sensitivity analysis, several missing data patterns can be observed: either totally missing or missing in one of the two sets. 
Therefore Section~\ref{sec:linear-causal-model} provides sensitivity analysis frameworks for all the possible missing data patterns, including the case of a proxy variable that would replace the missing one. 
These results can be useful for users as they may be tempted to consider the intersection of common covariates between the RCT and the observational data.
We detail how the different patterns \modif{involve} \sout{imply}
either one or two sensitivity parameters. To give users an interpretable analysis, and due to the specificity of the sensitivity parameters at hands, we propose an \sout{alternative to}\modif{adaptation of} sensitivity maps \citep{imbens2003sensitivity} that are \modifbis{commonly} used to \sout{represent and }communicate sensitivity analysis results. 
\sout{In particular, our representation includes the sign of the bias along with a bias landscape rather than a threshold.}
Section~\ref{sec:simulation} presents an extensive synthetic simulation analysis that illustrates theoretical results along with a semi-synthetic data simulation using the Tennessee Student/Teacher Achievement Ratio (STAR) experiment evaluating the effect of class size on children performance in elementary schools \citep{krueger1999star}.
Finally, Section~\ref{sec:traumabase} provides a real-world analysis to assess the effect of acid tranexomic on the Disability Rating Score (DRS) for trauma patients when a covariate is totally missing.

\section{Problem setting: generalizing a causal effect \label{sec:notations}}

This section recalls the complete case context and identification assumptions. Any reader \sout{used to}\modifbis{familiar with} the notations and willing to jump to the sensitivity analysis can directly go to Section~\ref{sec:linear-causal-model}. \sout{with definitions, assumptions, and our contributions on consistency for each of the three \modifbis{following} estimators \sout{considered}: G-formula, IPSW, and AIPSW.  He/She only needs to assume the $L^1$-consistency of the three estimators, which is rigorously proved in Section~\ref{sec:notations}.} 


\subsection{Notations \label{subsec:notations}}

Notations are grounded on the potential outcome framework \citep{imbens2015causal}. We model each observation in the RCT or observational population as described by a random tuple 
$(X_i,Y_i(0),Y_i(1),A_i,S_i)$ for $i \in \{1, \hdots, n\}$ drawn from a distribution $\left(X, Y(0), Y(1), A, S \right) \in \mathbb{R}^p \times \mathbb{R}^2 \times\{0,1\}^2$, such that the observations are \textit{iid}.
For each observation, $X_i$ is a $p$-dimensional vector of covariates,
$A_i$ denotes the binary treatment assignment (with $A_i=1$ if treated and $A_i=0$ otherwise), $Y_i(a)$ is the continuous outcome had the subject been given treatment $a$ (for $a\in\{0,1\}$), and $S_i$ is a binary indicator for RCT eligibility (i.e., meet the RCT inclusion and exclusion criteria) and willingness to participate if being invited to the trial ($S_i=1$ if eligible and $S_i=0$ if not). Assuming \textit{consistency of potential outcomes}, and no \sout{interaction}\modifbis{interference} between treated and non-treated subject (SUTVA assumption), we denote by $Y_i = A_iY_i(1)+(1-A_i)Y_i(0)$ the \modifbis{observed outcome under treatment assignment $A_i$}. 

\modif{Assuming the potential outcomes are integrable, we} define the conditional average treatment effect (CATE):
$$
\forall x\in \mathcal{X}\,,\quad \tau(x) = \EE [Y(1) - Y(0) \mid X=x]\,,
$$
and the population average treatment effect (ATE):
$$
\tau = \EE [Y(1) - Y(0)] = \EE[\tau(X)]\,.
$$

\modif{Unless explicitly stated, all expectations are taken with respect to all variables involved in the expression.}
We model the patients belonging to an RCT sample of size $n$ and in an observational data sample of size $m$ by $n+m$ independent random tuples: $\{X_i,Y_i(0),Y_i(1),A_i,S_i\}_{i=1}^{n+m},$ where the RCT samples $i=1,\ldots,n$ are identically distributed according to $\mathcal{P}(X,Y(0),Y(1),A,S \mid S=1)$, and the observational data samples $i=n+1,\ldots,n+m$ are identically distributed according to $\mathcal{P}(X,Y(0),Y(1),A,S)$. We also denote $\mathcal{R}=\{1,\ldots,n\}$ the index set of units observed in the RCT study, and $\mathcal{O}=\{n+1,\ldots,n + m\}$ the index set of units observed in the observational study.\\
For each RCT sample $i\in \mathcal{R}$, we observe $(X_i,A_i,Y_i,S_i=1)$, while for observational data $i\in \mathcal{O}$, we consider the setting where \textit{we only observe the covariates $X_i$}, which is a common case in practice. A typical data set is presented on Table~\ref{tab:typicalsituation}.

Because the RCT sample and observational data do not follow the same covariate distribution, the ATE $\tau$ is different from the RCT's (or sample\footnote{Usually $\tau_1$ is also called the Sample Average Treatment Effect (SATE), when $\tau$ is named the Population Average Treatment Effect (PATE) \citep{stuart2011use, miratrix2017worth, egami2021covariate, degtiar2021reviewgeneralizability}.}) average treatment effect $\tau_1$ which can be expressed as:
\begin{equation*}
    \tau \neq \tau_1, \quad \text{where}\; \tau_1 := \mathbb{E} [Y(1) - Y(0) \mid S=1]\,.
\end{equation*}

This difference is the core of the \textit{lack of external validity} introduced in the beginning of the work, but formalized with a mathematical expression\footnote{We would like to emphasize the fact that the target quantity is not $\mathbb{E} [Y(1) - Y(0) \mid S=0]$, but $\tau := \mathbb{E} [Y(1) - Y(0)]$. This notation highlights that the trial sample is a biased sample from a superpopulation, while the observational data is an unbiased sample of this population. In other words, the target population contains individuals with $S=1$ or $S=0$. Note that the generalizability problem tackled in this work  - aiming to recover from a sampling bias - can also be equivalently seen as a transportability problem with two separate populations and a common support. See \cite{colnet2020causal} for a discussion, or \cite{nie2021covariate} for a similar sensitivity analysis method, presented as a transportability problem.}. Throughout the paper, we denote $\mu_a(x) := \EE[Y(a) \mid X=x]$
the conditional mean outcome under treatment $a\in\{0,1\}$ \modifbis{(also called responses surfaces)}.\sout{ in the observational data.}
\sout{$\mu_{(a)}$ are also called responses surfaces}
and  $e_{1}(x) := \mathbb{P}(A=1 \mid X=x, S=1)$ the propensity score in the RCT population.
This function is imposed by the trial characteristics and is usually a constant denoted by $e_1$ (other cases include stratified RCT trials).

For notational clarity, estimators are indexed by the number of observations used for their computation. 
For instance, response surfaces can be
estimated using controls and treated individuals 
in the RCT to obtain respectively $\hat{\mu}_{0,n}$ and $\hat{\mu}_{1,n}$.
Similarly, we denote by $\hat \tau_n$ an estimator of $\tau$ depending only on the RCT samples (for example the difference-in-means estimator), and by $\hat \tau_{n,m}$ an estimator computed using both datasets.

\modifbis{\subsection{Identifiability (or causal) assumptions} \label{subsec:assumption-to-generalize}}


The \textit{consistency of treatment assignment} assumption ($Y=AY(1) + (1-A)Y(0)$) has already been introduced in  Section~\ref{sec:notations}\sout{are}\sout{and randomization within the RCT ($Y(a) \independent A\mid S=1$)}. \sout{Note that to ensure the internal validity of the RCT a last assumption has to be verified, that is, all units from the RCT have a non-zero probability to be in the treated or in the control group.}To ensure the internal validity of the RCT, \modifbis{we need to assume} randomization of \sout{the} treatment assignment and \sout{the} positivity of trial treatment assignment.
\modif{
\begin{assumption}[Treatment randomization within the RCT]\label{a:RCT-randomization}
$\forall a \in \{0,1\} , \,Y(a) \independent A\mid S=1, X$.
\end{assumption}

In some cases, the trial is said to be completely randomized, that is $\forall a \in \{0,1\}, \, Y(a) \independent A\mid S=1$, thus removing any potential stratification of the treatment assignment.}

\begin{assumption}[Positivity of trial treatment assignment]\label{a:eta_1_bounded}
$\exists \eta_1 > 0, \forall x  \in \mathcal{X},  \eta_1 \leq  e_1(x) \leq 1- \eta_1$
\end{assumption}

\modifbis{Under these two assumptions, along with the SUTVA assumption (see, e.g.,  \cite{imbens2015causal}), the most classical difference-in-means estimator is consistent for $\tau_1$.} 
In order to generalize the RCT estimate to the target population, three additional assumptions are required for identification of the target population ATE $\tau$.

\begin{assumption}[Representativity of observational data]\label{a:repres} For all $i \in \mathcal{O}, X_i \sim \mathcal{P}(X)$ where $\mathcal{P}$ is the target population distribution. 
\end{assumption}

Then, a key assumption concerns the set of covariates that allows the identification of the target population treatment effect. This implies a conditional independence relation being called the \textit{ignorability assumption on trial participation} or \textit{S-ignorability} \citep{Hotz2005training, stuart2011use, tipton2013improving, hartman2015sate, pearl2015findings, kern2016generalization, stuart2017generalizing, nguyen2018sensitivitybis, egami2021covariate}. 

\begin{assumption}[Ignorability assumption on trial participation - \cite{stuart2011use}]
\label{a:cate-indep-s-knowing-X}
$Y(1) - Y(0) \independent S \mid  X$.
\end{assumption}
Assumption~\ref{a:cate-indep-s-knowing-X} indicates that covariates $X$ needed to generalize correspond to covariates being \textit{both} treatment effect modifiers \textit{and} subject to a distributional shift between the RCT sample and the target population.
\modif{\sout{Research works} \modifbis{Different strategies} have been proposed to assess whether \sout{or not} a treatment effect is constant or not, and usually relies on marginal variance, CDFs or quantiles comparison \citep{Miratrix2016ATEVariation}.}
Other \modifbis{techniques} are possible such as comparing $\operatorname{Var}[Y \mid X_{obs}, A=1, S=1]$ to $\operatorname{Var}[Y \mid X_{obs}, A=0, S=1]$, \modifbis{in order to  assess whether or not} an important treatment effect modifier is missing.
\sout{Still, the scope of this work is limited to the case where} \modifbis{In our work, we assume that the user is aware of which variables are treatment effect modifiers and subject to a distributional shift.}
We call these covariates as \textit{key covariates}. 

\begin{assumption}[Positivity of trial participation - \cite{stuart2011use}]\label{a:pos}There exists a constant
$c$ such that for all $x$ with probability $1$,  \mbox{$\mathbb{P}(S=1\mid X = x)\ge c>0$}
\end{assumption} 


\subsection{Estimation strategies}

To \sout{properly} transport the ATE, several methods exist: the G-formula
\citep{lesko2017generalizing, bareinboim2011formaltransportability, dahabreh2019generalizing}, Inverse Propensity Weighting Score (IPSW) \citep{cole2010generalizing, lesko2017generalizing, buchanan2018generalizing}, and the Augmented IPSW (AIPSW) estimators. Note that other methods exist, such as calibration \citep{dong2020integrative, Chattopadhyay2022OneStep}. 
For example the G-formula estimator consists in modeling the expected values for each potential outcome, conditional on the covariates. 
\begin{definition}[G-formula - \cite{dahabreh2018generalizing}]\label{def:g-formula} The G-formula is denoted $\hat \tau_{\text{\tiny G},n,m}$, and defined as
\begin{equation}
\hat \tau_{\text{\tiny G},n,m} = \frac{1}{m}\sum_{i=n+1}^{n+m}\left(\hat \mu_{1, n}(X_i) - \hat \mu_{0, n}(X_i)\right),
\end{equation}
where $\hat \mu_{a,n}(X_i)$ is an estimator of $\mu_{a}(X_i)$ obtained on the RCT sample. These intermediary estimates are called nuisance components.
\end{definition}

\modifbis{Beyond causal assumptions stated above, the behavior of the G-formula estimator strongly depends on that of the surface response estimators $\hat \mu_{a,n}$ for $a \in \{0,1\}$. To analyze the G-formula, we introduce below assumptions on the consistency of the nuisance parameters $\hat \mu_{0,n}$ and $\hat \mu_{1,n}$.} 
\sout{One can observe that the surface response estimators condition the consistency of the estimator. We introduce an assumption on the consistency\sout{ rate} of those nuisance parameters.}

\begin{assumption}[Consistency of surface response estimators]\label{a:consistency-mu}
Denote $\hat \mu_{0,n}$ (respectively $\hat \mu_{1,n}$) an estimator of $\mu_{0}$ (respectively $\mu_{1}$). Let $\mathcal{D}_n$ the RCT sample, so that 
    \item (H1-G) For $a \in \{0,1\}$, $\mathbb{E}\left[ | \hat \mu_{a,n}(X) -  \mu_{a}(X) | \mid \mathcal{D}_n\right] \stackrel{p}{\rightarrow}  0$ when $n \rightarrow \infty$,
    \item (H2-G) For $a \in \{0,1\}$, there exist $C_1, N_1$ so that for all $ n \geqslant N_{1}$, almost surely, $\mathbb{E}[\hat{\mu}_{a,n}^2(X) \mid \mathcal{D}_n] \leqslant C_1$.
\end{assumption}


\begin{proposition}[Informal - $L^1$-consistency of G-formula, IPSW, and AIPSW]
Under causal assumptions (Assumptions~ \ref{a:RCT-randomization}, \ref{a:eta_1_bounded}, \ref{a:repres}, \ref{a:cate-indep-s-knowing-X}, and \ref{a:pos}) and Assumption~\ref{a:consistency-mu}, the G-formula is $L^1$-consistent (asymptotically unbiased). In appendix we recall definitions of IPSW and AIPSW estimators and give the precise conditions under which $L^1$-consistency of those estimators is achieved (see Section~\ref{sec:consistency}).  
\end{proposition}
Proofs and a more formal statement are in Section~\ref{sec:proof-consistency-ipsw-gformula}. The sensitivity analysis presented below holds for any $L^1$-consistent estimator.

\section{Impact of a missing key covariate for a linear CATE \label{sec:linear-causal-model}}

\sout{In this section, we assume thatFor the rest of the work, we assume $X, Y(0), Y(1) \in \mathbb{R}^{p+2}$.while previous results did not require such assumption.}

\subsection{Situation of interest: a missing covariate \sout{on}\modifbis{in} one dataset \label{subsec:start-missing}}

We study the common situation where both data sets (RCT and observational) contain a different subset of the total covariates $X$.
$X$ can be decomposed as $X = X_{m i s} \cup X_{o b s}$ where $X_{o b s}$ denotes the covariates that are present\sout{s} in both data sets, the RCT and the observational study. $X_{m i s}$ denotes the covariates that are either partially observed in one of the two data sets or totally unobserved in both data sets. \sout{Note that in this work}We do not consider (sporadic) missing \modifbis{data problems} as in \cite{mayer2021generalizingincomplete}, but only cases where the covariate is totally observed or not per data sources.
We denote by $obs$ (resp. $mis$) the index set of observed (resp. missing) covariates. 
An illustration of a typical data set\sout{we consider} is presented in Table~\ref{tab:typicalsituation}, with an example of two missing data patterns.

\begin{figure}[!h]
\begin{center}
\begin{tabular}{ |c|c|ccc|c|c| } 
 \hline
 & &\multicolumn{3}{c|}{Covariates} &  & \\
 & Set &$X_1$ & $X_2$ & $X_3$ & $A$ & $Y$ \\ 
 \hline
1 & $\mathcal{R}$ &1.1 & 20 & 5.4 & 1 &  10.1 \\ 
 & $\mathcal{R}$ &-6 & 45 & 8.3 &  0 & 8.4  \\ 
$n$& $\mathcal{R}$ &0 & 15 & 6.2 & 1 & 14.5  \\ 
$n+1$& $\mathcal{O}$ && $\dots$ & & $\dots$ & $\dots$   \\
& $\mathcal{O}$ &-2 & 52 & \texttt{NA} & \texttt{NA} & \texttt{NA}  \\
& $\mathcal{O}$ &-1 & 35 & \texttt{NA} & \texttt{NA} &    \texttt{NA} \\
$n+m$& $\mathcal{O}$ &-2 & 22 & \texttt{NA} & \texttt{NA} &\texttt{NA} \\
 \hline
\end{tabular}
 \hspace{1cm}
\begin{tabular}{ |c|c|ccc|c|c| } 
 \hline
 & &\multicolumn{3}{c|}{Covariates} &  & \\
&  Set &$X_1$ & $X_2$ & $X_3$ & $A$ & $Y$ \\ 
 \hline
1& $\mathcal{R}$ &1.1 & 20 & \texttt{NA} & 1 &  10.1 \\ 
&  $\mathcal{R}$ &-6 & 45 & \texttt{NA} & 0  & 8.4  \\ 
$n$&  $\mathcal{R}$ &0 & 15 & \texttt{NA} & 1 & 14.5  \\ 
$n+1$&  $\mathcal{O}$ && $\dots$ & & $\dots$ & $\dots$   \\
& $\mathcal{O}$ &-2 & 52 & 3.4 & \texttt{NA} & \texttt{NA}  \\
&  $\mathcal{O}$ &-1 & 35 & 3.1 & \texttt{NA} & \texttt{NA} \\
$n+m$& $\mathcal{O}$ &-2 & 22 & 5.7 & \texttt{NA} & \texttt{NA} \\
 \hline
\end{tabular}
\end{center}
\caption{\textbf{Typical \modifbis{data} structure\sout{of the data considered}}, where a covariate would be available in the RCT, but not in the observational data set (left) or the reverse situation (right). In this specific example, $obs = \{1,2\}$ ($mis = \{3\}$), corresponds to common (resp. different) covariates in the two datasets.}
\label{tab:typicalsituation}
\end{figure}

\sout{As a consequence}\modifbis{In our context, due to (partially-)unobserved covariates,} estimators of the target population ATE may be implemented on $X_{o b s}$ only. 
\modifbis{To make the notations clear}, we add a subscript \textit{obs} on any estimator applied on the set $X_{o b s}$ rather than $X$. 
Such estimators may suffer from bias due to  Assumption~\ref{a:cate-indep-s-knowing-X} violation, that is:
\begin{align*}
   Y(1) - Y(0) \independent S \mid X \quad \textrm{but} \quad     Y(1) - Y(0) \not \independent S \mid X_{obs}
\end{align*}

We denote \sout{$\hat \tau_{\text{\tiny G}, n, m, o b s}$}\modifbis{$\hat \tau_{n, m, o b s}$} any generalization estimator \modifbis{(G-formula, IPSW, AIPSW)} applied on the covariate set $X_{o b s}$ rather than $X$. 
\sout{This estimator uses the same observations, but is different from $\hat \tau_{\text{\tiny G}, n, m}$ (Definition~\ref{def:g-formula}) as it requires estimation of nuisance components $\hat \mu_{a, n, obs}(x_{obs})$ for $a \in \{0,1\}$, which corresponds to the estimated surface response on the set $X_{obs}$ instead of $X$.}


\subsection{Expression of the missing-covariate bias}

\subsubsection{Model and hypothesis}
\modifbis{To analyze the effect of a missing covariate, we introduce a nonparametric generative model.}
\modifbis{In particular, we focus on zero-mean additive-error representation,}
where the CATE depends linearly on $X$. We admit \modifbis{that} there exist $\delta \in \mathbb{R}^{p}$, $\sigma \in \mathbb{R}^{+}$\modifbis{, and a function $g:\mathcal{X} \to \mathbb{R}$,} such that:
\begin{equation}
\label{eq:linear-causal-model}
    Y= g(X) + A \langle X, \delta\rangle + \varepsilon, \qquad \text{where } \varepsilon \sim \mathcal{N}\left(0, \sigma^2\right),
\end{equation}
assuming $\tau(X) := \langle X, \delta\rangle$. \modifbis{In appendix (see Section~\ref{subsec:toward-non-param}) we prove why this assumption on the generative model for $Y$ does not come with a loss of generality.}

\modif{Under this model, the Average Treatment Effect (ATE) takes the following form:
\begin{equation*}
 \label{eq:ate-form-when-linear-model}
    \tau=\int \mathbb{E}\left[Y(1)-Y(0) \mid X=x\right] f(x) \mathrm{d}x=\int \langle \delta, x \rangle f(x) \mathrm{d}x=\delta^{T} \mathbb{E}[X] \modifbis{.}
\end{equation*}}

Only variables that are both treatment effect modifier ($\delta_j \neq 0$) and subject to a distributional change between the RCT and the target population are necessary to generalize the ATE. If \sout{such a}\modifbis{some of these} key covariates are missing, the estimation of the target population ATE \modifbis{will be biased}. \modifbis{Our goal here is} \sout{We want} to express the bias of \sout{such an omitted} \modifbis{a missing} variable on the transported ATE. But first, we have to specify a context in which a certain permanence of the relationship between $X_{o b s}$ and $X_{m i s}$ in the two data sets holds. Therefore, we introduce the \textit{Transportability of covariate\sout{s} relationship} assumption.

\begin{assumption}[Transportability of covariate\sout{s} relationship]\label{a:trans-sigma} The distribution of $X$ is Gaussian, that is, $X \sim \mathcal{N}\left(\mu, \Sigma\right)$, and transportability of $\Sigma$ is true, that is, $X \mid S = 1 \sim \mathcal{N}\left(\mu_{RCT}, \Sigma\right)$.
\end{assumption} 

This assumption, \modifbis{and in particular, the transportability of $\Sigma$,}  is of major importance for the sensitivity analysis \modifbis{we develop below}\sout{proposed, in particular the transportability of $\Sigma$}. Indeed, as soon as the correlation pattern changes in amplitude and sign between the two populations, the sensitivity analysis can be invalidated.
The plausibility of Assumption~\ref{a:trans-sigma} can be partially assessed through a statistical test on $\Sigma_{obs,obs}$ for example \modifbis{a} Box's M test \citep{box1949boxmtest}, supported with vizualizations \citep{friendly2020covariancematrixviz}. A discussion can be found in the experimental study (Section~\ref{sec:simulation}) and in appendix (Section~\ref{appendix:assumption-8}), \modifbis{showing that this assumption is plausible in many situations}.

\subsubsection{Main result}

\begin{theorem}
\label{lemma:linear-gformula-unbiased}
 Assume that Assumptions~\ref{a:RCT-randomization}, \ref{a:eta_1_bounded}, \ref{a:repres}, \ref{a:cate-indep-s-knowing-X}, \ref{a:pos} (identifiability) hold, along with Model \eqref{eq:linear-causal-model} and Assumption~\ref{a:trans-sigma} (sensitivity model). Let B be the following quantity:
 \begin{align}
        B = - \sum_{j\in mis} \delta_j  \left( \EE[X_j] - \EE[X_j \mid S=1] - \Sigma_{j, o b s}\Sigma_{o b s, o b s}^{-1}(\EE[X_{o b s}]-\EE[X_{o b s}\mid S = 1]) \right),
        \label{eq_bias_definition}
\end{align}
 where $\Sigma_{o b s, o b s}$ is the submatrix of $\Sigma$ \sout{corresponding to observed index rows and columns, } \modifbis{composed of rows and columns corresponding to variables present in both data sets}. \modifbis{Similarly, $\Sigma_{j, o b s}$ is composed of the $j$th row of $\Sigma$ and has columns corresponding to variables present in both data sets.} \sout{  rows and columns corresponding to variables present in both data setsand $\Sigma_{j, o b s}$ is the row $j$ with column corresponding to observed index of $\Sigma$.} 
 \modif{Consider a procedure  $\hat \tau_{n,m}$ that estimates $\tau$ with no asymptotic bias (for example the G-formula introduced in Definition~\ref{def:g-formula} under Assumption~\ref{a:consistency-mu}). Let $\hat \tau_{n, m, o b s}$ be the same procedure but trained on observed data only. Then
 \begin{align}
      \lim\limits_{n,m \to \infty} \mathbb{E}[\hat \tau_{n, m, o b s}] - \tau = B.
      \label{eq_th_bias}
     \end{align}

 } 


\end{theorem}

Proof is given in appendix (see Section~\ref{sec:appendix-original-proof}).

\paragraph{Comment on $L^1$-consistency}
\modifbis{Theorem~\ref{lemma:linear-gformula-unbiased} is valid for any $L^1$-consistent generalization estimator. In particular, we provide in appendix the detailed assumptions (similar as Assumption~\ref{a:consistency-mu}) under which two other popular estimators, IPSW and AIPSW, are asymptotically unbiased (see Section~\ref{sec:consistency}).}
Note that most of the existing works on estimating the target population causal
effect focus on identification or establish consistency
for parametric models or oracle estimators which are not bona fide
estimation procedures as they require knowledge of some population
data-generation mechanisms
\citep{cole2010generalizing, stuart2011use, lunceford04stratificationand, buchanan2018generalizing, correa2018selectionbias, dahabreh2019generalizing, egami2021covariate}. 
To our knowledge, no general $L^1$-consistency results for the G-formula, IPSW, and AIPSW procedures are available in a non-parametric case, when either the CATE or the weights are estimated from the data without prior knowledge.
\sout{Such general consistency results are established in this work. We detail such results in the main text only for the G-formula. IPSW and AIPSW results are detailed in appendix (see Section~\ref{sec:consistency}), while the proofs of all $L^1$-consistency results are in appendix (see Section~\ref{sec:proof-consistency-ipsw-gformula}). In Section~\ref{sec:linear-causal-model}, we build on these results to compute the bias associated to a missing covariate.}

\paragraph{What if outcomes are also available in the observational sample?} Who can do more can do less, therefore this outcome covariate could be dropped and the analysis conducted without it. But alternative strategies exist. First, the outcome in the observational data -- even if present in only one of the treatment group --  would allow to test for the presence or absence of a missing treatment effect modifier \citep{degtiar2021reviewgeneralizability} (see their Section 4.2), and therefore its strength. Moreover this would allow to rely on strategies to diminish the variance of the estimates \citep{huang2021leveraging}. Finally, the assumption of a linear CATE could be reconsidered and softened, but we let this question to future work.

\subsection{Sensitivity analysis}

The above theoretical bias $B$ (see equation~\ref{eq_bias_definition}) can be used to translate expert judgments about the strength of  missing covariates, which corresponds to sensitivity analysis. \modifbis{In the rest of our work, we exemplify Theorem~\ref{lemma:linear-gformula-unbiased} in scenarios for which there is a totally unobserved covariate (Section~\ref{subsec:totally-unobserved}), a missing covariate in RCT (Section~\ref{subsec:observed-obs}), or a missing covariate in the observational sample (Section~\ref{sec:nguyen-sensitivity-method}).}\sout{study the case 
Section~\ref{subsec:totally-unobserved} details the case of a totally unobserved covariate, Section~\ref{subsec:observed-obs} the case of a missing covariate in RCT, Section~\ref{sec:nguyen-sensitivity-method} the case of a missing covariate in the observational sample.} Section~\ref{subsec:austen-plots} completes the previous sections presenting an \sout{extension}\modif{adaptation} to \modif{sensitivity maps} \sout{Austen plots}. Finally Section~\ref{subsec:imputation} details the imputation case, and Section~\ref{sec:proxy-variable} the case of a proxy variable.
All these methods rely on different assumptions recalled in Table~\ref{tab:recap-linear-methods}.
\begin{table}[!h]
\begin{center}
\begin{tabular}{lll}
\hline
\multicolumn{1}{c}{\textbf{Missing covariate pattern}} & \multicolumn{1}{c}{\textbf{Assumption(s) required}} & \multicolumn{1}{c}{\textbf{Procedure's label}} \\ \hline
Totally unobserved covariate                             & \textit{$X_{m i s} \independent X_{o b s}$} & \ref{algo:totally-unobserved}        \\
Partially observed in observational study                & \textit{Assumption~\ref{a:trans-sigma}}                            & \ref{algo:observed-obs}    \\
Partially observed in RCT                                & \textit{No assumption}                           & \ref{algo:observed-rct}   \\
Proxy variable                                           & \textit{Assumptions~\ref{a:trans-sigma} and \ref{a:proxy}}  & \ref{algo:proxy}  \\                    
\hline
\end{tabular}
\caption{\textbf{Summary of the assumptions and results pointer for all the sensitivity methods} according to the missing covariate pattern when the generative outcome is semi-parametric with a linear CATE \eqref{eq:linear-causal-model}.}
\label{tab:recap-linear-methods}
\end{center}
\end{table}
\vspace{-0.8cm}
\subsubsection{Sensitivity analysis when a key covariate is totally unobserved \label{subsec:totally-unobserved}}
\modifbis{When a covariate is totally unobserved, a common and natural assumption is to assume independence between this covariate and the observed ones \citep{imbens2003sensitivity}. Although strong, this assumption allows us to estimate the identification bias.}
\sout{The approach we propose relies on \cite{imbens2003sensitivity}'s prototypical framework. Suppose we have a totally unobserved key covariate, \sout{from} \modifbis{for} which the association with observed covariates is known and
\footnote{If $S$ was a child of covariates in a causal diagram, this independence assumption could be challenged. Indeed, in such a situation the association of $X_{obs}$ and $X_{mis}$ may be different between the trial sample and the target population due to collider bias when conditioning on sample membership. A nested-trial design would fall in this situation.}:
$X_{m i s} \independent X_{o b s}$  and $X_{o b s} \independent X_{m i s} \mid S = 1$.
We also suppose that the complete parametric model is: }

\begin{corollary}[Sensitivity model]
\label{lemma:sensitivity-bias-linear}
Assume that Model \eqref{eq:linear-causal-model} holds, along with Assumptions~\ref{a:RCT-randomization}, \ref{a:eta_1_bounded}, \ref{a:repres}, \ref{a:cate-indep-s-knowing-X}, \ref{a:pos}, and \ref{a:trans-sigma}. \modifbis{Assume also that} $X_{m i s} \independent X_{o b s}$, \sout{ and}  $X_{m i s} \independent X_{o b s} \mid S = 1$.\sout{, along with the sensitivity model \eqref{eq:sensitivity-model-linear}.} 
\modif{Consider a procedure  $\hat \tau_{n,m}$ that estimates $\tau$ with no asymptotic bias. Let $\hat \tau_{n, m, o b s}$ be the same procedure but trained on observed data only. Then
$$ \lim\limits_{n,m \to \infty} \mathbb{E}[\hat \tau_{n, m, o b s}] - \tau  = - \delta_{m i s} \, \Delta_m$$}
where $\Delta_m = \mathbb{E}[X_{m i s}] - \mathbb{E}[X_{m i s} \mid S=1]$. 
\end{corollary}
Corollary~\ref{lemma:sensitivity-bias-linear} is a direct consequence of  Theorem~\ref{lemma:linear-gformula-unbiased}, particularized for the case where 
$X_{o b s} \independent X_{m i s}$ and $X_{o b s} \independent X_{m i s} \mid S = 1$. In this expression, $\Delta_m$ and $\delta_{m i s}$ are called the sensitivity parameters.
To estimate the bias implied by an unobserved covariate, we have to determine how strongly $X_{m i s}$ is a treatment effect modifier (through $\delta_{m i s}$), and how strongly it is linked to the trial inclusion (through the shift between the trial sample and the target population $\Delta_m = \EE[X_{m i s}] - \EE[X_{m i s} \mid S=1]$). 
Table~\ref{tab:comparison-imbens-sensitivity} summarizes the similarities and differences with \cite{imbens2003sensitivity}, \cite{Andrews2018ValidityBias}'s approaches, and our approach.

\begin{table}[!h]
\begin{center}
\footnotesize
\begin{tabular}{|l|l|l|l|}
\hline
& \multicolumn{1}{c|}{\cite{imbens2003sensitivity}}              & \cite{Andrews2018ValidityBias} &  \multicolumn{1}{c|}{Sensitivity model} \\ \hline
Assumption on covariates  &  $X_{m i s} \independent X_{o b s}$ &  $X_{m i s} \independent X_{o b s}$  & $X_{m i s} \independent X_{o b s}$   \\ \hline
Model on $Y$ & Linear model & Linear model  & Linear CATE \eqref{eq:linear-causal-model} \\ \hline
Other assumption & Model on $A$ (\texttt{logit}) &  Model on $S$  (\texttt{logit})  & - \\ \hline
First sensitivity  parameter & Strength on $Y$, using $\delta_{m i s}$ &  Strength on $Y$, using $\delta_{m i s}$                                &              Strength on $Y$, using $\delta_{m i s}$                          \\ \hline
Second sensitivity parameter &   Strength on $A$ (\texttt{logit}'s
coefficient) & Strength on S (\texttt{logit}'s coefficient) &
$\Delta_m$: shift of $X_{mis}$      \\ \hline
\end{tabular}
\caption{Summary of the differences \sout{in} between \cite{imbens2003sensitivity}'s method, being a prototypical method for sensitivity analysis for observational data and hidden counfounding, \cite{Andrews2018ValidityBias}'s method and our method.}
\label{tab:comparison-imbens-sensitivity}
\end{center}
\end{table}

\modifbis{In the setting of Corollary~\ref{lemma:sensitivity-bias-linear}, sensitivity analysis can be carried out using Procedure~\ref{algo:totally-unobserved} described below}
\sout{In practice in this setting of a completely missing covariate sensitivity analysis can be proceed using Procedure~\ref{algo:totally-unobserved}}.
To represent the \modifbis{bias magnitude} \sout{amplitude of bias} \modifbis{as a function of}\sout{ depending on} the sensitivity parameters \sout{values}, \modifbis{we develop} a graphical aid \sout{derived}\modifbis{adapted} from  \modif{sensitivity map\modifbis{s}} \citep{imbens2003sensitivity, veitch2020sense} \sout{is developed and adapted toward a heatmap}, see Section~\ref{subsec:austen-plots}.

\begin{algorithm}[H]
\label{algo:totally-unobserved}
    \SetAlgorithmName{Procedure}{Procedure}{Totally unobserved covariate} 
    \SetKwInOut{Input}{input}
    \SetKwInOut{Init}{init}
    \SetKwInOut{Parameter}{param}
    \caption{A totally-unobserved covariate}
    \Init{\quad $\delta_{mis}:= [\dots]$\tcp*{Define \modifbis{a} range for plausible $\delta_{mis}$ values}}
     \Init{\quad $\Delta_{m}:=[\dots]$ \tcp*{Define \modifbis{a} range for plausible $\Delta_{m}$ values}}
     Compute all possible bias $-\delta_{mis} \Delta_{m}$ (see Lemma~\ref{lemma:sensitivity-bias-linear})
     \\
\Return{Sensitivity map}
\end{algorithm}

\smallskip



A partially-observed covariate could always be removed so that this sensitivity analysis could be conducted for every missing data patterns (the variable being missing in the RCT or in the observational data).\sout{But this simple method has drawbacks.} \sout{But the independence assumption ($X_{m i s} \independent X_{o
b s}$) needed here is a strong one, and} \modifbis{However} dropping a partially-observed
covariate \sout{seems inefficient as it discards available information}\modifbis{\textit{(i)} \sout{does not make the most of the available information} is inefficient as it discards available information,} \textit{(ii)} \modifbis{amounts to considering the variable as totally unobserved which, in turn, leads us to assume \sout{puts us in the situation of assuming} independence between observed and unobserved covariates, a very strong hypothesis}.
Therefore, in the following subsections, we propose methods that use the partially-observed covariate -- when available -- to improve the bias estimation\sout{ and to remove this independence assumption}.

\subsubsection{Sensitivity analysis when a key covariate is partially \sout{un}observed }\label{subsec:sensitivity-partially-obs}

When partially available, we propose to use $X_{m i s}$ to have a better
estimate of the bias. Unlike the above, this approach does not need 
the partially observed covariate to be independent of all other
covariates, but rather captures the dependencies from the data.


\paragraph{Observed in observational study \label{subsec:observed-obs}} 
Suppose one key covariate $X_{m i s}$ is observed in the observational study, but not in the RCT. 
Under Assumption~\ref{a:trans-sigma}, the asymptotic bias of any $L^1$-consistent estimator $\hat \tau_{n,m, o b s}$
is derived in Theorem~\ref{lemma:linear-gformula-unbiased}.
The quantitative bias is informative as it depends only on the regression coefficients $\delta$, and on the shift in expectation between covariates. \modifbis{Indeed, the bias term can be decomposed as follows:} \sout{Furthermore, parts of the bias term can be estimated from the data, helping to reveal sensitivity parameters:}
\begin{equation*}
\label{eq:sensitivity-when-observed-in-obs}
         B =  - \underbrace{\delta_{m i s}}_{\text{ $X_{m i s}$'s strength}}  \left( \underbrace{\EE[X_{m i s}] - \EE[X_{m i s} \mid S=1]}_{\text{Shift of $X_{m i s}$: } \Delta_m} - \underbrace{\Sigma_{m i s, o b s}\Sigma_{o b s, o b s}^{-1}(\EE[X_{o b s}]-\EE[X_{o b s}\mid S = 1])}_{\text{Can be estimated from the data}} \right).
\end{equation*}

\modifbis{Using the observational study where the necessary covariates are all observed, one can estimate the covariance term $\Sigma_{m i s, o b s}\Sigma_{o b s, o b s}^{-1}$ together with the shift for the observed set of covariates}. \modifbis{Unfortunately, }the remaining parameters $\delta_{m i s}$, corresponding to \modifbis{the coefficient of} the missing covariates in the complete linear model, and $\Delta_m = \mathbb{E}[X_{m i s}] - \mathbb{E}[X_{m i s} \mid S = 1]$ are not identifiable from the observed data. These two parameters correspond respectively to the strength of the treatment effect modifier and the distributional shift \sout{importance} of the missing covariate.
\sout{They can be turned into an estimate of the bias using this explicit formulation.}\modifbis{These two quantities are used as sensitivity parameters to estimate a plausible range of the bias (see Procedure~\ref{algo:observed-obs}).} 
\sout{In practice, we can take this bias into account to propose a refined sensitivity analysis, described in Procedure~\ref{algo:observed-obs}.} 
\sout{In practice, when facing the case of a variable observed only in the observational study, the procedure to follow to perform a sensitivity analysis is 
summarized in procedure~\ref{algo:observed-obs}.}
Simulations illustrate how these sensitivity parameters can be used, along with graphical visualization derived from \sout{Austen plots}\modif{sensitivity maps} (see Section~\ref{sec:simulation}).

\begin{algorithm}[H]
\label{algo:observed-obs}
    \SetAlgorithmName{Procedure}{Procedure}{Observed in observational}
    \SetKwInOut{Input}{input}
    \SetKwInOut{Init}{init}
    \SetKwInOut{Parameter}{param}
    \caption{Observed in observational}
    \Init{\quad $\delta_{mis}:= [\dots]$\tcp*{Define \modifbis{a} range for plausible $\delta_{mis}$ values}}
     \Init{\quad $\Delta_{mis}:=[\dots]$ \tcp*{Define \modifbis{a} range for plausible $\Delta_{mis}$ values}}
     Estimate $\Sigma_{o b s, o b s}$, $\Sigma_{m i s, o b s}$, and $\mathbb{E}[X_{o b s}]$ on the observational dataset;\\
     Estimate $\mathbb{E}[X_{o b s} \mid S = 1]$ on the RCT dataset; \\
     Compute all possible bias\modifbis{es} for \modifbis{the predefined} range\modifbis{s} of $\delta_{mis}$ and $\Delta_{mis}$, according to Theorem~\ref{lemma:linear-gformula-unbiased}.\\
\Return{Sentivity map}
\end{algorithm} 

\paragraph{Data-driven approach to determine sensitivity parameter}\label{para:remark-on-sensitiviy-delta}
Note that \modifbis{guessing a good range} \sout{giving a range} for the shift $\Delta_{mis}$ is probably easier\sout{,} than giving a range \sout{of} \modifbis{for} the coefficients $\delta_{mis}$. 
We propose a \sout{empirical} \modifbis{data-driven} method to \sout{have a practical data-driven estimation of} \modifbis{estimate} $\delta_{mis}$. First, learn a linear model of $X_{mis}$ from observed covariates $X_{obs}$ on the observational data, then impute the missing covariate in the trial, and finally obtain $\hat \delta_{mis}$ with a Robinson procedure on the imputed trial data \citep{robinson1988semiparam, wager2020courses, nie2020quasioracle}. \modifbis{ The Robinson procedure is recalled in Appendix (see Section~\ref{appendix:robinson})} \sout{This can give an idea of a possible value for $\delta_{mis}$.} This method is used in the semi-synthetic simulation (see Section~\ref{sec:STAR}). 

\paragraph{Observed in the RCT \label{sec:nguyen-sensitivity-method}} 
\modifbis{The} method \modifbis{we propose here was} already developed by \cite{nguyen2017sensitivity, nguyen2018sensitivitybis}, and we briefly recall its principle in this part. Note that we extend this method by considering a semi-parametric model \eqref{eq:linear-causal-model}, while they considered a completely linear model.
For this missing covariate pattern, only one sensitivity parameter is necessary. 
As the RCT is the complete data set, the regression coefficients $\delta$ of \eqref{eq:linear-causal-model} 
can be estimated for all the key covariates, leading to an estimate $\hat \delta_{m i s}$ for the partially unobserved covariate. \cite{nguyen2017sensitivity, nguyen2018sensitivitybis} showed that:

\begin{equation}
\label{eq:ngyuen-sensitivity-method}
  \tau=\left\langle\delta_{o b s}, \mathbb{E}\left[X_{o b
s}\right]\right\rangle \,+\, \langle\delta_{m i s},\;
\underbrace{\mathbb{E}\left[X_{m i s}\right]}_{\text {Unknown }}\rangle . 
\end{equation}

In this case, and as the influence of $X_{m i s}$ as a treatment effect modifier can be estimated from the data t\modifbis{h}rough $\hat \delta_{m i s}$, only one sensitivity parameter is needed\modifbis{, namely $\mathbb{E}[X_{m i s}]$. Therefore, we assume to be given a range of plausible values for $\mathbb{E}[X_{m i s}]$, for example according to a domain expert prior.}
\sout{Here, it consists of a range of plausible $\mathbb{E}[X_{m i s}]$, for example according to a domain expert prior, which is also interpretable.}

Note that $\delta_{m i s}$ can be estimated
following a Robinson procedure.\sout{ \citep{robinson1988semiparam, wager2020courses, nie2020quasioracle}, in particular in the case of a semi-parametric model with a linear CATE as in \eqref{eq:linear-causal-model}. The Robinson procedure, also called R-learner, is recalled in Appendix (see Section~\ref{appendix:robinson}).} This allows \sout{to}extend\modifbis{ing} \citep{nguyen2018sensitivitybis}'s work to the semi-parametric case. \modifbis{Softening even more the parametric assumption where only $X_{mis}$ is additive in the CATE is a natural extension, but out of the scope of the present work.}
\sout{Note that in this missing covariate pattern further extensions are possible, in particular softening the parametric assumption where only $X_{mis}$ is additive in the CATE.} 

\begin{algorithm}[H]
\label{algo:observed-rct}
    \SetAlgorithmName{Procedure}{Procedure}{Observed in RCT \citep{nguyen2017sensitivity}}
    \SetKwInOut{Input}{input}
    \SetKwInOut{Init}{init}
    \SetKwInOut{Parameter}{param}
    \caption{Observed in RCT}
     \Init{\quad $\mathbb{E}[X_{m i s}]:=[\dots]$ \tcp*{Define \modifbis{a} range for plausible  values of $\mathbb{E}[X_{m i s}]$}}
     Estimate $\delta$ with the Robinson procedure, that is: \\
     Run a non-parametric regression $Y \sim X$ on the RCT, and denote
$\hat m_n(x) = \mathbb{E}[Y\mid X = x, S = 1]$
 the obtained estimator; \\
     Define the transformed features $\tilde Y = Y - \hat m_n(X)$ and $\tilde Z = (A - e_1(X)) X$. \\
     Estimate $\hat \delta$ run\sout{n}ning the OLS regression on $\tilde Y \sim \tilde Z$; \\
     Estimate $\mathbb{E}[X_{o b s}]$ on the observational dataset; \\
     Compute all possible bias\modifbis{es} for \modifbis{the} range of $\mathbb{E}[X_{m i s}]$ according to \eqref{eq:ngyuen-sensitivity-method}.\\
\Return{Sensitivity map}
\end{algorithm}

\subsubsection{Vizualization: sensitivity maps \label{subsec:austen-plots}}
\modifbis{From now} on, each of the sensitivity method suppose to translate
sensitivity parameter(s) and to compute the range of bias associated. A
last step is to communicate or visualize the range of bias\modifbis{es}, which is 
slightly more complicated when there are two sensitivity parameters.
\modif{Sensitivity map} is a way to aid such judgement
\citep{imbens2003sensitivity,veitch2020sense}. It consists in having a two-dimensional plot,
each of the axis representing the sensitivity parameter, and the solid
curve is the set of sensitivity parameters that leads to an estimate that
induces a certain bias' threshold. Here, we adapt this method
to our settings with several changes.
Because coefficients interpretation is hard, a typical practice is to translate a regression coefficient into a partial $R^2$. For example, \cite{imbens2003sensitivity} prototypical example proposes to interpret the two parameters with partial $R^2$. In our case, a close quantity can be used:
\begin{equation}
R^2 \sim \frac{\mathbb{V}[\delta_{m i s}\, X_{m i s}]}{\mathbb{V}[\sum_{j \in obs} \hat \delta_j\, X_j]}
\end{equation}
where the denominator term is obtained when regressing $Y$ on $X_{o b
s}$. If this $R^2$ coefficient is close to 1, then the missing covariate
has a similar influence on $Y$ compared to other covariates. On the
contrary, if $R^2$ is close to 0, then the impact of $X_{m i s}$ on $Y$
as a treatment effect modifier is small compared to other covariates. But
in our case one of the sensitivity parameter is really palpable as it is
the covariate shift $\Delta_{m}$. 
We advocate keeping the regression coefficient and shift as sensitivity parameter rather than a $R^2$ to help practitioners as it allows to keep the sign of the bias, than can be in favor of the treatment or not and help interpreting the sensitivity analysis. Furthermore, even if postulating an hypothetical value of a coefficient is tricky, when the covariate is partially observed \modifbis{an imputation procedure can be proposed} to have a grasp of the coefficient true value.

\begin{figure}[!t]
    \centering
    \includegraphics[width=0.45\textwidth]{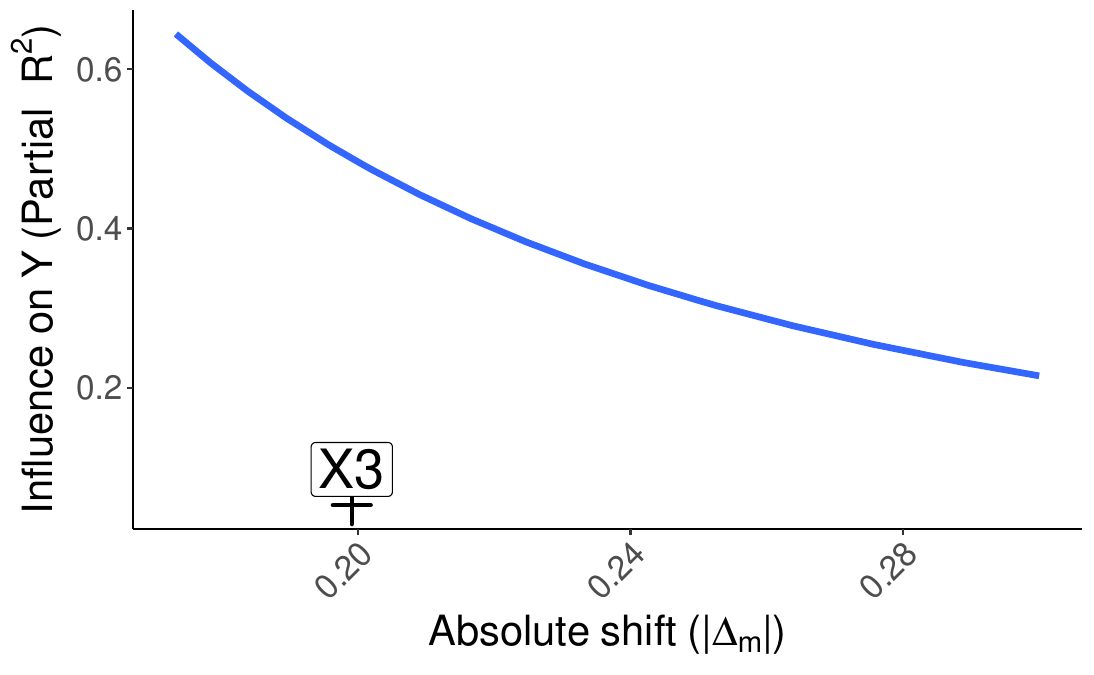} \hspace{1cm}
    \includegraphics[width=0.45\textwidth]{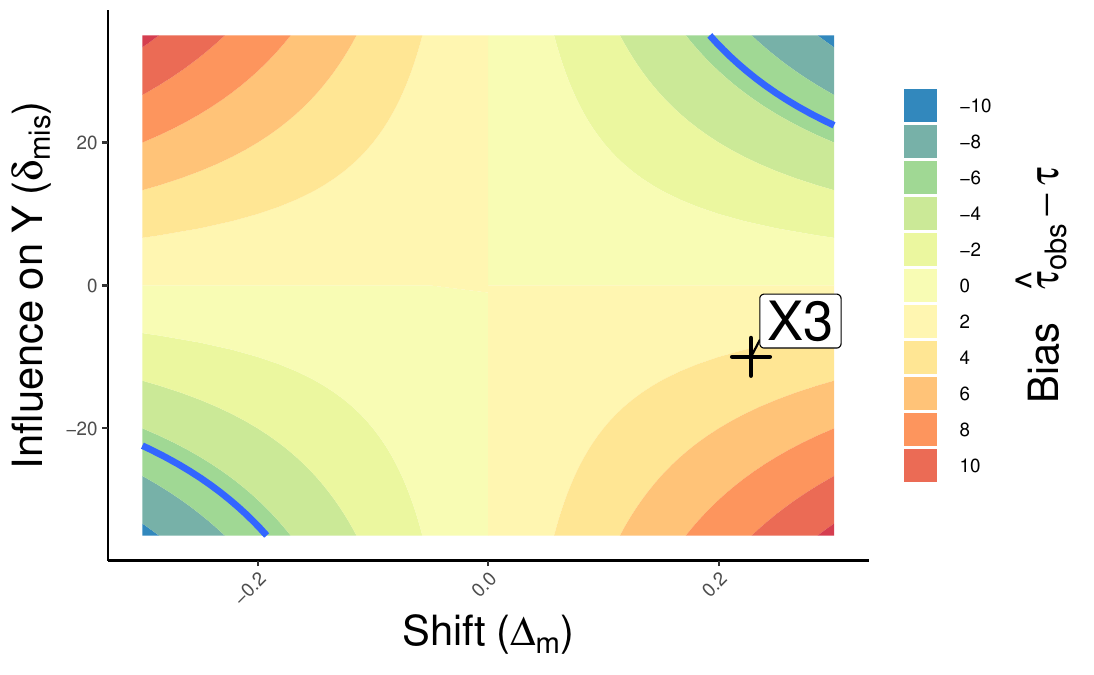}
	\caption{\textbf{\sout{Austen plots}\modif{Sensitivity maps}}: On this figure $X_3$ is supposed to be a missing covariate.
	(Left) Regular \sout{Austen plot}\modif{sensitivity map} showing how strong an key covariate would need to be to induce a bias of $\sim 6$ in function of the two sensitivity parameters $\Delta_m$ and partial $R^2$ when a covariate is totally unobserved. 
	(Right) The exact same simulation data are represented, while using rather $\delta_{m i s}$ than the partial $R^2$, and superimposing the heatmap of the bias which allows to reveal the general landscape along with the sign of the bias.} 
	\label{fig:totally-missing-linear}
\end{figure}

On Figure~\ref{fig:totally-missing-linear} we present a glimpse of the simulation result, to introduce the principle of the \sout{Austen plot}\modif{sensitivity map}, with on the left the representation using $R^2$ and on the right a representation keeping the raw sensitivity parameters.
In this plot, we consider the covariate $X_3$ to be missing, so that we represent \textit{what would be the bias if we missed $X_3$?},
The associated sensitivity parameters are represented on each axis.
In other word, the \sout{Austen plot}\modif{sensitivity map} shows how strong an unobserved key covariate would need to be to induce a bias that would force to reconsider the conclusion of the study because the bias is above a certain threshold, that is represented by the blue line. For example in our simulation set-up, $X_3$ is below the threshold as illustrated on Figure~\ref{fig:totally-missing-linear}. The threshold can be proposed by expert, and here we proposed the absolute difference between $\hat \tau_{n,m,obs}$ and the RCT estimate $\hat \tau_1$ as a natural quantity.
In particular, we observe that keeping the sign of the sensitivity parameter allows to be even more confident on the direction of the bias.

\subsubsection{Partially observed covariates: imputation \label{subsec:imputation}}
Another practically appealing solution is to impute the partially-observed covariate, based on the complete data set (whether it is the RCT or the observational one) following Procedure~\ref{algo:imputation}. We analyse theoretically in this section the bias of such procedure in Corollary~\ref{lem:imputation}, and show there is no gain in linearly imputing the partially-observed covariate.\sout{\modif{, that can be understood as a direct consequence of Assumption~\ref{a:trans-sigma} and Theorem~\ref{lemma:linear-gformula-unbiased}}.}

To ease the mathematical analysis, we focus on a G-formula estimator
based on oracles quantities: the best imputation function and the surface
responses are assumed to be known. While these are not available
in practice, they can be approached with
consistent estimates of the imputation functions and the surface
responses. 
The precise formulation\modifbis{s} of our oracle estimates are given in Definition~\ref{def:imputation-procedure-missing-obs} and Definition~\ref{def:imputation-procedure-missing-rct}.

\begin{definition}[Oracle estimator when covariate is missing in the observational data set]\label{def:imputation-procedure-missing-obs} 
Assume that the RCT is complete and that the observational sample contains one missing covariate $X_{mis}$. We assume that we know
\begin{itemize}
    \item[(I)] the true response surface\modifbis{s} $\mu_1$ and $\mu_0$
    
    \item[(II)] the true linear relation between $X_{mis}$ as a function of $X_{obs}$\sout{ and}. 
\end{itemize}
Our oracle estimate $\hat \tau_{\text{\tiny G}, \infty, m, imp}$ consists in applying the G-formula with the true response surfaces $\mu_1$ and $\mu_0$ (I) on the observational sample, in which the missing covariate has been imputed by the best (linear) function (II). 
\end{definition}


\begin{definition}[Oracle estimator when covariate is missing in the RCT data set]\label{def:imputation-procedure-missing-rct} 
Assume that the observational sample is complete and that the RCT contains one missing covariate $X_{mis}$. We assume that we know
\begin{itemize}
    \item[(I)] the true linear relation between $X_{mis}$ as a function of $X_{obs}$, which leads to the optimal imputation $\hat{X}_{mis}$\sout{ and},
    
    \item[(II)] the \sout{true response surfaces}\modifbis{conditional expectations}, $\EE[Y(a) | X_{obs}, \hat{X}_{mis}, S=1]$, for $a \in \{0,1\}$. 
\end{itemize}
Our oracle estimate $\hat \tau_{G, \infty, \infty, imp}$ consists in optimally imputing the missing variable $X_{mis}$ in the RCT (I). Then, the G-formula is applied to the observational sample, with the surface responses that have been perfectly fitted on the completed RCT sample. 
\end{definition}

\begin{corollary}[Oracle bias of imputation in a Gaussian setting]\label{lem:imputation}
Assume that the CATE is linear  \eqref{eq:linear-causal-model} and that Assumption~\ref{a:trans-sigma} holds. Let B be the following quantity:
 \begin{equation*}
        B = \delta_{mis}  \left( \EE[X_{mis}] - \EE[X_{mis} \mid S=1] - \Sigma_{j, o b s}\Sigma_{o b s, o b s}^{-1}(\EE[X_{o b s}]-\EE[X_{o b s}\mid S = 1]) \right).
\end{equation*}
\begin{itemize}
    \item \textbf{Complete RCT.} Assume that the RCT is complete and that the observational data set contains a missing covariate $X_{mis}$. Consider the oracle estimator $\hat{\tau}_{\text{\tiny G}, \infty, m, imp}$ \sout{defined} in Definition~\ref{def:imputation-procedure-missing-obs}. Then, 
    \begin{align*}
    \tau - \lim\limits_{m \to \infty}  \E[\hat \tau_{\text{\tiny G}, \infty, m, imp}] = B
    \end{align*}
    
    \item \textbf{Complete Observational.} Assume that the observational data set is complete and that the RCT contains a missing covariate $X_{mis}$. Consider the oracle estimator $\hat{\tau}_{\text{\tiny G}, \infty, \infty, imp}$ \sout{defined} in Definition~\ref{def:imputation-procedure-missing-rct}. Then, 
    \begin{align*}
    \tau - \E[\hat{\tau}_{\text{\tiny G}, \infty, \infty, imp}] = B
    \end{align*}
\end{itemize}
\end{corollary}

Derivations are detailed in appendix (see Subsection~\ref{proof:imputation-no-bias}). 
Corollary~\ref{lem:imputation}\sout{ proves that, when the data are Gaussian,
learning the imputation on one data set and applying it on the other one
leads to a biased estimate, where the asymptotic value of the bias is quantified} \modif{\sout{provides} \modifbis{highlights} that there is no gain in linearly imputing the missing covariate compared to dropping it.}
Simulations (Section~\ref{appendix:synthetic-simulation-extension}) show that the average bias
of a finite-sample imputation procedure is similar to the bias of $\hat \tau_{\text{\tiny G},\infty,\infty,obs}$.

\begin{algorithm}[H]
\label{algo:imputation}
    \SetAlgorithmName{Procedure}{Procedure}{Linear imputation}%
    \SetKwInOut{Input}{input}
    \SetKwInOut{Init}{init}
    \SetKwInOut{Parameter}{param}
    \caption{Linear imputation}
    Model $X_{m i s}$ a linear combination of $X_{o b s}$ on the complete data set;\\
    Impute the missing covariate with $\hat X_{m i s}$ with the previous fitted model;\\
    Compute $\hat \tau$ with the G-formula using the imputed data set $X_{o b s} \cup \hat X_{m i s}$;\\
\Return{$\hat \tau$}
\end{algorithm}
 
\subsubsection{Using a proxy variable in place of the missing covariate \label{sec:proxy-variable}}
Another solution is to use a so-called proxy variable.\sout{, as illustrated by
the structural diagram in Figure~\ref{fig:dag_proxy}}\sout{ Such a variable is a variable associated with the missing one, but not a treatment effect modifier. 
Could a proxy
covariate, available in both data set, overcome this issue? }
The impact of a proxy in the case of a linear model is documented in
econometrics \citep{chen2005auxiliarydata, Chen2007measureerror,
angrist2008mostlyharmless, wooldridge2016introductory}. 
An example of
a proxy variable is the height of children as a proxy for their age. Note
that in this case, even if the age is present in one of the two datasets,
\sout{then} only the children's height is kept in for this method.

Here, we propose a framework to handle a missing key covariate 
with a proxy variable and estimate the bias reduction
accounting for the additional noise brought by the proxy.

\begin{assumption}[Proxy framework]\label{a:proxy} Assume that $X_{m i s} \independent X_{o b s}$, and that there exist\modifbis{s} a proxy variable $X_{p r o x}$ such that,
$$X_{p r o x} =  X_{m i s} + \eta $$
where $\mathbb{E}[\eta] = 0$, $\operatorname{Var}[\eta] = \sigma_{p r o x}^2$, and $\operatorname{Cov}\left(\eta, X_{m i s}\right)=0$. 
In addition we suppose that $\operatorname{Var}[X_{m i s}] = \operatorname{Var}[X_{m i s} \mid S =1 ] = \sigma_{m i s}^2$.
\end{assumption}

\begin{definition}
Let $ \hat \tau_{\text{\tiny G}, n, m, p r o x}$ be the G-formula estimator where $X_{m i s}$ is substituted by $X_{p r o x}$ as detailed in assumption~\ref{a:proxy}.
\end{definition}

\begin{lemma}
\label{lemma:bias-proxy} Assume that the generative linear model \eqref{eq:linear-causal-model} holds, along with Assumption~\ref{a:trans-sigma} and the proxy framework \eqref{a:proxy}. Then the asymptotic bias of $\hat \tau_{\text{\tiny G}, n, m, p r o x}$ is:
$$ \lim\limits_{n,m \to \infty}\mathbb{E}[\hat \tau_{\text{\tiny G}, n, m, p r o x}] -\tau = - \delta_{m i s}\, \Delta_m \left(1- \frac{\sigma_{m i s}^2}{\sigma_{m i s}^2 + \sigma_{p r o x}^2}\right) $$
\end{lemma}

We denote $\hat \delta_{p r o x}$\sout{,} the estimated coefficient for $X_{p r o x}$\modifbis{. Such an estimation can} be obtained using a Robinson procedure when regressing $Y$ on the set $X_{o b s} \cup X_{p r o x}$. 
\begin{corollary}
\label{corol:bias-proxy}
The asymptotic bias in lemma~\ref{lemma:bias-proxy} can be estimated using the following expression:
$$\lim\limits_{n,m \to \infty}\mathbb{E}[\hat \tau_{\text{\tiny G}, n, m, p r o x}]-\tau = -\hat \delta_{p r o x} \bigl(\mathbb{E}[X_{p r o x}] - \mathbb{E}[X_{p r o x} \mid S = 1]\bigr) \frac{\sigma_{p r o x}^2}{\sigma_{m i s}^2} $$
\end{corollary}

Proof\modifbis{s} of Lemma~\ref{lemma:bias-proxy} and Corollary~\ref{corol:bias-proxy} \sout{is}\modifbis{are} detailed in Appendix (Proof~\ref{proof:bias-proxy}).
Note that, as expected, the average bias reduction strongly depends on
the quality of the proxy. 
In the limit case, if $\sigma_{p r o x} \sim 0$ so that the correlation between the proxy and the missing variable is one, then the bias is null.
In general, if $\sigma_{p r o x} \gg \sigma_{m i s}$ then the proxy variable does not diminish the bias. 
Finally, \modifbis{we propose a practical approach}\sout{the practical approach is detailed} in Procedure~\ref{algo:proxy}. Note that it requires to have a range of possible $\sigma_{p r o x}$ values. We recommend to use the data set on which the proxy along with the partially-unobserved covariate are present, and to obtain an estimation of this quantity on this subset. 

\begin{algorithm}[H]
\label{algo:proxy}
    \SetAlgorithmName{Procedure}{Procedure}{Proxy variable}
    \SetKwInOut{Input}{input}
    \SetKwInOut{Init}{init}
    \SetKwInOut{Parameter}{param}
    \caption{Proxy variable}
     \Init{\quad $\sigma_{p r o x}:=[\dots]$ \tcp*{Define \modifbis{a} range for plausible $\sigma_{p r o x}$ values}}
     \If{$X_{m i s}$ is in RCT}{
      \Init{\quad $\Delta_{m i s}:=[\dots]$ \tcp*{Define \modifbis{a} range for plausible $\Delta_{m i s}$ values}}
      Estimate $\delta_{m i s}$ with the Robinson
procedure (see Procedure~\ref{algo:observed-rct} for details);\\
      Compute all possible bias\modifbis{es} for \modifbis{the} range of $\sigma_{p r o x}$ according to Lemma~\ref{proof:bias-proxy}.\\
      }\Else{
      Estimate $\delta_{p r o x}$ with the Robinson procedure (see Procedure~\ref{algo:observed-rct} for details);\\
      Estimate $\mathbb{E}[X_{p r o x}]$ and $\mathbb{E}[X_{p r o x} \mid S = 1]$;\\
      Compute all possible bias for range of $\sigma_{p r o x}$ according to Corollary~\ref{corol:bias-proxy}.\\
                           }  
\Return{Bias\modifbis{es}'s range}
\end{algorithm}
\sout{In the above, we extended the work of \cite{nguyen2017sensitivity,
nguyen2018sensitivitybis, Andrews2018ValidityBias}, by studying (i)
a method for each missing covariate pattern, (ii) a strategy\sout{that
proposes} to impute the partially-observed covariate, (iii)
replacing a missing variable with a noisy proxy.}


\section{Synthetic and semi-synthetic simulations \label{sec:simulation}}

\modifbis{More information on simulation settings can be found in Appendix see Section~\ref{appendix:synthetic-simulation-extension}} \sout{In particular the synthetic simulations results are summarized, while a detailed presentation of the method along with comments is available in appendix for the interested reader (see Section~\ref{appendix:synthetic-simulation-extension}).}

\subsection{Synthetic simulations}

\begin{wrapfigure}{r}{0.45\textwidth}
  \begin{center}
    \includegraphics[width=0.42\textwidth]{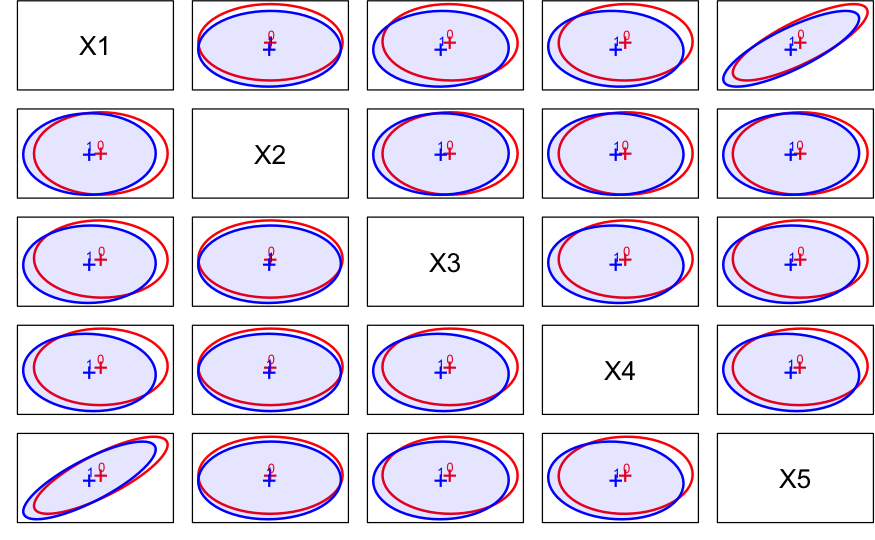}
  \end{center}
  \caption{\textbf{Variance-covariance} preservation in the simulation set-up highlighted with pairwise covariance ellipses for one realization of the simulation (package \texttt{heplots}).}
  \label{fig:sigma-simulation}
\end{wrapfigure}

While results \modifbis{presented in} \sout{from} Section~\ref{sec:linear-causal-model} apply to any function $g$ \modifbis{(see  \eqref{eq:linear-causal-model})}, \sout{here we illustrate the case of a linear function} \modifbis{we choose $g$ as a linear function to illustrate our findings}. All \sout{the}simulations are available on \href{https://github.com/BenedicteColnet/unobserved-covariate}{github}\footnote{\texttt{BenedicteColnet/unobserved-covariate}}, and include non-linear forms for $g$\sout{ where results hold}.

\paragraph{Simulations parameters}

We use a similar simulation framework as in \cite{dong2020integrative} and \cite{colnet2020causal},
where $5$ covariates are generated \sout{--} independently, except for $X_1$ and $X_5$ \sout{which} \modifbis{whose} correlation is set at 0.8, except when explicitly mentioned\modifbis{. We simulate marginals as} \sout{ -- with} $X_{j} \sim \mathcal{N}(1,1)$ for \sout{each} \modifbis{all} $j=1, \ldots, 5$. 
The trial selection process is defined using a logistic regression model, such that:  
\begin{equation}
\label{eq:rctmodel}
    \operatorname{logit}\left\{P(S=1\mid X)\right\}=\beta_{s,0}+\beta_{s,1}\,X_{1} 
    + \dots + \beta_{s,5}\,X_{5} .
\end{equation}
This selection process implies that the variance-covariance matrix in the RCT sample and in the target population \textit{may be different} \modif{depending on the (absolute) value of the coefficients $\beta_s$.}\sout{, so that Assumption~\ref{a:trans-sigma} does not hold all along this
simulation scheme.}\sout{This is interesting to illustrate the
resilience to small violations on Assumption~\ref{a:trans-sigma}} \modif{In \sout{the standard} \modifbis{our} simulation set-up, the overall variance-covariance structure is kept identical as visualized on Figure~\ref{fig:sigma-simulation}.} The outcome is generated according to a linear model, following \modifbis{Model~\ref{eq:linear-causal-model}, that is} \sout{
\eqref{eq:linear-causal-model}:}

\begin{equation}
\label{eq:Ymodel}
    Y(a) = \beta_0 + \beta_1 X_1 + \dots + \beta_5 X_5 + a(\delta_1 X_1 + \dots + \delta_5 X_5)+\varepsilon \mbox{ with } \varepsilon \sim \mathcal{N}(0,1).
\end{equation}

In this simulation, \modifbis{we set} $\beta = (5,5,5,5,5)$, and other parameters \modifbis{as described in} \sout{are further detailed in}  Table~\ref{tab:params-simulations-linear}.

\begin{table}[!h]
\begin{center}
\begin{tabular}{l|ccccc}
\hline
 \textbf{Covariates} &  $X_1$ & $X_2$ & $X_3$  & $X_4$ & $X_5$  \\ \hline 
 Treatment effect modifier & Yes & Yes & Yes & No & No  \\ 
 Linked to trial inclusion & Yes & No & Yes & Yes & No  \\ \hline
 $\delta$ & $\delta_1 = 30$ & $\delta_2 = 30$ & $\delta_3  = -10 $ & $\delta_4 = 0$ & $\delta_5 = 0$ \\
$\beta_s$ & $\beta_{s,1} = -0.4$ & $\beta_{s,2} = 0$ & $\beta_{s,3} = -0.3$ & $\beta_{s,4} = -0.3$ & $\beta_{s,5} = 0$ \\
 $. \independent X_1$ & -  & $X_2 \independent X_1$  & $X_3 \independent X_1$ & $X_4 \independent X_1$ & $X_5 \not\independent X_1$ \\ \hline
\end{tabular}
\caption{\textbf{Simulations parameters.}}
\label{tab:params-simulations-linear}
\end{center}
\end{table}

First a sample of size $10,000$ is drawn from the covariate distribution. 
From this sample, the selection model \eqref{eq:rctmodel} is applied which leads to an RCT sample of size $n \sim 2800$. Then, the treatment is generated according to a Bernoulli distribution with probability equal to $e_1 = 0.5$. 
Finally, the outcome is generated according to \eqref{eq:Ymodel}. The observational sample is obtained by drawing a new sample of size $m=10,000$ from the \modifbis{covariate distribution}\sout{ distribution of the covariates}. \modifbis{In this setting, the} \sout{
These parameters imply a target population} ATE \modifbis{equals} \sout{of} $\tau =  \sum_{j=1}^5 \delta_j \mathbb{E}[X_j] = \sum_{j=1}^5 \delta_j =50$.
\sout{Also,} \modifbis{Besides,} the sample selection ($S=1$) in \eqref{eq:rctmodel} \sout{with these parameters} is biased toward lower values of $X_1$ (and indirectly $X_5$), and higher values of $X_3$. This situation illustrate\modifbis{s} a case where $\tau_1 \neq \tau$. Empirically, we obtain\sout{ed} $\tau_1 \sim 44$.

\paragraph{Illustration of Theorem~\ref{lemma:linear-gformula-unbiased}}

Figure~\ref{fig:simulations-all-pattern} presents results of a simulation with $100$ repetitions with \sout{the full set of} \modifbis{no missing} covariates (on the Figure see \texttt{none}), and the impact of missing covariate(s) when using the G-formula or the IPSW to generalize. The theoretical bias from Theorem~\ref{lemma:linear-gformula-unbiased} \modifbis{is}\sout{are} also represented.

\begin{figure}[!h]
    \begin{minipage}{.27\linewidth}
	\caption{\textbf{Illustration of Theorem~\ref{lemma:linear-gformula-unbiased}}: Simulation results for the linear model with \sout{one, or more,} missing covariate(s) when generalizing the treatment effect using the G-formula (Definition~\ref{def:g-formula}) or IPSW (see Definition~\ref{def:ipsw} in appendix) estimator\modifbis{s} on the set of observed covariates. \sout{The}Missing covariate\sout{(s)} are indicated \sout{on the bottom of the plots} \modifbis{on the x-axis}.\sout{The continue blue line refers to $\tau_1$, the ATE on the RCT population, and the red dashed line to the ATE in the target population.} The theoretical bias (orange dot) is obtained from Theorem~\ref{lemma:linear-gformula-unbiased}. Simulations are repeated 100 times.}
	\label{fig:simulations-all-pattern}
    \end{minipage}%
    \hfill%
    \begin{minipage}{.72\linewidth}
    \includegraphics[width=\textwidth]{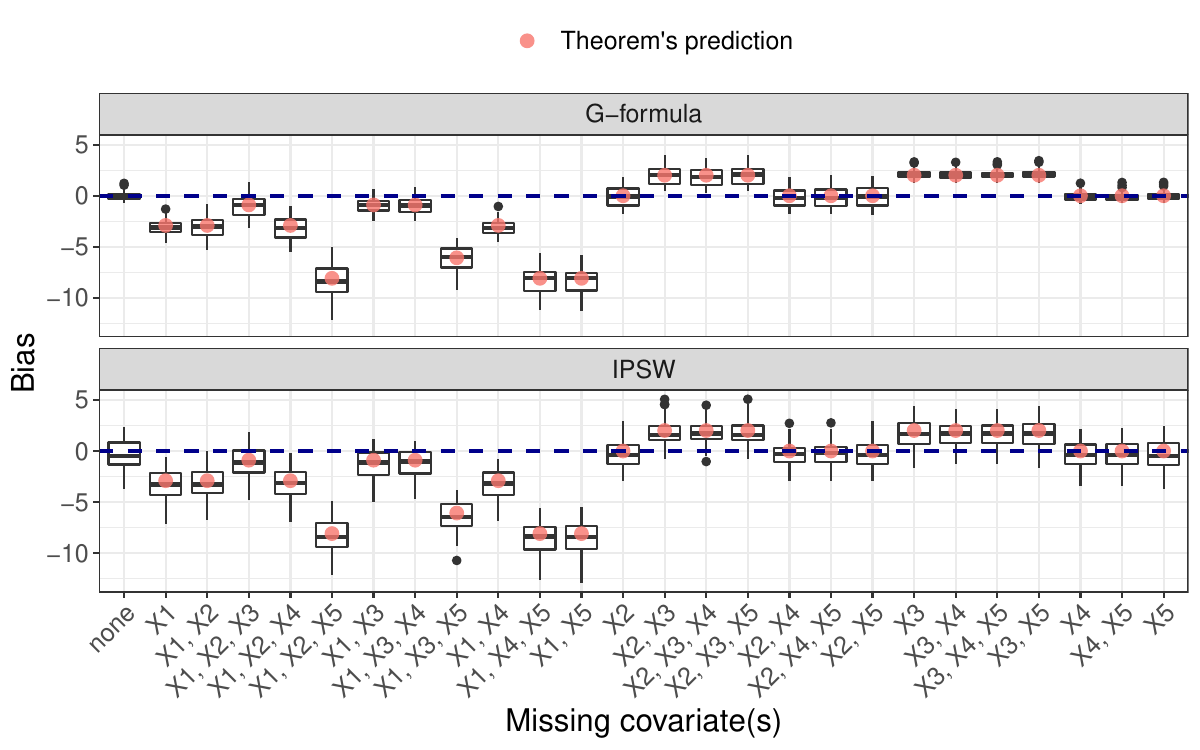}
    \end{minipage}
\end{figure}

\sout{Figure~\ref{fig:simulations-all-pattern} illustrates that Theorem~\ref{lemma:linear-gformula-unbiased} describes well the bias.}
The absence of covariates $X_2, X_4$ and/or $X_5$ does \sout{do} not affect \sout{the} ATE generalization because \modifbis{these covariates} \sout{they} are not simultaneously treatment effect modifiers and shifted (between the RCT sample and the target population). In addition, the sign\modifbis{s} of the bias\modifbis{es} depend\sout{s} on the sign\modifbis{s} of the coefficient\modifbis{s associated to the missing variables,} as highlighted by \modifbis{settings for which} $X_1$ and $X_3$ \modifbis{are missing}\sout{covariate}. 
\modifbis{As shown in Theorem~\ref{lemma:linear-gformula-unbiased}, variables acting on $Y$ without being treatment effect modifiers and linked to trial inclusion can help to reduce the bias, if correlated to a (partially-) unobserved key covariate. This is stressed out in our experiment by comparing the settings for which $X_1, X_5$ are missing and the one where only $X_1$ is missing.}
\sout{We also added the case where the two correlated variables $X_1$ and $X_5$ are missing, we observe that the bias is higher than when only $X_1$ is missing. This result is also predicted by \sout{the} Theorem~\ref{lemma:linear-gformula-unbiased} and illustrates a case where variables acting on $Y$ without being treatment effect modifiers and linked to trial inclusion can be of interest when a key covariate is unobserved or partially-unobserved, through correlation, to diminish the bias.}

\paragraph{\sout{Sensitivity analysis for} \modifbis{A} totally-unobserved covariate (from Section~\ref{subsec:totally-unobserved})}

To illustrate this case, the missing covariate has to be supposed independent of all the others. For this paragraph we consider $X_3$.
Then, according to Lemma~\ref{lemma:sensitivity-bias-linear}, the two sensitivity parameters $\delta_{m i s}$ and the shift $\Delta_m$ can be used to produce a sensitivity map for the bias on the transported ATE. 
The procedure~\ref{algo:totally-unobserved} summarizes the different steps, and the \sout{Austen plot}\modif{sensitivity map}'s output result was presented in Figure~\ref{fig:totally-missing-linear}.

\paragraph{\sout{Sensitivity analysis when \sout{missing} a covariate \modifbis{is missing} in the RCT} \modifbis{A missing covariate in the RCT} (from Section~\ref{subsec:observed-obs})}
In this case, \modifbis{we need to specify ranges of values for the two sensitivity parameters \modifbis{$\delta_{m i s}$ and $\Delta_m$}}. 
\modifbis{The experimental protocol is designed such that} \sout{Here we present a situation where} all covariates are successively partially missing in the RCT.
\sout{The sensitivity analysis is represented similarly as in Figure~\ref{fig:totally-missing-linear}.}
Because each missing variable implies a different landscape due to the dependence relation to other covariates (as stated in Theorem~\ref{lemma:linear-gformula-unbiased}), each variable requires a different heatmap (except if covariates are all independent). 
\modifbis{Results are depicted in Figure \ref{fig:sensitivity_missing_rct}}\sout{ presents 
results from \sout{the} \modifbis{this} simulation}. 
Figure~\ref{fig:sensitivity_missing_rct} illustrates the benefit of \modifbis{Protocol~\ref{algo:observed-obs}} 
accounting for other correlated covariates\modifbis{, and compared to a protocol assuming independent covariates}. 
\modifbis{Indeed, $X_1$ and $X_2$ are strong treatment effect modifiers (see Table~\ref{tab:params-simulations-linear}, where $\delta_1=\delta_2$), but $X_1$ is correlated to other completely observed covariates, which allows to "lower" the bias if $X_1$ is completely removed from the analysis compared to a similar covariate that would be independent of all other covariates. This is highlighted with a non-symetric bias landscape for $X_1$ on Figure~\ref{fig:sensitivity_missing_rct}. As a consequence, for a same value of $\delta_{mis}$ value, a guessed shift of $\Delta_{mis} = 0.25$ allows to conclude on a lower bias on the map for $X_1$, while it would not be the case for covariate $X_2$ (which is completely independent).}
\sout{Sensitivity parameters that invalidate the study are larger when
$X_1$ is missing than when $X_2$ is missing, as $X_2$ is independent of
the other covariates and $X_1$ is correlated to $X_5$ that is included in the study.}
\begin{figure}[!h]
    \centering
    \includegraphics[width=.75\textwidth]{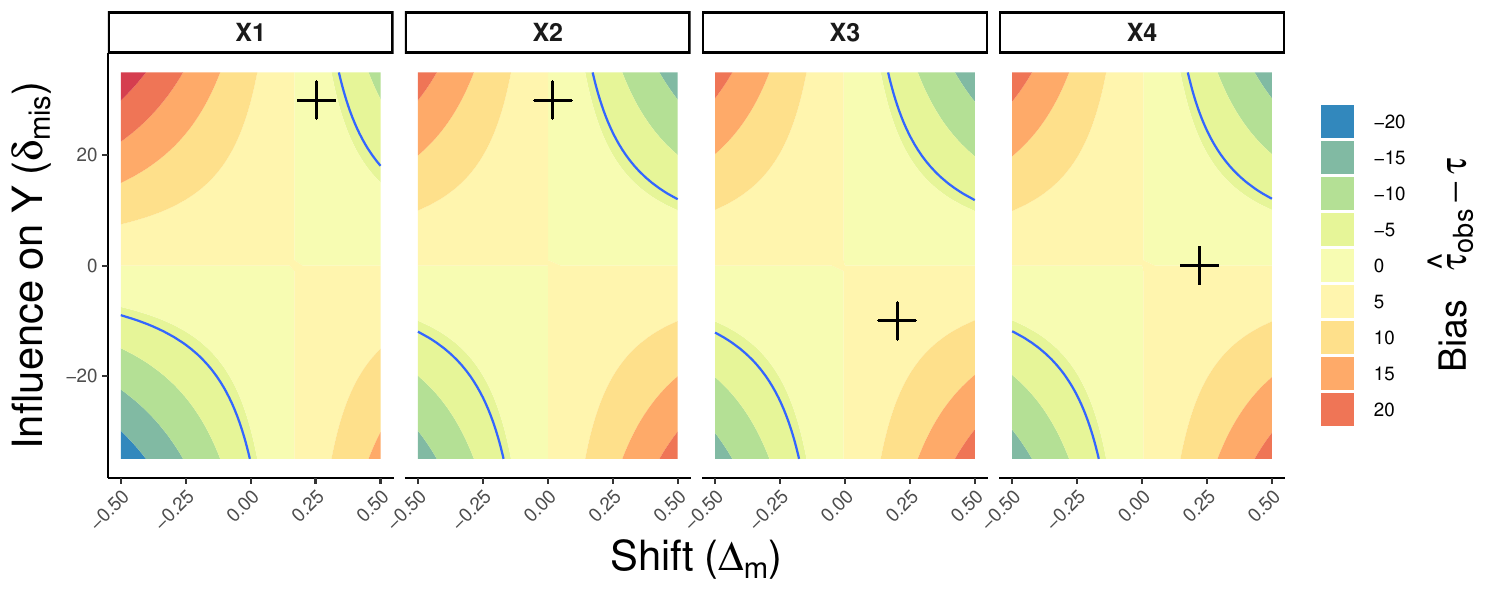}
    \caption{\textbf{Simulations results when applying
procedure~\ref{algo:observed-obs}}: Heatmaps with a blue curve showing
how strong an unobserved key covariate would need to be to induce a bias of
$\tau_1-\tau \sim -6$ in function of the two sensitivity parameters
$\Delta_m$ and $\delta_{m i s}$ when a covariate is totally unobserved.
Each heatmap illustrates a case where the covariate would be missing
(indicated on the top of the map), given all other covariates. The cross indicate the coordinate of true sensitivity parameters, in adequation with the bias empirically observed in Figure~\ref{fig:simulations-all-pattern}. The bias
landscape depends on the dependence of the covariate with other observed
covariates, as illustrated with an asymmetric heatmap when $X_1$ is
partially observed, due to the presence of $X_5$.}
    \label{fig:sensitivity_missing_rct}
\end{figure}

\paragraph{\sout{Sensitivity analysis when missing a covariate in the observational data} \modifbis{A missing covariate in the observational data} (from Section~\ref{sec:nguyen-sensitivity-method})}

\modifbis{In this case, we need to specify a range for the values of only one sensitivity parameter, namely} \sout{ only one sensitivity parameter is needed, being} $\mathbb{E}[X_{m i s}]$  \sout{as recalled in} \modifbis{(see \eqref{eq:ngyuen-sensitivity-method})}. \modifbis{In our experimental protocol, we assume that $X_1$ is missing and apply Procedure~\ref{algo:observed-rct}} \sout{This simulation is used, considering that $X_1$ is missing, and the procedure~\ref{algo:observed-rct} is applied}. Results are presented in Table~\ref{tab:simulation-nguyen}. 

\begin{table}[!h]
\begin{center}
\begin{tabular}{|l|l|l|l|l|l|}
\hline
Sensitivity parameter $\mathbb{E}[X_{m i s}]$ & 0.8 & 0.9 & 1.0 & 1.1 & 1.2 \\ \hline
Empirical average $\hat \tau_{\text{\tiny G}, n, m, o b s}$  & 44 & 47 & 50 & 53 & 56  \\ \hline
Empirical standard deviation $\hat \tau_{G, n, m, o b s}$ & 0.4 & 0.4 & 0.3 & 0.3 & 0.4 \\ \hline
\end{tabular}
\end{center}
\caption{\textbf{Simulations results when applying procedure~\ref{algo:observed-rct}}: Results of the simulation considering $X_1$ being partially observed in the RCT, and using the sensitivity method of \cite{nguyen2017sensitivity}, but with a Robinson procedure to handle semi-parametric generative functions. When varying the sensitivity parameters, the estimated ATE is close to the true ATE ($\tau = 50$) when the sensitivity parameter is closer to the true one ($\mathbb{E}[X_{m i s}] = 1$). The results are presented for 100 repetitions.}
\label{tab:simulation-nguyen}
\end{table}

\begin{figure}[!htb]
    \centering
    \begin{minipage}{.25\textwidth}
\begin{center}

\begin{tabular}{cc}
\hline
$\beta_{s,1}$ & Averaged p-value \\ \hline\hline
0     & 0.44                          \\
-0.2  & 0.37                          \\
-0.4  & 0.31                          \\
-0.6  & 0.14                          \\
-0.8  & 0.04                          \\
-1    & 0.012                         \\
-1.2  & 0.0001                        \\
-1.4  & $1\cdot10^{-9}$  \\
-1.6  & $1\cdot10^{-10}$ \\
-1.8  & $3\cdot10^{-15}$ \\
-2    & $1.4\cdot10^{-23}$ \\ \hline
\end{tabular}
    
\end{center}
\caption{\textbf{Empirical link between the logistic regression coefficient for sampling bias $\beta_{s,1}$ and the p-value of a Box-M test}. The average p-value is computed \sout{when}\modifbis{by} repeating $50$ times the simulation. We recall that in Figure~\ref{fig:simulations-all-pattern}, $\beta_{s,1}:=-0.4$.}
\label{tab:simulation-coeff-pvalue}
    \end{minipage}%
    \hspace{0.3cm}
    \begin{minipage}{0.70\textwidth}
        \centering
        \includegraphics[width=0.95\textwidth]{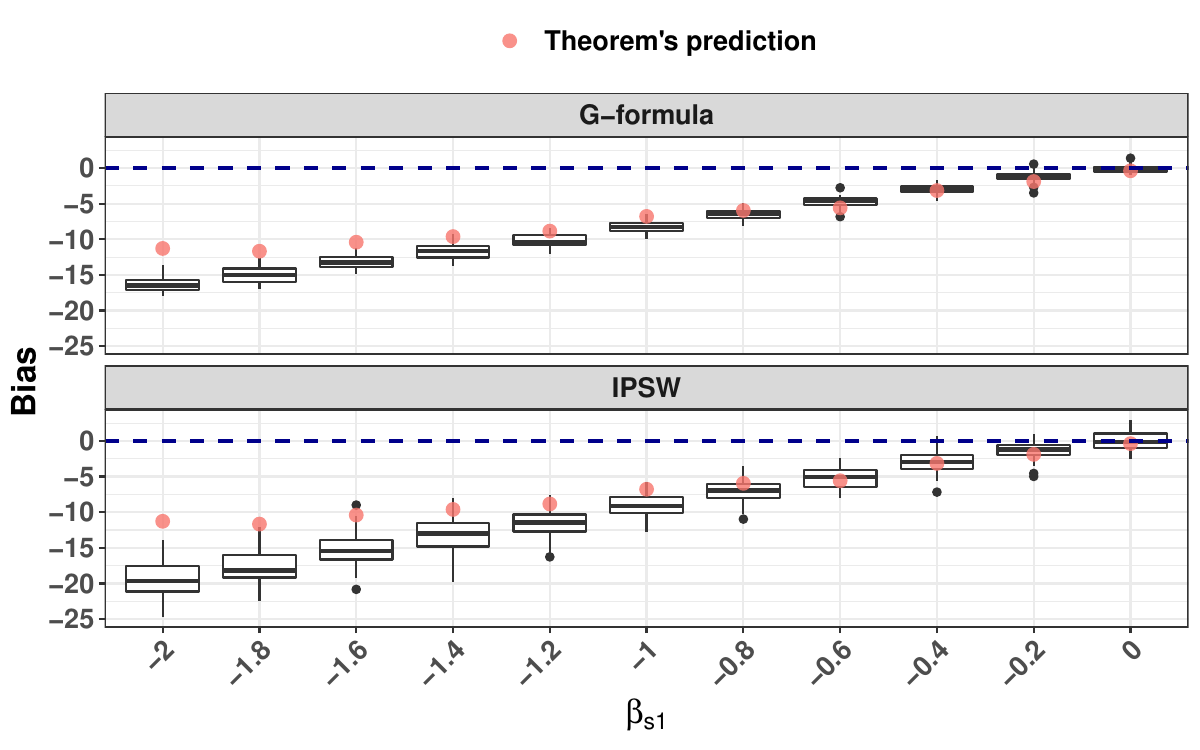}
        \caption{\textbf{Impact of poor\sout{-} transportability of the variance-covariance matrix} which is simulated with a decreasing coefficient $\beta_{s,1}$\modifbis{,} responsible of the covariate shift between the RCT sample and the observational sample. The lower $\beta_{s,1}$, the higher the absolute empirical bias (boxplots), and the higher the difference between the prediction\modifbis{s} given by Theorem~\ref{lemma:linear-gformula-unbiased} (orange dots) compared to the effective empirical bias\modifbis{es} (boxplots).}
        \label{fig:simulation-effect-of-strong-shift}
    \end{minipage}
\end{figure}

Simulations illustrating imputation (Corollary~\ref{lem:imputation}) and usage of a proxy (Lemma~\ref{lemma:bias-proxy}) are available in appendix, in Section~\ref{appendix:synthetic-simulation-extension}.

\paragraph{\modifbis{Violation of} \sout{Impact of a poorly verified} Assumption~\ref{a:trans-sigma}} To assess the impact of a lack of transportability of the variance\modifbis{-}covariance matrix (Assumption~\ref{a:trans-sigma}) we propose to observe the effect of an increasing (in absolute value) coefficient involved in the sampling process (Equation~\ref{eq:rctmodel}). We observe that the bigger the coefficient, \modifbis{the bigger the deviations from the theory, as expected.} 
\modifbis{To illustrate this phenomenon, we associate the logistic regression coefficient (the further away from the zero, the more Assumption~\ref{a:trans-sigma} is unvalidated) to the p-value of a Box-M test assessing if the variance covariance matrix from the two sources are different.} Empirically, the bias is still well estimated \modifbis{by procedures described in Section~\ref{sec:linear-causal-model}} even if the p-value is lower than $0.05$. Results are available on Figures~\ref{tab:simulation-coeff-pvalue} and \ref{fig:simulation-effect-of-strong-shift}.

\sout{Simulations illustrate Theorem~\ref{lemma:linear-gformula-unbiased} and the sensitivity analysis patterns. Additional simulations also illustrate the case of imputation and the proxy in appendix (see Section~\ref{appendix:synthetic-simulation-extension}). Interestingly the simulation set up can violate the Assumption~\ref{a:trans-sigma} if the coefficients in the logistic regression are big enough in the sampling model \eqref{eq:rctmodel}. This empirically supports that deviations from the framework of Assumption~\ref{a:trans-sigma} can be supported up to some extent.} 


\subsection{\modifbis{A} semi-synthetic simulation: the STAR experiment \label{sec:STAR}}
\modifbis{The semi-synthetic experiment is a mean to evaluate the methods on (semi) real data where neither the data generation
process nor the distribution of the covariates are under control.}

\subsubsection{Simulation details}

\sout{To illustrate the previous methods on data for which neither the data generation
process nor the distribution of the covariates are under our control.} 
We use the data from a randomized controlled trial, the Tennessee
Student/Teacher Achievement Ratio (STAR) study. This RCT is a pioneering
randomized study from the domain of education
\citep{angrist2008mostlyharmless}, started in 1985, and designed to
estimate the effects of smaller classes in primary school, on the \modifbis{children's grades}\sout{results}. This experiment showed a strong payoff to smaller classes
\citep{finn1990star}. In addition, the effect has been shown to be
heterogeneous \citep{krueger1999star}, where class size\modifbis{s} \sout{has}\modifbis{have} a larger
effect for minority students and those on subsidized lunch. For our purposes, we focus on the same subgroup of children, same treatment (small versus regular classes), and same outcome (average of all grades at the end) as in \cite{kallus2018removing}.

4\,218 \sout{students}\modifbis{children} are concerned by the treatment randomization,
with treatment assignment at first grade only. On the whole data, we
estimated an average treatment effect of 12.80 additional points on the
grades (95\% CI [10.41-15.2]) with the difference-in-means estimator. We consider this estimate as the ground truth $\tau$ as it is
the global RCT. Then, we generate a random sample of 500 children to
serve as the observational study. From the rest of the data, we sample a biased
RCT according to a logistic regression that defines
probability for each class to be selected in the RCT, and using only the
variable \texttt{g1surban} informing on the neighborhood of the school\modifbis{, which can be considered as a proxy for the socioeconomic status}.
The final selection is performed using a Bernoulli procedure, which
leads to 563 children in the RCT. The resulting RCT is such that $\hat \tau_1$ is
4.85 (95\% CI [-2.07-11.78]) which is underestimated. This is due to the
fact that that the selection is performed toward children that benefit
less from the class size reduction according to previous studies \citep{finn1990star, krueger1999star, kallus2018removing}. When generalizing the ATE with the G-formula on the full set of covariates, estimating the nuisance components with a linear model, and estimating the confidence intervals with a stratified bootstrap (1000 repetitions), the target population ATE is recovered with an estimate of 13.05 (95\% CI [5.07-22.11]) 
\sout{When n}\modifbis{N}ot including the covariate on which the selection is performed (\texttt{g1surban}) leads to a biased generalized ATE of 5.87 (95\% CI [-1.47-12.82]). 
These results are represented on Figure~\ref{fig:simulation-star-bilan}, along with AIPSW estimates. The IPSW is not represented due to a too large variance.
\begin{figure}[tbh!]
 \begin{minipage}{.48\linewidth}
  \caption{\textbf{Simulated STAR data}: True target population ATE estimation using all the STAR's RCT data is represented (difference-in-means). This is highlighted with a red dashed line to represent the ground truth. The ATE estimate of a biased RCT (difference-in-means) is also represented showing a lower treatment effect due to a covariate shift along the covariate \texttt{g1surban}. 
  Two estimators are used for the generalization, the G-formula (Definition~\ref{def:g-formula}) and the AIPSW (Definition~\ref{def:aipsw}); both relying on linear or logistic models for the nuisance components.
  The generalized ATE is either estimated with all covariates (\textcolor{cyan}{\textbf{blue}}) or with all covariates \emph{except \texttt{g1surban}} (\textcolor{orange}{\textbf{orange}}).  The confidence intervals are estimated with a stratified bootstrap (1000 repetitions).  Similar results are obtained when nuisance components are estimated with random forest.}
     \label{fig:simulation-star-bilan}
    \end{minipage}%
    \hfill%
    \begin{minipage}{.5\linewidth}
    \begin{center}
           \includegraphics[width=0.9\textwidth]{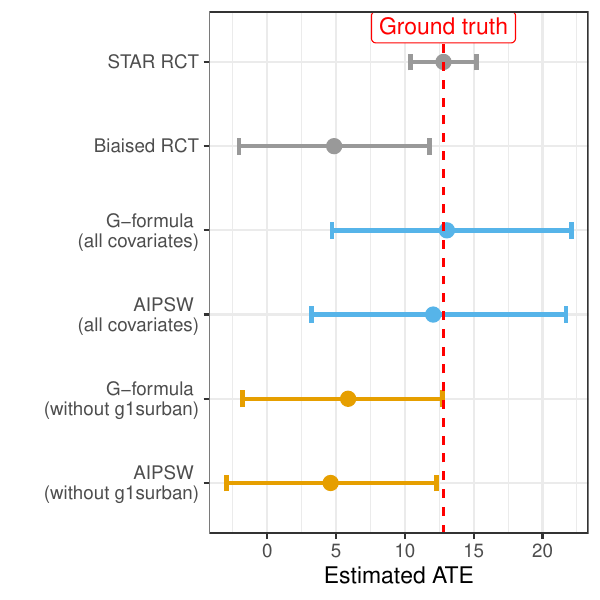}
     \end{center}
    \end{minipage}
\end{figure}

\subsubsection{Application of the sensitivity methods}

We now successively consider two different missing covariate patterns to apply the methods \modifbis{from Section~\ref{subsec:sensitivity-partially-obs}}.

\paragraph{Considering \texttt{g1surban} is missing in \modifbis{the} observational study}

\sout{Then } \cite{nguyen2017sensitivity}'s method \modifbis{(recalled in Section~\ref{sec:nguyen-sensitivity-method})} can be applied, \modifbis{if we are given a set of} \sout{and when
considering a range of} plausible \modifbis{values for} $\mathbb{E}[\text{g1surban}]$\sout{, one can
estimate the target population ATE}. \sout{Applying such a method and}
Specifying the following range $ \left]2.1, 2.7\right[\ $
(containing the true value for $\mathbb{E}[\text{g1surban}]$) leads to a
range for the generalized ATE \sout{is}\modifbis{of} $ \left]9.5, 16.7\right[.$ 
Recalling that the \textit{ground truth} is 12.80 (95\%
CI[10.41-15.2]), the estimated range \sout{is in}\modifbis{has a} good overlap with the ground
truth\sout{ value}. In other words, with this specification of the range, a user 
would correctly conclude that without this key variable, the generalized ATE is
probably underestimated.

\paragraph{Considering \texttt{g1surban} is missing in the RCT}

Figure~\ref{fig:simulation-star-heatmap} 
illustrates the method when the missing covariate is in the RCT data set
(see Procedure~\ref{algo:observed-obs}).
This method relies on Assumption~\ref{a:trans-sigma}, which we test with
a Box M-test on
$\Sigma$ (though in practice such a test could only be performed on
$\Sigma_{obs,obs}$). Including only numerical covariates would reject the
null hypothesis ($p-value = 0.034$). Note that beyond violating
Assumption~\ref{a:trans-sigma}, 
some variables  are categorical (\emph{eg} \texttt{race} and \texttt{gender}).
Further discussions about violation of this assumption are available in appendix (Section~\ref{appendix:assumption-8}).

In this application, applying recommendations from Section~\ref{para:remark-on-sensitiviy-delta} (see paragraph entitled \textit{Data-driven approach to determine sensitivity parameter}) allow us to get $\delta_{\texttt{g1surban}} \sim 11$. We consider that the shift is correctly given by domain expert, and so the true shift is taken with uncertainty corresponding to the 95\% confidence interval of a difference in mean.
Finally,  Figure~\ref{fig:simulation-star-heatmap} allows to conclude on a \sout{positive}\modif{negative} bias, that is $\mathbb{E}[\hat \tau_{n, m, obs}] \leq \tau$. Note that our method underestimate a bit the true bias, with an estimated bias of $-6.4$ when the true bias is $-7.08$, delimited with the continue red curve on the top right. 

\begin{figure}[tbh!]
 \begin{minipage}{.48\linewidth}
  \caption{\textbf{Sensitivity analysis of STAR data}: considering the covariate \texttt{g1surban} is missing in the RCT. The black cross indicates the point estimate value for the bias would an expert have the true sensitivity values ($-6.4$) and the true bias value is represented with the red line ($-7.08$). \modif{Dashed lines corresponds to the 95\% confidence intervals around the estimated sensitivity parameters.}}
     \label{fig:simulation-star-heatmap}
    \end{minipage}%
    \hfill%
    \begin{minipage}{.5\linewidth}
    \begin{center}
           \includegraphics[width=0.9\textwidth]{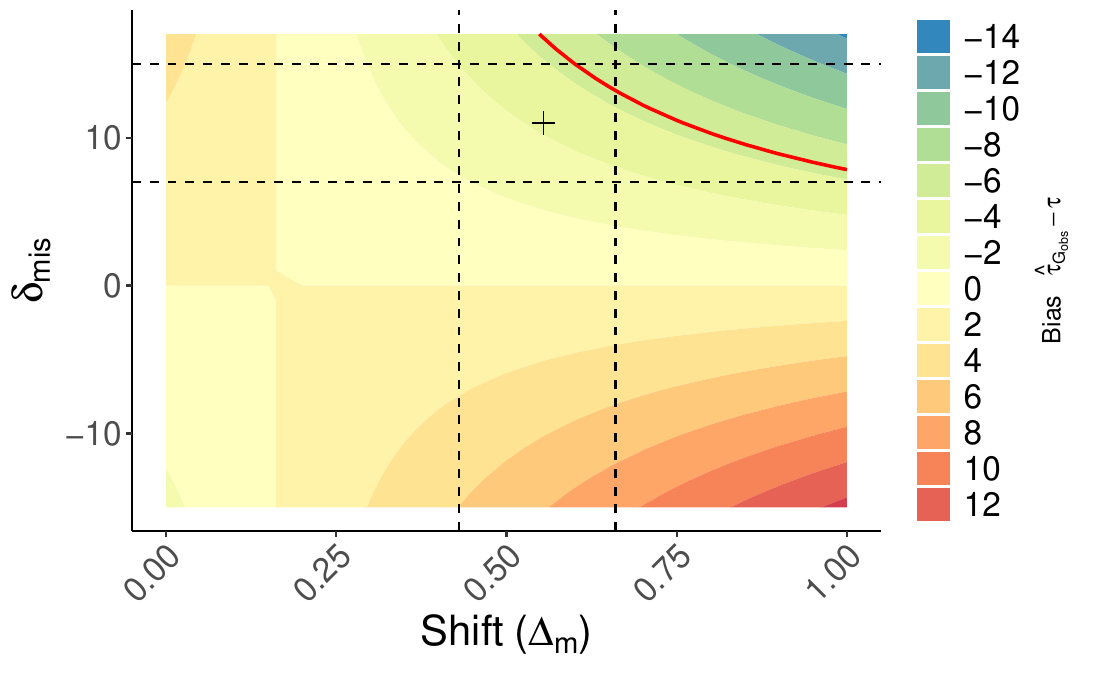}
     \end{center}
    \end{minipage}
\end{figure}

\section{Application on critical care data \label{sec:traumabase}}


A motivating application \modifbis{of our}\sout{for this} work is the generalization to
a French target population -- represented by the
\href{http://www.traumabase.eu/fr_FR}{Traumabase} registry --
of the CRASH-3 trial \citep{crash32019},
evaluating Tranexamic Acide (TXA) to prevent death from Traumatic Brain Injury (TBI). 

\paragraph{CRASH-3}
A total of 175 hospitals in 29 different countries participated to the randomized and placebo-controlled trial, called CRASH-3 \citep{crash3protocol}, where adults with TBI suffering from intracranial bleeding were randomly administrated TXA \citep{crash32019}. 
The inclusion criteria of the trial are patients with a Glasgow Coma
Scale (GCS)\footnote{The Glasgow Coma Scale (GCS) is a neurological scale
which aims to assess a person's consciousness. The lower the score, the
higher the gravity of the trauma.} score of 12 or lower or any
intracranial bleeding on CT scan, and no major extracranial bleeding\sout{,
leading to 9\,202 patients}. The outcome \modif{we consider in this application}\sout{studied} is the Disability Rating Scale (DRS) after 28 days of injury in patients treated within 3 hours of injury. Such an index is a composite ordinal indicator ranging from 0 to 29, the higher the value, the stronger the disability. \modifbis{This outcome can be considered as a secondary outcome.}
This outcome has some drawbacks in the sense that TXA diminishes the probability to die from TBI, and therefore may increase the number of high DRS values \citep{brenner2018outcome}.
Therefore, to avoid a censoring or truncation due to death,\sout{we
consider two approaches: either only individuals with
a mild or moderate TBI (according to \cite{crash32019} this corresponds to individuals with a Glasgow score above 8) and who survived are kept \textit{or} either}\modif{ we keep all individuals \modifbis{and set the DRS score of deceased ones to 30}\sout{with deceased individuals imputed a DRS score of 30}.}\sout{Both approaches provide similar results. The first approach is
presented in this section, and the second approach in appendix (see
Section~\ref{appendix:traumabase-restricted-patient}).}
The difference-in-means estimator\sout{s} gives an ATE of -0.29 with [95\% CI -0.80  0.21]), therefore not giving a significant
evidence of an effect of TXA on DRS.

\paragraph{Traumabase}

To improve decisions and patient care in emergency departments, the
\href{http://www.traumabase.eu/fr_FR}{Traumabase} group, comprising
23 French Trauma centers, collects detailed clinical data from the scene
of the accident to the release from the hospital. The resulting database,
called the Traumabase, comprises 23,000 trauma admissions to date, and is
continually updated. \modif{In this application, we consider only the patients suffering from TBI, along with considering an imputed database. The Traumabase comprises a large number of missing values, this is why we used a multiple imputation via chained equations (MICE) \citep{vanbuuren_2018} prior to \modifbis{applying our methodology}\sout{the application}.} 
\sout{The Traumabase currently comprises around 8,270 patients suffering from TBI.}

\paragraph{Predicting the treatment effect on the Traumabase data}

We want to generalize the treatment effect to the French patients - represented by the Traumabase data base. 
Six covariates are present at baseline, with age, sex, time since injury, systolic blood pressure, Glasgow Coma Scale score (GCS), and pupil reaction. 
\modif{Sex is not considered in the final sensitivity analysis as a non-continuous covariate, and pupil reaction is considered as continuous ranging from 0 to 2.}
However an important treatment effect modifier is missing, that is the time between treatment and the trauma. 
For example, \cite{crash32020timetotreatment} reveal a 10\% reduction in treatment effectiveness for every 20-min increase in time to treatment (TTT).
In addition TTT is probably shifted between the two populations.
Therefore this covariate breaks assumption~\ref{a:cate-indep-s-knowing-X}
(ignorability on trial participation), and we propose to apply the methods developed in Section~\ref{sec:linear-causal-model}.

\paragraph{Sensitivity analysis}

The concatenated data set with the RCT and observational data contains
12\,496 observations (with $n=8\,977$ and $m=7\,743$). 
Considering a totally\modifbis{-}missing covariate, we apply Procedure~\ref{algo:totally-unobserved}.
We assume that time-to-treatment (TTT) is independent of all other variables, for example the ones related to the patient baseline characteristics (e.g. age) or to the severity of the trauma (e.g. the Glasgow score). 
Clinicians support this assumption as the time to receive the treatment depends on the time for the rescuers to come to the accident area, and not on the other patient characteristics. 
We first estimate\sout{d} the target population treatment effect with the set of
observed variables and the G-formula estimator\sout{(Definition~\ref{def:g-formula})}, leading to an estimated \sout{target
population} ATE \sout{of} $\hat \tau_{n,m, obs}$ \modifbis{of} -0.08 (95\% CI [-0.50  0.44]). The nuisance parameters are estimated using random forests, and the confidence interval with non-parametric stratified bootstrap.
As the omission of the TTT variable could affect this conclusion, the
sensitivity analysis gives insights on the potential bias.

\modifbis{We apply the method relative to a completely missing covariate (Section~\ref{subsec:totally-unobserved}). 
A common practice in sensitivity analysis is to use observed covariates as benchmark to guess the impact of an unobserved covariates. 
For example, the Glasgow score is also suspected to be a treatment effect modifier and is shifted between the two populations. We place it on a sensitivity map (Figure~\ref{fig:austen-plot-traumabase-glasgow}) along with the true corresponding values for $\delta_{\text{\tiny glasgow}}$ and $\Delta_{\text{\tiny glasgow}}$.}
\modifbis{As the Traumabase contain more individuals with a higher Glasgow score, a positive shift is reported. In addition, the higher the Glasgow score the higher the effect (low DRS), so that $\delta_{\text{\tiny glasgow}} < 0$. Finally, removing the Glasgow score from the analysis would lead to $\hat \tau_{obs,n,m} > \tau$. The sensitivity map does not allow to conclude that this bias is big enough compared to the confidence intervals previously mentioned for $\hat \tau_{obs,n,m}$.}
\modifbis{\textit{Is the TTT a stronger or more shifted covariate than the Glasgow score?} 
Previous publications have suggested a huge impact of TTT, and therefore one could expect a bigger impact on the bias.
On Figure~\ref{fig:austen-plot-traumabase-ttt} we represent a sensitivity map for TTT that could be drawn by domain experts. 
Here, sensitivity parameters are guessed. 
For example, one can suspect that treatment is given on average 20 minutes earlier in the Traumabase (for example interviewing nurses and doctors in Trauma centers), and the coefficient $\delta_{\text{\tiny TTT}}$ is inferred from a previous work on TXA.}
\modifbis{On Figure~\ref{fig:austen-plot-traumabase-ttt}, one can see that not observing TTT has a bigger impact on the bias than not observing the Glasgow score (almost 10 times bigger), suggesting another conclusion: a positive and significant effect of TXA on the Traumabase population, if the sensitivity parameters are correctly guessed.
Also, as soon as there is a risk of a treatment given later than in the CRASH3 trial, this sensitivity map would help raising an alarm on a negative effect on the Traumabase population.}

\begin{figure}[!htb]
    \centering
    \begin{minipage}{.45\textwidth}
        \centering
        \includegraphics[width=0.95\linewidth]{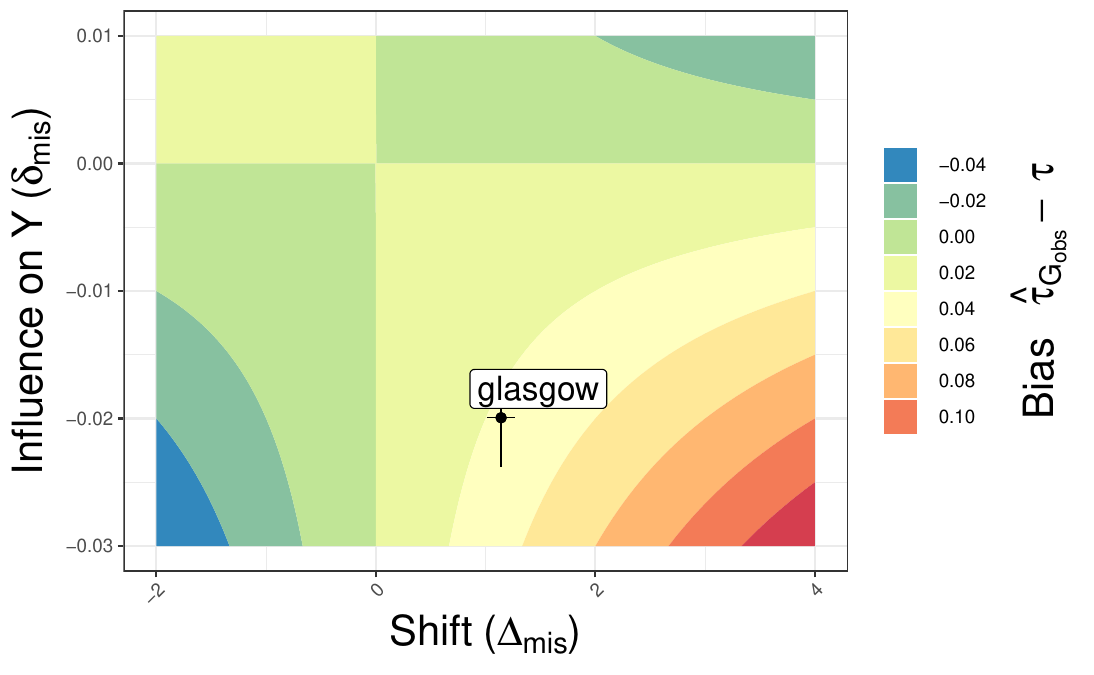}
        \caption{\textbf{Sensitivity map if the Glasgow score covariate was missing}: the true corresponding values for $\delta_{\text{\tiny glasgow}}$ and $\Delta_{\text{\tiny glasgow}}$ are computed with respectively a Robinson proceadure and a mean difference. Intervals correspond to 95\% confidence intervals.}
        \label{fig:austen-plot-traumabase-glasgow}
    \end{minipage}%
    \hspace{0.2cm}
    \begin{minipage}{0.45\textwidth}
        \centering
        \includegraphics[width=0.95\linewidth]{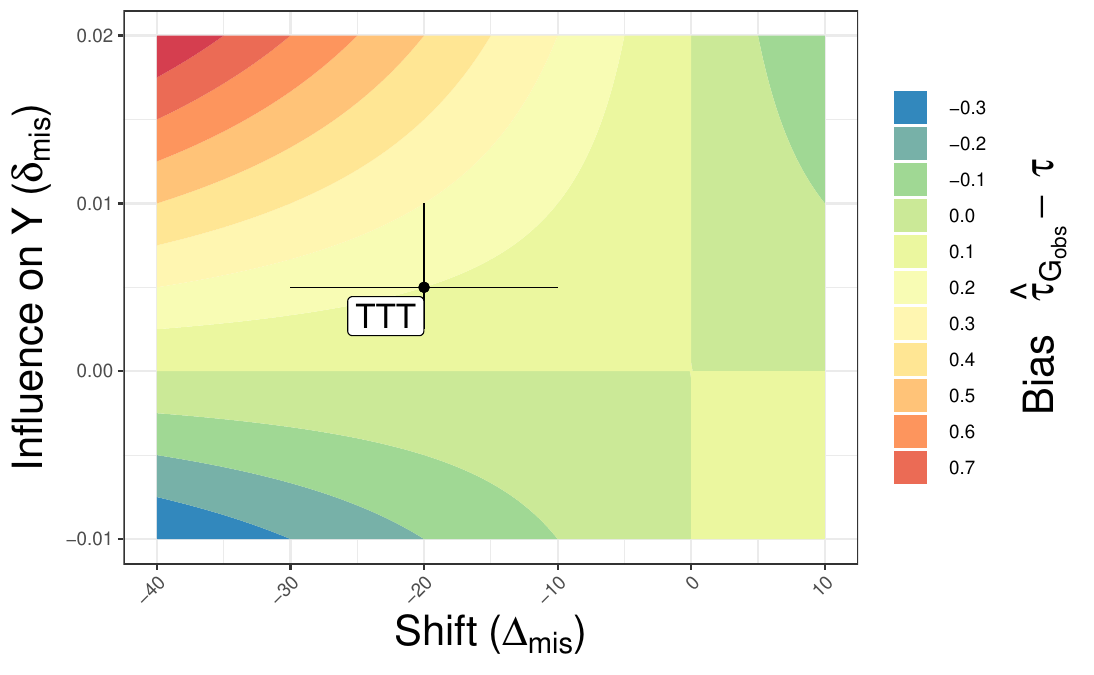}
        \caption{\textbf{Sensitivity map for TTT}: Intervals represent \emph{plausible} parameters range, with a treatment given on average 10 to 30 minutes earlier in the Traumabase, and an heterogeneous coefficient inspired from \cite{crash32020timetotreatment}.}
        \label{fig:austen-plot-traumabase-ttt}
    \end{minipage}
\end{figure}

\section*{Conclusion}

\modifbis{In this work,} we have studied sensitivity analyses for causal-effect generalization \modifbis{to assess the impact of a partially-unobserved confounder (either in the RCT or in the observational data set) on the ATE estimation}. \sout{:
the impact of an unobserved confounder, a missing covariate required for identifiability of the
causal effect.}  \sout{Generalization settings deal with two datasets,
one RCT and on observational data.  Hence,} \sout{We propose procedures suited
for when the covariate is missing in one or the other data, or both. We also investigate solving the issue with imputation or a proxy variable.}
In particular:
\begin{enumerate}
\item To go beyond the common requirement that the unobserved confounder
is independent from the observed covariates, we instead assume that their
covariance is transported (Assumption~\ref{a:trans-sigma}). Our
simulation study (\ref{sec:simulation}) shows that even with a slightly
deformed covariance, the proposed sensitivity analysis procedure give\modifbis{s}
useful estimates of the bias. 

\item Leveraging the high interpretability of our sensitivity parameter,
our framework concludes on the sign of the estimated bias. This sign is
important as accepting a treatment effect highly depends on the direction
of the generalization shift.
We integrate the above methods into the existing \sout{Austen plot}\modif{sensitivity map}
visualization, using a heatmap to represent the sign of the estimated bias.
\item Our procedures use a sensitivity parameter with a direct interpretation: \modifbis{the shift in expectation $\Delta_m$  of the missing covariate between the RCT and the observational data.}
We hope that this will \modifbis{ease practical applications of} \sout{facilitate specifying the} sensitivity analyses by domain experts. 
\end{enumerate}

Our proposal inherits limitations from the more standard sensitivity analysis methods with observational data,
namely \modifbis{the} semi-parametric assumption \modifbis{of}\sout{ on} the outcome model along with an \modifbis{hypothesis on covariate structures (Gaussian inputs)}\sout{ Gaussian covariate distribution}.
Therefore, future extensions of this work could explore ways to relax either the parametric assumption or the distributional assumption to support more robust sensitivity analyses. 
Another possible extension to a missing binary covariate could be deduced from this work, \modifbis{in the case where }\sout{but this would require that} this covariate is independent of the others in both population\modifbis{s}. 

\section*{Acknowledgments} 
We would like to acknowledge helpful discussions with Drs Marine \textsc{Le Morvan}, Daniel \textsc{Malinsky}, and Shu \textsc{Yang}.
We also would like to acknowledge the insights, discussions, and medical expertise from the Traumabase group and physicians, in particular Drs Fran\c{c}ois-Xavier \textsc{Ageron}, and Tobias \textsc{Gauss}. In addition, none of the data analysis part could have been done without the help of Dr. Ian Roberts and the CRASH-3 group, who shared with us the clinical trial data.
\modif{Finally, we thank the reviewers for their careful reading allowing to deeply improve this research work.}

\bibliographystyle{chicago}
\bibliography{bibliography.bib}

\newpage

\appendix

\begin{center}
    {\Large \textbf{Supplementary information}}
\end{center}

\section{Estimators of the target population ATE \label{sec:consistency}}

In this section, we grant assumptions presented in Section~\ref{subsec:notations} and study the asymptotic behavior -- and in particular the $L^1$-consistency -- of three estimators: the G-formula, the IPSW, and the AIPSW.

\subsection{G-formula}
The G-formula procedure and its consistency assumption are detailed in the core text, see Section~\ref{sec:linear-causal-model}, and in particular Definition~\ref{def:g-formula} and Assumption~\ref{a:consistency-mu}. Here, we present the Theorem for consistency.

\begin{theorem}[G-formula consistency]
\label{lem:Consistency_Gformula}
Consider the G-formula estimator in Definition~\ref{def:g-formula} along with Assumptions~\ref{a:RCT-randomization}, \ref{a:eta_1_bounded}, \ref{a:repres}, \ref{a:cate-indep-s-knowing-X}, \ref{a:pos}  (identifiability), and Assumption~\ref{a:consistency-mu} (consistency), then the G-formula estimator converges toward $\tau$ in $L^1$ norm,

\begin{equation*}
    \hat{\tau}_{\text{{\tiny G}}, n, m} \underset{n,m \to \infty}{\overset{L^1}{\longrightarrow}} \tau.
\end{equation*}

\end{theorem}

\subsection{IPSW}

Another approach, called Inverse Propensity Weighting Score (IPSW), consists in weighting the RCT sample so that is ressembles the target population distribution.

\begin{definition}[Inverse Propensity Weighting Score - IPSW - \cite{stuart2011use, buchanan2018generalizing}]\label{def:ipsw} The IPSW estimator is denoted $\hat \tau_{\ipsw, n, m}$, and defined as
\begin{equation}
\hat \tau_{\ipsw, n, m}= \frac{1}{n} \sum_{i=1}^{n} \frac{n}{m} \frac{Y_i}{\hat \alpha_{n,m}(X_i)} \left( \frac{A_i}{e_{1}(X_i)} - \frac{1-A_i}{1- e_{1}(X_i)} \right)\,,
\end{equation}
where $\hat \alpha_{n, m}$ is an estimate of the odd ratio of the indicatrix of being in the RCT:

$$\alpha(x) 
= \frac{\pr(i\in\mathcal{R}\mid \exists i\in\mathcal{R}\cup\mathcal{O}, X_i=x)}{\pr(i\in\mathcal{O}\mid \exists i\in\mathcal{R}\cup\mathcal{O}, X_i=x)}.$$
This intermediary quantity to estimate, $\alpha(.)$, is called a nuisance component.
\end{definition}

\modifbis{Similarly to the G-formula, we introduce here an assumption on the behavior of the nuisance component $\alpha$ to carry out the mathematical analysis of the IPSW.} \sout{ assume the  consistency of the nuisance components of the IPSW, so that the consistency of the IPSW can be proven.}

\begin{assumption}[Consistency assumptions - IPSW]\label{a:consistency-alpha}
Denoting by $\frac{n}{m  \hat \alpha_{n,m}(x)}$ the estimated weights on the set of observed covariates $X$, 
the following conditions hold,
\begin{itemize}
    \item (H1-IPSW) $ \sup _{x \in \mathcal{X}}|\frac{n}{m  \hat \alpha_{n,m}(x)} - \frac{f_{X}(x)}{f_{X \mid S = 1}(x)} |=\epsilon_{n,m} \stackrel{a . s .}{\longrightarrow} 0 \text{ , when } n, m \rightarrow \infty$,
    \item (H2-IPSW) for all $n,m$ large enough $\mathbb{E}[\varepsilon_{n,m}^2]$ exists and $\mathbb{E}[\varepsilon_{n,m}^2] \stackrel{a . s .}{\longrightarrow} 0 \text{ , when } n, m \rightarrow \infty$,
    \item (H3-IPSW) $Y$ is square integrable.
\end{itemize}
\end{assumption}

\begin{theorem}[IPSW consistency]
\label{lem:Consistency_IPSW}
Consider the IPSW estimator in Definition~\ref{def:ipsw} along with Assumptions~\ref{a:RCT-randomization}, \ref{a:eta_1_bounded}, \ref{a:repres}, \ref{a:cate-indep-s-knowing-X}, \ref{a:pos} (identifiability), and Assumption~\ref{a:consistency-alpha} (consistency).
Then, \modif{$\hat \tau_{\text{\tiny IPSW}, n, m}$ converges toward $\tau$ in $L^1$ norm,}
\begin{equation*}
     \hat{\tau}_{\text{{\tiny IPSW}}, n, m} \underset{n,m \to \infty}{\overset{L^1}{\longrightarrow}} \tau.
\end{equation*}

\end{theorem}


Theorem~\ref{lem:Consistency_IPSW} establish\modif{es} the consistency of IPSW in a more general framework than that of \cite{cole2010generalizing, stuart2011use, bareinboim2013general, buchanan2018generalizing, egami2021covariate}, assuming neither oracle estimator, nor parametric assumptions on $\alpha(.)$. 

\subsection{AIPSW}
The model for the expectation of the outcomes among randomized individuals
(used in the G-estimator in Definition~\ref{def:g-formula})
and the model for the probability of trial
participation (used in IPSW estimator in Definition~\ref{def:ipsw}) can be combined to form
an Augmented IPSW estimator (AIPSW) that has a doubly robust statistical property. 

\begin{definition}[Augmented IPSW - AIPSW - \cite{dahabreh2019generalizing}]\label{def:aipsw}  The AIPSW estimator is denoted $\hat \tau_{\text{\tiny AIPSW},n,m}$, and defined as
\begin{equation*}
\begin{split}
\hat \tau_{\text{\tiny AIPSW},n,m}= \frac{1}{n} \sum_{i=1}^{n}  \frac{n}{m \, \hat \alpha_{n,m}(X_i)}\left[\frac{A_{i}\left (Y_{i}-\hat\mu_{1,n}(X_{i})\right ) }{e_{1}(X_i)}-\frac{(1-A_{i})\left ( Y_{i}-\hat\mu_{0,n}(X_{i})\right ) }{1-e_{1}(X_i)}\right]
\nonumber \\
+\frac{1}{m}\sum_{i=n+1}^{m+n}\left( \hat\mu_{1,n}(X_{i})-\hat\mu_{0,n}(X_{i})\right ).
\end{split}
\end{equation*}
\end{definition}

\modifbis{Recently, it has been shown that the AIPSW estimator can be derived from the influence function of the parameter $\tau$  \citep[see][]{dahabreh2019generalizing}. Under additional conditions on the rate of convergence of the nuisance parameters, it is possible to obtain asymptotic normality results\footnote{A primer for semiparametric theory can be found in \cite{kennedy2016semiparametric}.}. As in this work we only require $L^1$-consistency for the sensitivity analysis to hold, we therefore do not detail asymptotic normality conditions.}

To \sout{ensure} \modifbis{prove} AIPSW consistency, \modifbis{we make the following assumptions on the nuisance parameters.} 
\sout{additional assumptions are required on the procedure and the estimated nuisance parameters.}

\begin{assumption}[Consistency assumptions - AIPSW]\label{a:consistency-aipsw}
The nuisance parameters are bounded, and more particularly

\begin{itemize}
      \item (H1-AIPSW) There exists a function $\alpha_0$ bounded from
above and below (from zero), satisfying $$ \lim\limits_{m,n \to \infty}
\sup_{x \in \mathcal{X}}\biggr| \frac{n}{m\hat \alpha_{n,m}(x)} -
\frac{1}{\alpha_0(x)} \biggl| = 0,$$
    \item (H2-AIPSW) 
    There exist two bounded functions $\xi_1, \xi_0: \mathcal{X} \to \R$, such that $\forall a \in \{0, 1\},$ 
    $$\lim\limits_{n \rightarrow +\infty} \sup_{x \in \mathcal{X}}|\xi_{a}(x) - \hat \mu_{a,n}(x)| = 0.$$
\end{itemize}

\end{assumption}

\begin{theorem}[AIPSW consistency]
\label{lem:Consistency_AIPSW}
Consider the AIPSW estimator in Definition~\ref{def:aipsw}, along with Assumptions~\ref{a:RCT-randomization}, \ref{a:eta_1_bounded}, \ref{a:repres}, \ref{a:cate-indep-s-knowing-X}, \ref{a:pos} hold (identifiability), and Assumption~\ref{a:consistency-aipsw} (consistency). Considering that estimated surface responses $\hat \mu_{a,n}(.)$ where $a \in \{0, 1\}$ are obtained following a cross-fitting estimation, then if Assumption~\ref{a:consistency-mu} \textbf{or} Assumption~\ref{a:consistency-alpha} also holds then, \modif{$\hat \tau_{\text{\tiny AIPSW}, n, m}$ converges toward $\tau$ in $L^1$ norm,}
\begin{equation*}
     \hat{\tau}_{\text{{\tiny AIPSW}}, n, m} \underset{n,m \to \infty}{\overset{L^1}{\longrightarrow}} \tau.
\end{equation*}
\end{theorem}

\sout{Proof is available in appendix (see Section~\ref{sec:proof-consistency-ipsw-gformula}).}
\sout{Note that the AIPSW estimator can be derived from the influence function of the parameter $\tau$  \citep[see][]{dahabreh2019generalizing}. Under additional conditions on the rate of convergence of the nuisance parameters, it is possible to obtain asymptotic normality results. A primer for semiparametric theory can be found in \cite{kennedy2016semiparametric}. In this work we only require $L^1$-consistency for the sensitivity analysis to hold, and therefore we do not detail asymptotic normality conditions.}

\section{\sout{Consistency}\modif{$L^1$-convergence} of G-formula, IPSW, and AIPSW \label{sec:proof-consistency-ipsw-gformula}}

This appendix contains the proofs of theorems given in  Section~\ref{sec:consistency}. \modif{We recall that this work completes and details existing theoretical work performed by \cite{buchanan2018generalizing} on IPSW (but focused on a so called nested-trial design and assuming parametric model for the weights) and from \cite{dahabreh2020extending} developing results within the semi-parametric theory.}

\subsection{\sout{Consistency}\modif{$L^1$-convergence} of G-formula}

This section contains the proof of Theorem~\ref{lem:Consistency_Gformula}, which assumes Assumption~\ref{a:consistency-mu}. For the state of clarity, we recall here Assumption~\ref{a:consistency-mu}.
Denoting $\hat \mu_{0,n}(.)$ and $\hat \mu_{1,n}(.)$ estimators of $\mu_{0}(.)$ and $\mu_{1}(.)$ respectively, and $\mathcal{D}_n$ the RCT sample, so that 
\begin{itemize}
    \item (H1-G) For $a \in \{0,1\}$, $\mathbb{E}\left[ | \hat \mu_{a,n}(X) -  \mu_{a}(X) | \mid \mathcal{D}_n\right] \stackrel{p}{\rightarrow}  0$ when $n \rightarrow \infty$, 
    \item (H2-G) For $a \in \{0,1\}$, there exist $C_1, N_1$ so that for all $ n \geqslant N_{1}$, almost surely, $\mathbb{E}[\hat{\mu}_{a,n}^2(X) \mid \mathcal{D}_n] \leqslant C_1$.
\end{itemize}

\begin{proof}[Proof of Theorem~\ref{lem:Consistency_Gformula}]

In this proof, we largely rely on a oracle estimator $\hat{\tau}_{\text{\tiny G}, \infty, m}^*$ (built with the true response surfaces), defined as

\begin{equation*}
    \hat \tau_{\text{\tiny G}, \infty, m}^* = \frac{1}{m} \sum_{i=n+1}^{n+m} \mu_1(X_i) - \mu_0(X_i).
\end{equation*}

The central idea of this proof is to compare the actual G-formula $\hat{\tau}_{\text{\tiny G}, n, m}$ - on which the nuisance parameters are estimated on the RCT data - with the oracle. 

\vspace{0.2cm}
\textbf{$L^1$-convergence of the surface responses}\\

For the proof, we will require that the estimated surface responses $\hat \mu_{1,n}(.)$ and $\hat \mu_{0,n}(.)$ converge toward the true ones in $L^1$. This is implied by assumptions (H1-G) and (H2-G). Indeed, for all $n>0$ and all $a \in \{0, 1\}$, thanks to the triangle inequality and linearity of expectation, we have
\begin{align*}
 \mathbb{E}\left[ \left| \hat \mu_{a,n}(X) - \mu_{a}(X) \right| \mid \mathcal{D}_n \right] &\le     \mathbb{E}\left[ \left| \hat \mu_{a,n}(X)\right| \mid \mathcal{D}_n \right] +\mathbb{E}\left[ \left| \mu_{a}(X)\right| \mid \mathcal{D}_n \right] \\
 &=    \underbrace{\mathbb{E}\left[ \left| \hat \mu_{a,n}(X)\right| \mid \mathcal{D}_n \right]}_\textrm{(*)} + \underbrace{\mathbb{E}\left[ \left| \mu_{a}(X)\right| \right]}_\textrm{(**)}.
\end{align*}

First, note that the quantity (*) is upper bounded thanks to assumption (H2-G), using Jensen's inequality.
Also note that the quantity (**) is upper bounded because the potential outcomes are integrables, that is $\mathbb{E}[|Y(1)|]$ and $\mathbb{E}[|Y(0)|]$ exist (see Section~\ref{subsec:notations}).

Therefore $\mathbb{E}\left[ \left| \hat \mu_{a,n}(X) - \mu_{a}(X) \right| \mid \mathcal{D}_n \right] $ is upper bounded. Consequently, using (H2-G) and a generalization of the dominated convergence theorem, one has

\begin{equation*}
     \mathbb{E}\left[  \left| \hat \mu_{a,n}(X) - \mu_{a}(X) \right| \right] = \mathbb{E}\left[  \mathbb{E}\left[ \left| \hat \mu_{a,n}(X) - \mu_{a}(X) \right| \mid \mathcal{D}_n \right]  \right] \underset{n \to \infty}{\longrightarrow} 0,
\end{equation*}
which implies 

\begin{equation*}
     \forall a \in \{0, 1\}, \, \hat{\mu}_{a,n}(X) \underset{n \to \infty}{\overset{L^1}{\longrightarrow}} \mu_{a}(X).
\end{equation*}

\vspace{0.2cm}
\textbf{$L^1$-convergence of $\hat{\tau}_{\text{\tiny G}, n, m}$ toward $\tau$}\\

For all $m,n >0$, 
\begin{align*}
    \hat{\tau}_{\text{\tiny G}, n, m} - \hat \tau_{\text{\tiny G}, \infty, m}^* &= \frac{1}{m} \sum_{i=n+1}^{n+m} \left( \hat \mu_{1,n}(X_i) -  \hat \mu_{0,n}(X_i)\right) - \left(  \mu_1(X_i) - \mu_0(X_i)\right) \\
    &= \frac{1}{m} \sum_{i=n+1}^{n+m} \left( \hat \mu_{1,n}(X_i) - \mu_1(X_i) \right) - \left( \hat \mu_{0,n}(X_i) - \mu_0(X_i) \right).
\end{align*}

Therefore, taking the expectation of the absolute value on both side, and using the triangle inequality and the fact that observations are iid,
\begin{align*}
    \mathbb{E}\left[ \left|  \hat{\tau}_{\text{\tiny G}, n, m} - \hat \tau_{\text{\tiny G}, \infty, m}^* \right| \right]&=  \mathbb{E}\left[ \left|  \frac{1}{m} \sum_{i=n+1}^{n+m} \left( \hat \mu_{1,n}(X_i) - \mu_1(X_i) \right) - \left( \hat \mu_{0,n}(X_i) - \mu_0(X_i) \right)   \right| \right] \\ 
   & \le \mathbb{E}\left[\left| \hat \mu_{1,n}(X) - \mu_1(X) \right| \right] + \mathbb{E}\left[\left| \hat \mu_{0,n}(X) - \mu_0(X)   \right| \right] .
\end{align*}

Note that this last inequality can be obtained because different observations are used to $(i)$ build the estimated surface responses $\hat{\mu}_{a,n}$ (for $a \in \{0,1\}$) and $(ii)$ to evaluate these estimators. 
Indeed, the proof would be much more complex if the sum was taken over the $n$ observations used to fit the models.
Due to the $L^1$-convergence of each of the surface response  when $n \rightarrow \infty$ (see the first part of the proof), we have 

\begin{equation*}\lim_{n\to\infty}  \mathbb{E}\left[ \left|  \hat{\tau}_{\text{\tiny G}, n, m} - \hat \tau_{\text{\tiny G}, \infty, m}^* \right| \right] = 0. \end{equation*}

In other words,

\begin{equation}\label{eq:L1-convergence-oracle}
    \forall m, \, \hat{\tau}_{\text{{\tiny G}}, n, m} \underset{n \to \infty}{\overset{L^1}{\longrightarrow}} \hat \tau_{\text{\tiny G}, \infty, m}^*.
\end{equation}

This equality is true for any $m$, and intuitively can be understood as the fitted response surfaces $\hat \mu_{a,n}(.)$ can be very close to the true ones as soon as $n$ is large enough. Then, the G-formula estimator, no matter the size of the observational data set, is close to the oracle one in $L^1$. Hence one can deduce a result on the difference between $\tau$ and the G-formula,
\begin{align*}
     \mathbb{E}\left[ \left|  \hat{\tau}_{\text{\tiny G}, n, m} - \tau \right|\right] 
     &\le  \mathbb{E}\left[ \left| \hat{\tau}_{\text{\tiny G}, n, m} - \hat \tau_{\text{\tiny G}, \infty, m}^* \right|\right] + \mathbb{E}\left[ \left|  \hat \tau_{\text{\tiny G}, \infty, m}^*- \tau \right|\right].
\end{align*}
According to the weak law of large number, we have 
\begin{align*}
 \hat \tau_{\text{\tiny G}, \infty, m}^* \underset{m \to \infty}{\overset{L^1}{\longrightarrow}} \tau
\end{align*}
Combining this result with equation~\eqref{eq:L1-convergence-oracle}, we have 
\begin{equation*}
    \hat{\tau}_{\text{{\tiny G}}, n, m} \underset{n, m \to \infty}{\overset{L^1}{\longrightarrow}} \tau,
\end{equation*}
which concludes the proof.

\end{proof}

\subsection{\sout{Consistency}\modif{$L^1$-convergence} of IPSW}

\modif{This section provides the proof of Theorem~\ref{lem:Consistency_IPSW}, and for the sake of clarity, we recall Assumption~\ref{a:consistency-alpha}.
Denoting $\frac{n}{m  \hat \alpha_{n,m}(x)}$, the estimated weights on the set of covariates 
$X$, \sout{and considering $Y$,} the following conditions hold,
\begin{itemize}
    \item  (H1-IPSW) $ \sup _{x \in \mathcal{X}}|\frac{n}{m  \hat \alpha_{n,m}(x)} - \frac{f_{X}(x)}{f_{X \mid S = 1}(x)} |=\varepsilon_{n,m} \stackrel{a . s .}{\longrightarrow} 0 \text{ , when } n, m \rightarrow \infty$,
    \item (H2-IPSW) we have for all $n,m$ large enough $\mathbb{E}[\varepsilon_{n,m}^2]$ exists and $\mathbb{E}[\varepsilon_{n,m}^2] \stackrel{a . s .}{\longrightarrow} 0 \text{ , when } n, m \rightarrow \infty$,
    \item (H3-IPSW) $Y$ is square integrable.
\end{itemize}}

\begin{proof}[Proof of Theorem~\ref{lem:Consistency_IPSW}]

First, we consider an oracle estimator $\hat \tau_{\text{\tiny IPSW}, n}^*$ that is based on the true ratio $\frac{f_{X}(x)}{f_{X \mid S = 1}(x)}$, that is \begin{align*}
\hat \tau_{\text{\tiny IPSW}, n}^* =  \frac{1}{n} \sum_{i=1}^{n} Y_{i} \frac{f_X(X_i)}{f_{X \mid S = 1}(X_i)} \left(\frac{A_{i}}{e_{1}(X_i)}-\frac{1-A_{i}}{1-e_{1}(X_i)}\right).
\end{align*}
Note that \cite{egami2021covariate} also consider such an estimator and document its consistency (see their appendix). Indeed, assuming the finite variance of $Y$, the strong law of large numbers (also called Kolmogorov's law) allows us to state that: 
\begin{align}
\hat \tau_{\text{\tiny IPSW}, n}^* \stackrel{a.  s.}{\longrightarrow} \mathbb{E}\Big[Y \frac{f_X(X)}{f_{X \mid S = 1}(X)} \Big(\frac{A}{e_1(X)} - \frac{1-A}{1-e_1(X)} \Big) \mid S=1 \Big] = \tau, \quad \textrm{as}~n \to \infty.
\label{eq:proof_ipsw}
\end{align}

Now, we need to prove that this result also holds for the estimate $\hat \tau_{\text{\tiny IPSW},n,m}$ \modif{where the weights are estimated from the data.} To this aim, we first use the triangle inequality \modif{comparing $\hat \tau_{\text{\tiny IPSW},n,m}$ with the oracle IPSW}:

\modif{

\begin{align*}
\Big| \hat \tau_{\text{\tiny IPSW},n,m}-\hat \tau_{\text{\tiny IPSW}, n}^* \Big|  &= \Big| \frac{1}{n} \sum_{i=1}^n \left( \frac{A_i Y_i}{e_1(X_i)} - \frac{(1-A_i)Y_i}{1-e_1(X_i)} \right) \left( \frac{n}{\hat \alpha_{n,m}(X_i)m} -\frac{f_X(X_i)}{f_{X \mid S = 1}(X_i)}\right)  \Big|\\
& \le  \frac{1}{n} \sum_{i=1}^n \Big| \left( \frac{A_i Y_i}{e_1(X_i)} - \frac{(1-A_i)Y_i}{1-e_1(X_i)} \right) \left( \frac{n}{\hat \alpha_{n,m}(X_i)m} -\frac{f_X(X_i)}{f_{X \mid S = 1}(X_i)} \right) \Big| && \text{Triangular inequality}\\
&=  \frac{1}{n} \sum_{i=1}^n \Big| \frac{A_i Y_i}{e_1(X_i)} - \frac{(1-A_i)Y_i}{1-e_1(X_i)}  \Big| \Big|\frac{n}{\hat \alpha_{n,m}(X_i)m} -\frac{f_X(X_i)}{f_{X \mid S = 1}(X_i)}  \Big| \\
& \le  \frac{\epsilon_{n,m}}{n} \sum_{i=1}^n \Big| \frac{A_i Y_i}{e_1(X_i)} - \frac{(1-A_i)Y_i}{1-e_1(X_i)}  \Big|   && \text{Assumption~\ref{a:consistency-alpha} (H1-IPSW)} \\
\end{align*}
}

\modif{Taking the expectation on the previous inequality gives,
\begin{align*}
    \mathbb{E}\left[ \Big| \hat \tau_{\text{\tiny IPSW},n,m}-\hat \tau_{\text{\tiny IPSW}, n}^* \Big|  \right] &\le  \mathbb{E}\left[ \varepsilon_{n,m} \frac{1}{n}  \sum_{i=1}^n  \Big|\frac{A_i Y_i}{e_1(X_i)} - \frac{(1-A_i)Y_i}{1-e_1(X_i)}  \Big| \right] \\
    &\le  \sqrt{\mathbb{E}\left[\varepsilon_{n,m}^2 \right]}\sqrt{\mathbb{E}\left[\left(\frac{1}{n} \sum_{i=1}^n  \Big|\frac{A_i Y_i}{e_1(X_i)} - \frac{(1-A_i)Y_i}{1-e_1(X_i)}  \Big|  \right)^2\right] } && \text{C.S., square integ, and H3-IPSW} \\
    &\le \sqrt{\mathbb{E}\left[\varepsilon_{n,m}^2 \right]}\sqrt{\mathbb{E}\left[\left(\ \frac{1}{n} \sum_{i=1}^n  \frac{2 |Y|}{\operatorname{min}(\eta_1, 1-\eta_1)}  \right)^2\right] } && \text{Assumption~\ref{a:eta_1_bounded} and triangular inequality} \\
    &\le \sqrt{\mathbb{E}\left[\varepsilon_{n,m}^2 \right]} \sqrt{\mathbb{E}\left[ \frac{1}{n}\sum_{i=1}^n  \frac{4 |Y|^2}{\operatorname{min}(\eta_1^2, (1-\eta_1)^2)}  \right] }   && \text{Jensen}\\
    &=  \sqrt{\mathbb{E}\left[\varepsilon_{n,m}^2 \right]} \frac{2 \sqrt{\mathbb{E}\left[ Y^2 \right]}}{\operatorname{min}(\eta_1, (1-\eta_1))}     &&\text{iid}
\end{align*}

Therefore, using (H2-IPSW), 
\begin{align} \label{eq:proofIPSW1}
    \mathbb{E}\left[ \Big| \hat \tau_{\text{\tiny IPSW},n,m}-\hat \tau_{\text{\tiny IPSW}, n}^* \Big|  \right] \to 0, \quad \textrm{as}~n,m \to \infty.
\end{align}
%

Finally, note that
\begin{align*}
    \mathbb{E}\left[ \Big| \hat \tau_{\text{\tiny IPSW},n,m}-\tau \Big|  \right] &\le \mathbb{E}\left[ \Big| \hat \tau_{\text{\tiny IPSW},n,m}-\hat \tau_{\text{\tiny IPSW}, n}^* \Big|  \right] + \mathbb{E}\left[ \Big| \hat \tau_{\text{\tiny IPSW}, n}^*-\tau \Big|  \right].
\end{align*}
The second right-hand side term tends to zero by the weak law of large numbers (same reasoning as for the G-formula) and the first term tends to zero using \eqref{eq:proofIPSW1}, which leads to  
\begin{equation*}
\hat{\tau}_{\text{{\tiny IPSW}}, n, m} \underset{n,m \to \infty}{\overset{L^1}{\longrightarrow}} \tau.    
\end{equation*}

}
\end{proof}

\subsection{\sout{Consistency}\modif{$L^1$ convergence} of AIPSW}

The proof of Theorem~\ref{lem:Consistency_AIPSW} is based on  Assumption~\ref{a:consistency-aipsw} and either Assumption~\ref{a:consistency-mu} or Assumption~\ref{a:consistency-alpha}. Therefore the proof contains two parts. for clarity, we recall here Assumption~\ref{a:consistency-aipsw}:


\begin{itemize}
    \item (H1-AIPSW) There exists a function $\alpha_0$ bounded from above and below (from zero), satisfying $$ \lim\limits_{m,n \to \infty} \sup_{x \in \mathcal{X}}| \frac{n}{m\hat \alpha_{n,m}(x)} - \frac{1}{\alpha_0(x)}| = 0,$$
    \item (H2-AIPSW) 
    There exist two bounded functions $\xi_1, \xi_0: \mathcal{X} \to \R$, such that $\forall a \in \{0, 1\},$ 
    $$\lim\limits_{n \rightarrow +\infty} \sup_{x \in \mathcal{X}}|\xi_{a}(x) - \hat \mu_{a,n}(x)| = 0,$$
    and, for all $ i \in \{1, \dots, n\}$,
    $$\lim\limits_{n \rightarrow +\infty} \sup_{x \in \mathcal{X}}|\xi_{a}(x) - \hat \mu^{-k(i)}_{a,n}(x)| = 0.$$
\end{itemize}

\begin{proof}[Proof of Theorem~\ref{lem:Consistency_AIPSW}] 

Note that the cross-fitting procedure supposes to divide the data into $K$ evenly sized folds, where $K$ is typically set to 5 or 10 \modif{(for example see \cite{chernozhukov2017doubledebiased})}. Let $k(.)$ be a mapping from the sample indices $i = 1, \dots, n$ to the $K$ evenly sized data folds, and fit $\hat \mu_{0,n}(.)$ and $\hat \mu_{1,n}(.)$ with cross-fitting over the $K$ folds using methods tuned for optimal predictive accuracy. For $i \in \{1,\dots, n\}$, $\hat \mu_{0,n}^{-k(i)}(.)$ and $\hat \mu_{1,n}^{-k(i)}(.)$ denote response surfaces fitted on all folds except the $k(i)$-th. 
Let us also denote by $\hat \mu_{0,n}(.)$ and $\hat \mu_{1,n}(.)$, the surface responses estimated using the whole data set.\\

\textbf{First case - Assumption~\ref{a:consistency-mu}.}\\

Grant Assumption~\ref{a:consistency-mu}. In this part, we show that, due to this assumption, surface responses are consistently estimated. Recall that the AIPSW estimator $\hat{\tau}_{\aipsw,n,m}$ is defined as
\begin{align*}
    \hat \tau_{\text{\tiny AIPSW},n,m} &= \frac{1}{n} \sum_{i=1}^{n}  \frac{n}{m \hat \alpha_{n,m}(X_i)}\frac{A_{i}\left( Y_{i}- \hat \mu_{1,n}^{-k(i)}(X_{i})\right) }{e_{1}(X_i)} && \text{$A_{n,m}$} \\
    & \qquad - \frac{1}{n} \sum_{i=1}^{n}  \frac{n}{m \hat\alpha_{n,m}(X_i)}\frac{(1-A_{i})\left( Y_{i}- \hat \mu_{0,n}^{-k(i)}(X_{i})\right) }{1-e_{1}(X_i)}  && \text{$B_{n,m}$} \\
    & \qquad  + \frac{1}{m}\sum_{i=n+1}^{m+n}\left(\hat \mu_{1,n}(X_{i})- \hat \mu_{0,n}(X_{i})\right)  && \text{$C_{n,m}$} 
\end{align*}
Note that $\hat \tau_{\text{\tiny AIPSW},n,m}$ is composed of three terms, where the last $C_{m,n}$ corresponds to the G-formula $\hat \tau_{G,n,m}$.
\modif{Now, considering $\mathbb{E}\left[ | \hat \tau_{\text{\tiny AIPSW},n,m} - \tau| \right]$, and using the triangle inequality and linearity of the expectation,

\begin{equation}\label{eq:aipsw-key-ineq-1}
    \mathbb{E}\left[ | \hat \tau_{\text{\tiny AIPSW},n,m} - \tau| \right] \leq \mathbb{E}\left[ | A_{n,m}|\right] + \mathbb{E}\left[ | B_{n,m}|\right] + \mathbb{E}\left[ |\hat \tau_{\text{\tiny G},n,m} - \tau |\right].
\end{equation}

}

Because Assumption~\ref{a:consistency-mu} holds and according to Theorem~\ref{lem:Consistency_Gformula}, we have
\begin{align}
\label{eq:proofAIPSW_gformula}
   \mathbb{E}\left[ |\hat \tau_{\text{\tiny G},n,m} - \tau |\right]{\longrightarrow} 0 \text{, when } n, m \rightarrow \infty.
\end{align}

Now, consider the term $A_{n,m}$, so that,
\begin{align*}
    A_{n,m} & = \frac{1}{n} \sum_{i=1}^{n}  \left( \frac{n}{m \hat \alpha_{n,m}(X_i)} - \frac{1}{\alpha_0(X_i)} \right) \frac{A_{i}\left( Y_{i}- \hat \mu_{1,n}^{-k(i)}(X_{i})\right) }{e_{1}(X_i)} && A_{n,m,1}\\
    & \quad + \frac{1}{n} \sum_{i=1}^{n}  \frac{1}{\alpha_0(X_i)}\frac{A_{i}\left( Y_{i}- \hat \mu_{1,n}^{-k(i)}(X_{i})\right) }{e_{1}(X_i)} && A_{n,m,2}.
\end{align*}
Regarding $A_{n,m,1}$, we have
\begin{align*}
 \mathbb{E}[|A_{n,m,1}|] & \leq  \frac{1}{\eta_1}  \sup_{x \in \mathcal{X}} \left| \frac{n}{m\hat \alpha_{n,m}(x)} - \frac{1}{\alpha_0(x)} \right| \left( \mathbb{E}[A_i|Y_i - \hat \mu_{1,n}^{-k(i)}(X_i)|]\right)  \\
 & \leq \frac{1}{\eta_1}  \sup_{x \in \mathcal{X}} \left| \frac{n}{m\hat \alpha_{n,m}(x)} - \frac{1}{\alpha_0(x)} \right| \left( \mathbb{E}[|Y(1)|] + \mathbb{E}[|\xi_{1}(X)|] + \varepsilon \right),
\end{align*}
which tends to zero according to (H1-AIPSW). Regarding $A_{n,m,2}$, by the weak law of large numbers,
\begin{align*}
\frac{1}{n} \sum_{i=1}^{n}  \frac{1}{\alpha_0(X_i)}\frac{A_{i}\left( Y_{i}- \hat \mu_{1,n}^{-k(i)}(X_{i})\right) }{e_{1}(X_i)}   \underset{n \to \infty}{\overset{L^1}{\longrightarrow}} & \mathbb{E}\left[ \frac{1}{\alpha_0(X_i)}\frac{A_{i}\left( Y_{i}- \hat \mu_{1,n}^{-k(i)}(X_{i})\right) }{e_{1}(X_i)}\right] \\
& =  \mathbb{E}\left[ \frac{1}{\alpha_0(X_i)} \mathbb{E}\left[ \left( Y_{i}- \hat \mu_{1,n}^{-k(i)}(X_{i})\right) | X_i,  \mathcal{D}_n^{-k(i)} \right] \right] \\
& =  \mathbb{E}\left[ \frac{1}{\alpha_0(X_i)} \mathbb{E} \left[ \mu_{1}(X) - \hat \mu_{1,n}^{-k(i)}(X)  |  \mathcal{D}_n^{-k(i)} \right] \right],
\end{align*}
where
\begin{align*}
\left| \mathbb{E}\left[ \frac{1}{\alpha_0(X_i)} \mathbb{E} \left[ \mu_{1}(X) - \hat \mu_{1,n}^{-k(i)}(X)  |  \mathcal{D}_n^{-k(i)} \right] \right] \right| \leq \sup_{x \in \mathcal{X}} \left( \frac{1}{\alpha_0(x)} \right) \mathbb{E}\left[ \mathbb{E} \left[ |\mu_{1}(X) - \hat \mu_{1,n}^{-k(i)}(X)|  |  \mathcal{D}_n^{-k(i)} \right] \right],
\end{align*}
which tends to zero according to Assumption~\ref{a:consistency-mu}. Therefore 
\begin{align}
\label{eq:proofAIPSW1}
A_{n,m} \underset{n \to \infty}{\overset{L^1}{\longrightarrow}} 0. 
\end{align}
%
Using equations~\eqref{eq:proofAIPSW_gformula} and \eqref{eq:proofAIPSW1} in \eqref{eq:aipsw-key-ineq-1} along with the $L^1$-convergence of the G-formula toward $\tau$ allows us to conclude that  
\begin{equation*}
    \hat{\tau}_{\text{{\tiny AIPSW}}, n, m} \underset{n, m \to \infty}{\overset{L^1}{\longrightarrow}} \tau.
\end{equation*}\\

\textbf{Second case - Assumption~\ref{a:consistency-alpha}.}\\

Grant Assumption~\ref{a:consistency-alpha}. In this part, we show that, due to this assumption, weights are consistently estimated. Note that the AIPSW estimate can be rewritten as 
\begin{align*}
    \hat \tau_{\text{\tiny AIPSW},n,m} &= \frac{1}{n} \sum_{i=1}^{n}  \frac{n}{m \hat \alpha_{n,m}(X_i)} \left( \frac{A_i Y_i}{e_1(X_i)} - \frac{(1-A_i) Y_i}{1-e_1(X_i)} \right)&& \text{$D_{n,m}$} \\
    & \qquad - \frac{1}{n} \sum_{i=1}^{n} \left(\frac{n}{m \hat\alpha_{n,m}(X_i)}  - \frac{f_X(X_i)}{f_{X \mid S = 1}(X_i)}\right) \left( \frac{A_i\hat \mu_{1,n}^{-k(i)}(X_i)}{e_1(X_i)} \right)  && \text{$E_{n,m}$} \\
    & \qquad + \frac{1}{n} \sum_{i=1}^{n} \left(\frac{n}{m \hat\alpha_{n,m}(X_i)}  - \frac{f_X(X_i)}{f_{X \mid S = 1}(X_i)}\right) \left( \frac{(1-A_i) \hat \mu_{0,n}^{-k(i)}(X_i)}{1-e_1(X_i)} \right)  && \text{$F_{n,m}$} \\
    & \qquad  - \frac{1}{n} \sum_{i=1}^{n} \frac{f_X(X_i)}{f_{X \mid S = 1}(X_i)} \left( \frac{A_i\hat \mu_{1,n}^{-k(i)}(X_i)}{e_1(X_i)} - \frac{(1-A_i) \hat \mu_{0,n}^{-k(i)}(X_i)}{1-e_1(X_i)} \right)  && \text{$G_{n}$} \\
      & \qquad  + \frac{1}{m}\sum_{i=n+1}^{m+n}\left(\hat \mu_{1,n}(X_{i})- \hat \mu_{0,n}(X_{i})\right).  && \text{$C_{n,m}$}
\end{align*}

Again, using the expectation and the triangle inequality, one has,

\begin{equation}\label{eq:key-ineq-aipsw}
    \mathbb{E}\left[ |  \hat \tau_{\text{\tiny AIPSW},n,m} - \tau |\right] \leq  \mathbb{E}\left[ |  D_{n,m} - \tau |\right] + \mathbb{E}\left[ | E_{n,m}|\right] + \mathbb{E}\left[ | F_{n,m} |\right]+ \mathbb{E}\left[ | G_{n} + C_{n,m} |\right]
\end{equation}
Note that the term $D_{n,m}$ corresponds to the IPSW estimator (Definition~\ref{def:ipsw}). According to Assumption~\ref{a:consistency-alpha} \modif{and Theorem~\ref{lem:Consistency_IPSW}}, $\mathbb{E}\left[ |  D_{n,m} - \tau |\right]$ converges to 0 as $n,m \rightarrow \infty$. Now, we study the convergence of each of the remaining terms in equation \eqref{eq:key-ineq-aipsw}.\\

\textbf{Considering $E_{n,m}$ and $F_{n,m}$} \\

Let us now consider the term $E_{n,m}$. First, note that, according to   Assumption~\ref{a:consistency-aipsw} (H2-AIPSW), the estimated surface responses are uniformly bounded for $n$ large enough, that is, there exists $\mu_M > 0$ such that, for all $a \in \{0, 1\}$, for all $n$ large enough, 
$$\sup_{x \in \mathcal{X}}|\hat \mu_{a,n}(x)| \leq \mu_M.$$

It follows that, for all $n$ large enough, 
\begin{align*}
|E_{n,m}| &\leq \frac{1}{n} \sqrt{\sum_{i=1}^{n} \left(\frac{n}{m \hat\alpha_{n,m}(X_i)}  - \frac{f_X(X_i)}{f_{X \mid S = 1}(X_i)}\right)^2} \sqrt{\sum_{i=1}^{n} \left( \frac{A_i\hat \mu_{1,n}^{-k(i)}(X_i)}{e_1(X_i)} \right)^2}  && \text{Cauchy-Schwarz}   \\
& \leq \frac{1}{n} \sqrt{\sum_{i=1}^{n} \left(\frac{n}{m \hat\alpha_{n,m}(X_i)}  - \frac{f_X(X_i)}{f_{X \mid S = 1}(X_i)}\right)^2} \frac{1}{\eta_1} \sqrt{\sum_{i=1}^{n} \left(\hat \mu_{1,n}^{-k(i)}(X_i)\right)^2} && \text{Assumption~\ref{a:eta_1_bounded}}   \\
& \leq \frac{1}{\sqrt{n}} \sqrt{\sum_{i=1}^{n} \left(\frac{n}{m \hat\alpha_{n,m}(X_i)}  - \frac{f_X(X_i)}{f_{X \mid S = 1}(X_i)}\right)^2} \frac{\mu_M}{\eta_1}  && \text{Assumption~\ref{a:consistency-aipsw} (H2-AIPSW)}  \\
 & \rightarrow 0 \text{, when } n, m \to \infty. &&\text{Assumption~\ref{a:consistency-alpha}}
\end{align*}
The reasoning is the same for the term $F_{n,m}$, which also converges uniformly toward 0  when $n, m \to \infty$.\\

\textbf{Considering $G_{n}$ and $C_{n,m}$} \\

By Assumption (H2-AIPSW), for all $\varepsilon>0$, for all $n$ large enough,  for all $x \in \mathcal{X},$
$$\hat \mu_{1,n}(x) \in [\xi_1(x) - \varepsilon, \xi_1(x) + \varepsilon].$$
Therefore, for all $n$ large enough, and for all $m$, 
\begin{align*}
\left| \frac{1}{m}\sum_{i=n+1}^{m+n} \hat \mu_{1,n}(X_{i}) - \frac{1}{m}\sum_{i=n+1}^{m+n} \xi_1(X_i) \right| & \leq  \frac{1}{m}\sum_{i=n+1}^{m+n} |\hat \mu_{1,n}(X_{i}) - \xi_1(X_i)| \\
& \leq \varepsilon.
\end{align*}
Consequently, 
\begin{align*}
\left| C_{n,m} - \frac{1}{m}\sum_{i=n+1}^{m+n} \xi_1(X_i) + \frac{1}{m}\sum_{i=n+1}^{m+n} \xi_0(X_i) \right| \leq 2 \varepsilon.
\end{align*}
Therefore, 
\begin{align*}
\big| C_{n,m} -  \mathbb{E}[\xi_1(X)] +  \mathbb{E}[\xi_0(X)] \big| & \leq 2 \varepsilon + \left| \frac{1}{m}\sum_{i=n+1}^{m+n} \xi_1(X_i) - \mathbb{E}[\xi_1(X)] \right|  + \left| \frac{1}{m}\sum_{i=n+1}^{m+n} \xi_0(X_i) - \mathbb{E}[\xi_0(X)] \right|.
\end{align*}
Hence, by the law of large numbers, 
\begin{align*}
 C_{n,m} \underset{n, m \to \infty}{\overset{L^1}{\longrightarrow}} \mathbb{E}[\xi_1(X)] - \mathbb{E}[\xi_0(X)].
\end{align*}


We can apply the same reasoning for the term $G_{n}$, by taking into account the fact that it uses a cross-fitting strategy. By Assumption~\ref{a:consistency-aipsw} (H2-AIPSW), for all $\varepsilon>0$, for all $n$ large enough, for all $x \in \mathcal{X},$ for all $ i \in \{1, \hdots, n\}$, 
$$\hat \mu_{1,n}^{-k(i)}(x) \in [\xi_1(x) - \epsilon, \xi_1(x) + \epsilon].$$
Using this inequality, we obtain
\begin{align*}
\left| \frac{1}{n}\sum_{i=1}^{n} \frac{f_X(X_i)}{f_{X \mid S = 1}(X_i)} \frac{A_i}{e_1(X_i)}\hat \mu_{1,n}^{-k(i)}(X_{i}) - \frac{1}{n}\sum_{i=1}^{n} \frac{f_X(X_i)}{f_{X \mid S = 1}(X_i)} \frac{A_i}{e_1(X_i)} \xi(X_i) \right| \leq \frac{\varepsilon}{\eta_1} \sup_{x \in \mathcal{X}} \left( \frac{1}{\alpha_0(x)} \right).
\end{align*}
Besides, by the law of large numbers, 
\begin{align*}
\lim\limits_{n\to\infty} \frac{1}{n}\sum_{i=1}^{n} \frac{f_X(X_i)}{f_{X \mid S = 1}(X_i)} \frac{A_i}{e_1(X_i)} \xi_1(X_{i})   = \mathbb{E} \left [\frac{f_X(X_i)}{f_{X \mid S = 1}(X_i)} \frac{A_i}{e_1(X_i)} \xi_a(X) \right ] = \mathbb{E}[\xi_a(X)].    
\end{align*}
Consequently, as above
\begin{align*}
 G_{n,m} \underset{n, m \to \infty}{\overset{L^1}{\longrightarrow}} \mathbb{E}[\xi_0(X)] - \mathbb{E}[\xi_1(X)].
\end{align*}
Finally,  
\begin{align*}
 C_{n,m} + G_{n,m} \underset{n, m \to \infty}{\overset{L^1}{\longrightarrow}} 0, 
\end{align*}
which concludes the proof. 
%



\end{proof}

\section{Proofs for the missing covariate setting \label{sec:appendix-original-proof}}

This section gathers proofs related to the case where key covariates (treatment effect modifiers with distributional shift) are missing. In particular this appendix contains the proofs of results presented in  Section~\ref{sec:linear-causal-model}.

\subsection{Proof of Theorem~\ref{lemma:linear-gformula-unbiased}}

\begin{proof}
\modif{
The Theorem~\ref{lemma:linear-gformula-unbiased} is essentially a statement about the observed distribution. One can first derived what is the partial-identification of $\tau$ under the observed distribution $\tau_{obs}$, that is,

\begin{align*}
    \tau_{obs} &= \mathbb{E}\left[\mathbb{E}[Y(1) - Y(0) \mid X_{o b s} = x_{obs},  S = 1]\right] \\
    &= \mathbb{E}\left[\mathbb{E}[ \langle \delta, X \rangle \mid X_{o b s} = x_{obs},  S = 1] \right] &&\text{Linear CATE} \\
    &= \mathbb{E}\left[ \mathbb{E}[ \langle \delta, X_{o b s} \rangle + \langle \delta, X_{m i s} \rangle \mid X_{o b s} = x_{o b s},  S = 1] \right] &&\text{$X = (X_{m i s}, X_{o b s})$} \\
    &= \mathbb{E}\left[ \langle \delta, X_{o b s} \rangle\right] + \mathbb{E}\left[ \mathbb{E}\left[ \langle \delta, X_{m i s} \rangle \mid X_{o b s} = x_{o b s},  S = 1\right] \right]. && \text{Ignorability}
\end{align*}

As the covariates $X$ are assumed to be a Gaussian vector distributed as $\mathcal{N}(\mu, \Sigma)$, and considering the assumption on the variance-covariance matrix (Assumption~\ref{a:trans-sigma}), one can have an explicit expression of the conditional expectation \citep{ross1998firstbook}.
\begin{equation*}
\mathbb{E}\left[X_{m i s} \mid X_{o b s} = x_{obs}\right]=\mathbb{E}[X_{m i s}]+\Sigma_{\operatorname{mis}, o b s}\left(\Sigma_{o b s, o b s}\right)^{-1}\left(x_{obs}-\mathbb{E}[X_{o b s}]\right).
\end{equation*}

Therefore, pluging this expression into $\tau_{obs}$ and comparing it to $\tau$,

\begin{align*}
\tau - \tau_{obs} &= \langle \delta, \mathbb{E}[X_{m i s}] - \mathbb{E}[X_{m i s} \mid S=1] - \Sigma_{m i s, o b s}\Sigma_{o b s, o b s}^{-1}(\mathbb{E}[X_{o b s}]-\mathbb{E}[X_{o b s}\mid S = 1] \rangle \\ 
&=  \sum_{j \in mis} \delta_j \left( \mathbb{E}[X_j] - \mathbb{E}[X_j \mid S=1] - \Sigma_{j, o b s}\Sigma_{o b s, o b s}^{-1}(\mathbb{E}[X_{o b s}]-\mathbb{E}[X_{o b s}\mid S = 1] \right)
\end{align*}

Note that the last row is only a different way to write the scalar product into a sum.\\

Then, any $L^1$-consistent estimator $hat \tau_{n, m, o b s}$ of $\tau$ on the observed set of covariates will follow 
\begin{align*}
      \lim\limits_{n,m \to \infty} \mathbb{E}[\hat \tau_{n, m, o b s}] - \tau = - \sum_{j \in mis} \delta_j \left( \mathbb{E}[X_j] - \mathbb{E}[X_j \mid S=1] - \Sigma_{j, o b s}\Sigma_{o b s, o b s}^{-1}(\mathbb{E}[X_{o b s}]-\mathbb{E}[X_{o b s}\mid S = 1] \right).
     \end{align*}
}

\end{proof}

\sout{This proof contains the proof for the G-formula and the IPSW. Each of the proof relies on an oracle estimator that knows the nuisance components (could it be the surface response or the weights), and then the consistency of the estimators is derived using assumptions on the convergence rate of the nuisance components, so that estimators have a behavior close to that of the oracles.} 

\subsection{Imputation \label{proof:imputation-no-bias}}

This part contains the proof of Corollary~\ref{lem:imputation}.

\begin{proof}
This proof is divided into two parts, depending on the missing covariate pattern.

\paragraph{Consider the RCT as the complete dataset} We assume that the linear link between the missing covariate $X_{mis}$ and the observed one $X_{obs}$ in the trial population is known, so is the true response surfaces $\mu_1(.)$ and $\mu_0(.)$.  We consider the estimator $\hat \tau_{G, \infty, m, imp}$ based on the two previous oracles quantities. 
We denote by $c_{0}, \dots, c_{\#{obs}}$ the coefficients linking $X_{obs}$ and $X_{mis}$ in the trial, so that, on the event $S=1$, 
\begin{equation}
    X_{m i s} =  
      c_{0} +  \sum_{j \in obs } c_{j}  X_j + \varepsilon, \label{eq_bonus1}
\end{equation}
where $\varepsilon$ is a Gaussian noise satisfying $\E[\varepsilon \mid X_{obs}] =0$ almost surely. Since we assume that the true link between $X_{mis}$ and $X_{obs}$ is known (that is we know the coefficients $c_0, \hdots, c_d$), the imputation of the missing covariate on the observational sample writes
%
\begin{equation}
    \hat X_{m i s} :=  
      c_{0} +  \sum_{j \in obs } c_{j}  X_j. \label{eq_bonus2}
\end{equation}

We denote $\tilde X$ the imputed data set composed of the observed covariates and the imputed one in the observational sample. 
The expectation of the oracle estimator $\hat \tau_{G, \infty,m,imp}$ is defined as,
\begin{align*}
    \mathbb{E}[\hat \tau_{G, \infty,m,imp}] &=\mathbb{E}\left[ \frac{1}{m} \sum_{i=n+1}^{n+m} \mu_1(\tilde X_i) -  \mu_0(\tilde X_i) \right] && \text{By definition of $\hat \tau_{G, \infty,m,imp}$}\\
    &= \mathbb{E}\left[ \frac{1}{m} \sum_{i=n+1}^{n+m} \langle \delta, \tilde X_i \rangle \right] && \text{Linear CATE \eqref{eq:linear-causal-model}} \\
    &=  \mathbb{E}\left[ \frac{1}{m} \sum_{i=n+1}^{n+m} \left(  \left( \sum_{j \in obs} \delta_j X_{j,i}\right) + \delta_{mis} \hat X_{mis,i} \right) \right] && 
    \\
\end{align*}


Because of the finite variance of $X_{obs}$ and $\hat X_{mis}$ the law of large numbers allows to state that: 
\begin{align*}
\lim\limits_{m \to \infty} \mathbb{E}[\hat \tau_{G, \infty,m,imp}] = \left(  \sum_{j \in obs} \delta_j \mathbb{E}[X_j]\right) + \delta_{mis}\mathbb{E}[\hat X_{mis}].
\end{align*}

Due to Assumption~\ref{a:trans-sigma}, the distribution of the vector $X$ is Gaussian in both populations, and one can use the conditional expectation for a multivariate gaussian law to write the conditional expectation in the trial population, that is
\begin{equation}
\mathbb{E}[X_{m i s} \mid X_{o b s}, S = 1] =  \mathbb{E}[X_{m i s} \mid S = 1] + \Sigma_{m i s, o b s}\Sigma_{o b s, o b s}^{-1}(X_{o b s}-\EE[X_{o b s} \mid S = 1]). \label{eq_bonus3}
\end{equation}
Combining \eqref{eq_bonus1} and \eqref{eq_bonus3}, one can obtain:
\begin{equation}
\label{eq_bonus4}
     c_{0} +  \sum_{j \in obs } c_{j}  X_j = \mathbb{E}[X_{m i s} \mid S = 1] + \Sigma_{m i s, o b s}\Sigma_{o b s, o b s}^{-1}(X_{o b s}-\EE[X_{o b s} \mid S = 1]).
\end{equation}
Now, we can compute, 
\begin{align*}
    \mathbb{E}[\hat X_{mis}] &= \mathbb{E}[ c_{0} +  \sum_{j \in obs } c_{j}  X_j] \\
    &= \mathbb{E}\left[ \mathbb{E}[X_{m i s} \mid S = 1] + \Sigma_{m i s, o b s}\Sigma_{o b s, o b s}^{-1}(X_{o b s}-\EE[X_{o b s} \mid S = 1]) \right] && \text{\eqref{eq_bonus4}} \\
    &=  \mathbb{E}[X_{m i s} \mid S = 1] + \Sigma_{m i s, o b s}\Sigma_{o b s, o b s}^{-1}(\EE[X_{o b s}]-\EE[X_{o b s} \mid S = 1]). 
\end{align*}

This last result allows to conclude that,

\begin{align*}
\lim\limits_{m \to \infty} \mathbb{E}[\hat \tau_{G, \infty,m,imp}] = \left(  \sum_{j \in obs} \delta_j \mathbb{E}[X_j]\right) + \delta_{mis}\left( \mathbb{E}[X_{m i s} \mid S = 1] + \Sigma_{m i s, o b s}\Sigma_{o b s, o b s}^{-1}(\EE[X_{o b s}]-\EE[X_{o b s} \mid S = 1]) \right).
\end{align*}

Finally, as $\tau = \sum_{j=1}^p \delta_j \mathbb{E}[X_j]$,

\begin{align*}
\tau - \lim\limits_{m \to \infty} \mathbb{E}[\hat \tau_{G, \infty,m,imp}]  = \delta_{mis}\left( \mathbb{E}[X_{mis}] -  \mathbb{E}[X_{m i s} \mid S = 1] - \Sigma_{m i s, o b s}\Sigma_{o b s, o b s}^{-1}(\EE[X_{o b s}]-\EE[X_{o b s} \mid S = 1]) \right),
\end{align*}

which concludes this part of the proof.

\paragraph{Consider the observational data as the complete data set} 
We assume here that the true relations between $X_{mis}$ and $X_{obs}$ is known and the true response model is also known. We denote by $\tau_{G, \infty, \infty, imp}$ the estimator based on these two quantities.

More precisely, 
we denote by $c_{0}, \dots, c_{\#{obs}}$ the coefficients linking $X_{obs}$ and $X_{mis}$ in the observational population, so that 
\begin{equation}
    X_{m i s} =  
      c_{0} +  \sum_{j \in obs } c_{j}  X_j + \varepsilon, \label{eq_bonus5}
\end{equation}
where $\varepsilon$ is a Gaussian noise satisfying $\E[\varepsilon \mid X_{obs}] =0$ almost surely.

As the estimator is an oracle, the relation in \eqref{eq_bonus5} is used to impute the missing covariate in the observational sample, so that

\begin{equation}
    \hat X_{m i s} :=  
      c_{0} +  \sum_{j \in obs } c_{j}  X_j. \label{eq_bonus8}
\end{equation}

We denote $\tilde X$ the imputed data set composed of the observed covariates and the imputed one in the trial population. 
Note that the $\hat X_{mis}$ is a linear combination of $X_{obs}$ in the trial population, and thus a measurable function of $X_{obs}$. This property is used below and labelled as \eqref{eq_bonus8}. 
As $\tau_{G, \infty, \infty, imp}$ is an oracle, one have:
\begin{align*}
    \mathbb{E}[\tau_{G, \infty, \infty, imp}] &= \mathbb{E}\left[ \mathbb{E}[Y(1) - Y(0) \mid \tilde X, S = 1] \right] \\
    &= \mathbb{E}\left[ \mathbb{E}[Y(1) - Y(0) \mid \hat X_{mis}, X_{obs}, S = 1] \right] && 
    \\
    &= \mathbb{E}\left[ \mathbb{E}[Y(1) - Y(0) \mid  X_{obs}, S = 1] \right] && \text{\eqref{eq_bonus8}} \\
    &= \mathbb{E}\left[ (\sum_{j \in obs} \delta_j X_j) + \delta_{mis}\mathbb{E}[X_{mis} \mid X_{obs}, S = 1] \right]. && \text{\eqref{eq:linear-causal-model}} \\
    &=  \left( \sum_{j \in obs} \delta_j \mathbb{E}[X_j]  \right) + \delta_{mis} \mathbb{E}\left[\mathbb{E}[X_{mis} \mid X_{obs}, S = 1] \right] \\
    &= \left( \sum_{j \in obs} \delta_j \mathbb{E}[X_j]  \right) \\
    &\qquad + \delta_{mis} \left( \mathbb{E}[X_{m i s} \mid S = 1] + \Sigma_{m i s, o b s}\Sigma_{o b s, o b s}^{-1}(\mathbb{E}[X_{o b s}] -\EE[X_{o b s} \mid S = 1]) \right). && \text{\eqref{eq_bonus3}} \\
\end{align*}


Finally, as $\tau = \sum_{j=1}^p \delta_j \mathbb{E}[X_j]$,
\begin{align*}
\tau - \mathbb{E}[\tau_{G, \infty, \infty, imp}]  = \delta_{mis}\left( \mathbb{E}[X_{mis}] -  \mathbb{E}[X_{m i s} \mid S = 1] - \Sigma_{m i s, o b s}\Sigma_{o b s, o b s}^{-1}(\EE[X_{o b s}]-\EE[X_{o b s} \mid S = 1]) \right),
\end{align*}
which concludes this part of the proof.

\end{proof}

\subsection{Proxy variable}
 
\begin{proof}[Proof of Lemma~\ref{lemma:bias-proxy}]
\label{proof:bias-proxy}

Recall that we denote $\hat \tau_{G, n, m, p r o x}$ the G-formula estimator using $X_{p r o x}$ instead of $X_{m i s}$ in the G-formula. The derivations of $\hat \tau_{G, n, m, p r o x}$ give:

\begin{align*}
   \mathbb{E}[\hat \tau_{G, n, m, p r o x}] &= \mathbb{E}\left[ \mathbb{E}[Y \mid X_{o b s}, X_{p r o x}, S = 1, A = 1] - \mathbb{E}[Y \mid X_{o b s}, X_{p r o x}, S = 1, A = 0] \right] \\ & \hspace{7cm} \text{Definition of $\hat \tau_{G, n,m,p r o x}$}\\
   &= \mathbb{E}\left[ \mathbb{E}[g(X) + \langle \delta, X \rangle \mid X_{o b s}, X_{p r o x}, S = 1] - \mathbb{E}[ g(X)  \mid X_{o b s}, X_{p r o x}, S = 1] \right] \\ 
   &= \mathbb{E}\left[ \mathbb{E}[\langle \delta, X \rangle \mid X_{o b s}, X_{p r o x}, S = 1]\right]\\ 
   &= \sum_{j \in obs} \delta_j \mathbb{E}[X_j] + \delta_{m i s} \mathbb{E}\left[ \mathbb{E}[X_{m i s} \mid X_{o b s}, X_{p r o x}, S = 1] \right]  \\ & \hspace{7cm} \text{Linearity of $Y$ \eqref{eq:linear-causal-model}} \\
    &= \sum_{j \in obs} \delta_j \mathbb{E}[X_j] + \delta_{m i s} \mathbb{E}\left[ \mathbb{E}[X_{m i s}  \mid X_{p r o x}, S = 1] \right]  \\ & \hspace{7cm}  \text{$X_{m i s} \independent X_{o b s}$ \eqref{a:proxy} and \modif{Assumption~\ref{a:RCT-randomization}}} \\
\end{align*}
The framework of the proxy variable \eqref{a:proxy} allows to have an expression of the conditional expectation of $X_{m i s}$ \citep{ross1998firstbook}:

$$\mathbb{E}\left[ \mathbb{E}[X_{m i s} \mid X_{p r o x}, S = 1] \right] = \mathbb{E}[X_{m i s} \mid S = 1] + \frac{\cov(X_{m i s}, X_{p r o x})}{\mathbb{V}[X_{p r o x}]} (X_{p r o x} - \mathbb{E}[X_{p r o x} \mid S = 1]),$$
where 
\begin{align*}
    \mathbb{V}[X_{p r o x}] &= \mathbb{V}[X_{m i s} + \eta] \\
    &= \mathbb{V}[X_{m i s}] + \mathbb{V}[\eta] + 2 \underbrace{\operatorname{Cov}\left(\eta, X_{m i s}\right)}_{=0 \eqref{a:proxy}} \\
    &= \sigma_{m i s}^2 + \sigma_{p r o x}^2
\end{align*}
and 
\begin{align*}
    \cov(X_{m i s}, X_{p r o x}) &= \mathbb{E}[X_{m i s}X_{p r o x}] - \mathbb{E}[X_{p r o x}]^2 &&\text{$\mathbb{E}[X_{m i s}] = \mathbb{E}[X_{p r o x}]$} \\
    &= \mathbb{E}[X_{p r o x}^2 - \eta X_{p r o x}] - \mathbb{E}[X_{p r o x}]^2 \\
    &=\mathbb{E}[X_{p r o x}^2] - \mathbb{E}[X_{p r o x}]^2 - \mathbb{E}[\eta X_{p r o x}] \\
    &= \mathbb{V}[X_{p r o x}] - \mathbb{E}[\eta X_{m i s}] - \mathbb{E}[\eta^2] \\
    &= \sigma_{m i s}^2 + \sigma_{p r o x}^2 -0 - \sigma_{p r o x}^2 \\
    &= \sigma_{m i s}^2
\end{align*}
Therefore, we have
$$\mathbb{E}\left[ \mathbb{E}[X_{m i s} \mid X_{p r o x}, S = 1] \right] = \mathbb{E}[X_{m i s} \mid S = 1] + \frac{\sigma_{m i s}^2}{\sigma_{m i s}^2 + \sigma_{p r o x}^2} (X_{p r o x} - \mathbb{E}[X_{p r o x} \mid S = 1]),$$
which allows us to complete the first derivation:
\begin{align*}
  \mathbb{E}[\hat \tau_{G, n, m, p r o x}] &= \sum_{j \in obs} \delta_j \mathbb{E}[X_j] + \delta_{m i s} \mathbb{E}\left[  \mathbb{E}[X_{m i s} \mid S = 1] + \frac{\sigma_{m i s}^2}{\sigma_{m i s}^2 + \sigma_{p r o x}^2} (X_{p r o x} - \mathbb{E}[X_{p r o x} \mid S = 1]) \right] \\
   &=  \sum_{j \in obs} \delta_j \mathbb{E}[X_j] + \delta_{m i s} \left(  \mathbb{E}[X_{m i s} \mid S = 1] + \frac{\sigma_{m i s}^2}{\sigma_{m i s}^2 + \sigma_{p r o x}^2} \left(\mathbb{E}[X_{p r o x}] - \mathbb{E}[X_{p r o x} \mid S = 1]) \right) \right)\\
   &=  \sum_{j \in obs} \delta_j \mathbb{E}[X_j] + \delta_{m i s} \left(  \mathbb{E}[X_{m i s} \mid S = 1] + \frac{\sigma_{m i s}^2}{\sigma_{m i s}^2 + \sigma_{p r o x}^2} \left(\mathbb{E}[X_{m i s}] - \mathbb{E}[X_{m i s} \mid S = 1]) \right) \right),
\end{align*}
since $\mathbb{E}[X_{p r o x}\mid S = 1] = \mathbb{E}[X_{m i s}\mid S = 1]$ and $\mathbb{E}[X_{p r o x}] = \mathbb{E}[X_{m i s}]$.
Recalling that $\tau = \sum \delta_j \mathbb{E}[X_j]$, the final form of the bias of $\hat \tau_{G, n, m, p r o x}$ can be obtained as
\begin{equation*}
\tau-  \mathbb{E}[\hat \tau_{G, n, m, p r o x}]  = \delta_{m i s} \left(\mathbb{E}[X_{m i s}] - \mathbb{E}[X_{m i s}\mid S = 1]\right) \left(1- \frac{\sigma_{m i s}^2}{\sigma_{m i s}^2 + \sigma_{p r o x}^2}\right).
\end{equation*}

\end{proof}


\begin{proof}[Proof of Corollary~\ref{corol:bias-proxy}]
Note that the final expression of the bias obtained in the previous proof can not be estimated in all missing covariate patterns. For example, if $X_{m i s}$ is partially observed in the RCT, then an estimate of $\delta_{m i s}$ can be computed, and therefore the bias can be estimated. But in all other missing covariate pattern, a temptation is to estimate $\delta_{p r o x}$ from the regression of $Y$ against $X = (X_{o b s}, X_{p r o x})$ with an OLS procedure. \cite{wooldridge2016introductory} details the infinite sample estimate of such a coefficient:
$$
\lim_{n,m\to\infty} \mathbb{E}\left[\hat \delta_{p r o x} \right] = \delta_{m i s}\frac{\sigma_{m i s}^2}{\sigma_{m i s}^2 + \sigma_{p r o x}^2}$$

Note that the quantity $\frac{\sigma_{m i s}^2}{\sigma_{m i s}^2 + \sigma_{p r o x}^2}$ is always lower than 1, therefore if $\delta_{m i s} \geq 1$, then $\hat \delta_{p r o x}$ underestimates $\delta_{m i s}$. This phenomenon is called the attenuation bias. This point is documented by \cite{wooldridge2016introductory}, and is due to heteroscedasticity in the plug-in regression:
$$\operatorname{Cov}[X_{p r o x}, \varepsilon]=\operatorname{Cov}\left[X_{m i s}+\eta, \epsilon-\delta_{m i s} \eta\right]=-\delta_{m i s} \sigma_{\eta}^{2} \neq 0$$ 

This asymptotic estimate can be plugged-in into the previous bias estimation:
\begin{equation*}
    \tau-\mathbb{E}[\hat \tau_{G, n, m, p r o x}] = \hat \delta_{p r o x} \left(\mathbb{E}[X_{p r o x}] - \mathbb{E}[X_{p r o x} \mid S = 1]\right) \frac{\sigma_{p r o x}^2}{\sigma_{m i s}^2}
\end{equation*}

\end{proof}

\section{Toward a semi-parametric model \label{subsec:toward-non-param}}

This section completes Model~\ref{eq:linear-causal-model}, and justifies why this the assumption of a linear CATE is somewhat natural when considering a continuous outcome $Y$.

For a continuous outcome $Y$, the outcome model can be written with two terms, a baseline and the CATE.  Indeed, when focusing on zero-mean additive-error representations leads to assume that the potential outcomes are generated according to:
\begin{equation}
\label{eq:nonparam-causal-model}
    Y(A) = \mu(A,X) + \varepsilon_A, 
\end{equation}
for some function $\mu \in \mathbf{L}^2(\{0,1\} \times \mathcal{X} \to \mathbb{R})$ and a noise $\varepsilon_A$ satisfying $\EE[ \varepsilon_A \mid X] = 0$ almost surely. 

\begin{lemma}
\label{lemma:simplified-nonparametric-causal-model}
 Assume that the nonparametric generative model of Equation~\eqref{eq:nonparam-causal-model} holds, then there exists a function $g:\mathcal{X} \to \mathbb{R}$ such that
 \begin{align}
 \label{eq:simplified-nonparametric-causal-model}
 Y(A) = g(X) + A\,\tau(X) + \varepsilon_A,\quad    \text{where } \tau(X):=\mathbb{E}[Y(1)-Y(0) \mid X].
 \end{align}
\end{lemma}

Lemma \ref{lemma:simplified-nonparametric-causal-model} follows from rewriting Equation~\eqref{eq:nonparam-causal-model}, accounting for the fact that $A$ is binary and $Y \in \mathbb{R}$. Such a decomposition is often used in the literature \citep{nie2020quasioracle}. 
This model allows to have a simpler expression of the treatment effect without any additional assumptions, due to the discrete nature of $A$. 
\modifbis{In other words, this model enables placing independent functional form 
on the CATE $\tau(X)$, sometimes relying on the idea that the CATE is smoother, while the baseline response can be more complex \citep{Gao2021Dina}. 
In the context of the sensitivity analysis, this model has the interest of highlighting treatment effect modifier variables, such as variables that intervene in the CATE $\tau(X)$. }

\section{Robinson procedure \label{appendix:robinson}}
This appendix recall the so-called Robinson procedure that aims at estimating the CATE coefficients $\delta$ in a semi-parametric equation such as \eqref{eq:linear-causal-model}. This method was developed by \cite{robinson1988semiparam} and has been further extended \citep{chernozhukov2017doubledebiased, wager2020courses, nie2020quasioracle}. Such a procedure is called a R-learner, where the \textit{R} denotes \textit{Robinson} or \textit{Residuals}. We recall the procedure,

\begin{enumerate}
    \item Run a non-parametric regressions $Y \sim X$ using a parametric or non parametric method. The best method can be chosen with a cross-validation procedure. We denote $\hat m_n(x) = \mathbb{E}[Y \mid X =x]$ the estimator obtained.
    \item Define the transformed features $\tilde Y = Y - \hat m_n(X)$ and $\tilde Z = (A-\modif{e_1(X)})X$, using the previous procedure $\hat m_n$.
    \item Estimate $\hat \delta_n$ running the OLS regression on the transformed features $\tilde Y \sim \tilde Z$.
\end{enumerate}

If the non-parametric regressions of $m(x)$ satisfies $\mathbb{E}\left[(\hat{m}(X)-m(X))^{2}\right]^{\frac{1}{2}} = o_{P}\left(\frac{1}{n^{1 / 4}}\right)$, then the procedure to estimate $\delta$ is $\sqrt{n}$-consistent and asymptotically normal,
$$\sqrt{n}\left(\hat{\delta}-\delta\right) \Rightarrow \mathcal{N}\left(0, V_{R}\right), \quad V_{R}=\operatorname{Var}\left[\widetilde{Z}\right]^{-1} \operatorname{Var}\left[\widetilde{Z} \tilde{Y}\right] \operatorname{Var}\left[\widetilde{Z}\right]^{-1}$$ See \cite{chernozhukov2017doubledebiased, wager2020courses} for details.

\section{Synthetic simulation - Extension \label{appendix:synthetic-simulation-extension}}


This section completes the synthetic simulation presented in Section~\ref{sec:simulation}.

\paragraph{Simulation parameters}

Parameters chosen highlight different covariate roles and strength importance.
In this setting, covariates $X_1$, $X_2$, $X_3$ are the so-called treatment effect modifiers due to a non-zero $\delta$ coefficients, and $X_1$, $X_3$, $X_4$ are shifted from the RCT sample and the target population distribution due to a non-zero $\beta_s$ coefficient. 
Therefore covariates $X_1$ and $X_3$ are necessary to generalize the treatment effect, because in both groups. Because in the simulation $X_2$ and $X_4$ are independent, the set $X_1$ and $X_3$ is also sufficient to generalize.
Only $X_2$ has the same marginal distribution in the RCT sample and in the observational study.
Note that the amplitude and sign of different coefficients used, along with dependence between variables allows to illustrate several phenomenons.
For example $X_3$ is less shifted in between the two samples compared to $X_1$ because $|\beta_{s, 1}| \leq |\beta_{s, 3}|$.

\paragraph{Additional comments on Figure~\ref{fig:simulations-all-pattern}}

Note that depending on the correlation strength between $X_1$ and $X_5$, the missingness of $X_1$ can lead to different coefficients estimations when using the G-formula estimation, and different bias on the ATE. Table~\ref{tab:simulation-linear-effects-on-coefficients} illustrates this situation, where the higher the correlation, the higher the error on the coefficients estimations, but the lower the bias on the ATE when only $X_1$ is missing.

\begin{table}[!h]
\begin{minipage}{.5\linewidth}
    \caption{\textbf{Coefficients estimated in the simulation}: Simulation with $X_1$ as the missing covariate repeated 100 times, means of estimated coefficients for $X_5$ and bias on ATE using the Robinson procedure.} 
    \label{tab:simulation-linear-effects-on-coefficients}
\end{minipage}%
\begin{minipage}{.5\linewidth}
\begin{center}
\begin{tabular}{cccc}
\hline
$\rho_{X_1, X_5}$  &  $\delta_5 - \hat \delta_5 $ &     $\hat \tau_{\text{\tiny G},o b s} - \tau$                 \\ \hline
0.05 &   6.34  &  -8.32  \\ 
0.5  &   16.83  & -6.29 \\ 
0.95 &  28.53 & -0.81 \\ 
\multicolumn{1}{l}{} & \multicolumn{1}{l}{} & \multicolumn{1}{l}{} & \multicolumn{1}{l}{}
\end{tabular}
\end{center}
\end{minipage}
\end{table}

\paragraph{Imputation}

When a covariate is partially observed, at temptation is to imputed the missing part with a model learned on the complete part as detailed in procedure~\ref{algo:imputation}.
Section~\ref{sec:linear-causal-model} illustrates Corollary~\ref{lem:imputation}, as it shows that linear imputation does not diminish the bias compared to a case where the generalization is performed using only the restricted set of observed covariates. On Figure~\ref{fig:simulations-imputation} we simulated all the missing covariate patterns (in RCT or in observational) considering $X_1$ is partially missing, with varying correlation strength between $X_5$ and $X_1$, and fitting a linear imputation model.
Imputation does not lead to a lower bias than totally removing the partially observed covariate. Therefore, in case of a partially missing covariate we advocate running a sensitivity analysis rather than a linear imputation.

\begin{figure}[!h]
    \centering
    \includegraphics[width=0.9\textwidth]{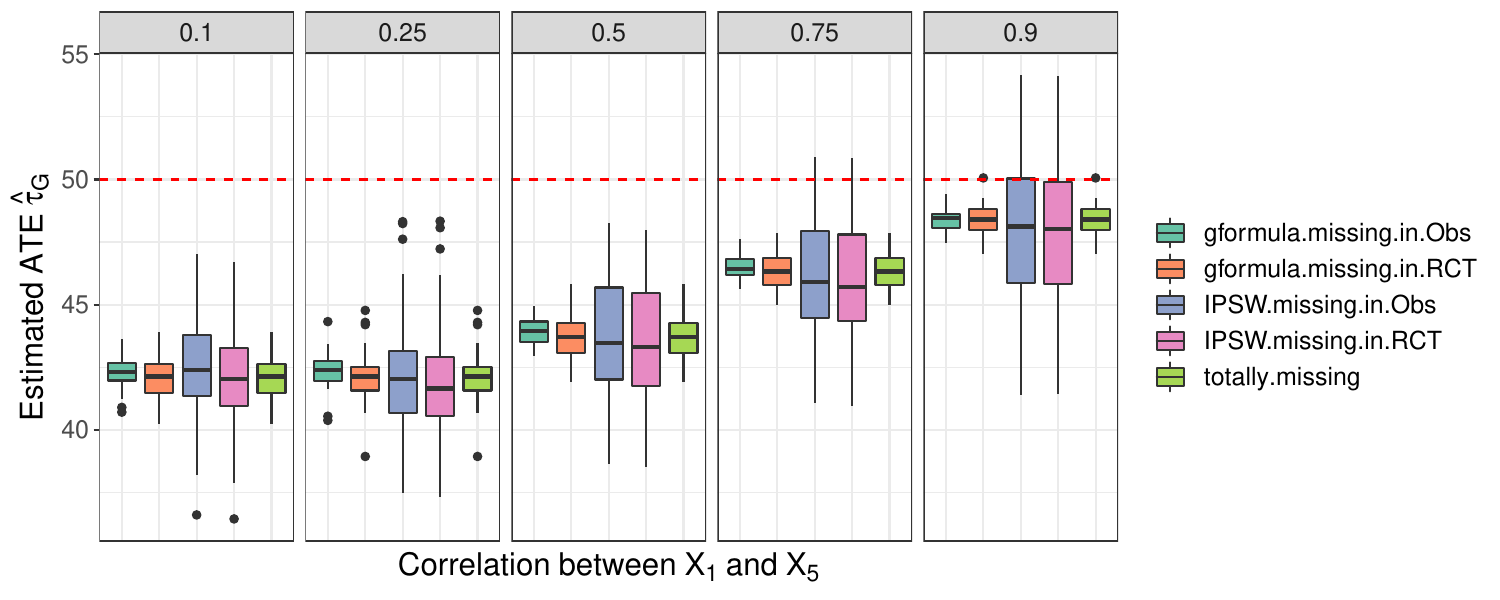}
    \caption{\textbf{Simulations results when imputing}
(procedure~\ref{algo:imputation}): Results when imputing $X_1$ with a
linear model fitted on the complete data set (either the RCT or the
observational). All the  missing covariate pattern are simulated using
either the G-formula or the IPSW estimators. The impact of the
correlation between $X_1$ and $X_5$ is investigated. Each simulation is
repeated 100 times. All procedures have a similar bias as the procedure
ignoring the partially-missing covariate (\texttt{totally.missing}), so
that a linear imputation (procedure~\ref{algo:imputation}) improves neither the bias nor the variance.
}
    \label{fig:simulations-imputation}
\end{figure}

\paragraph{Proxy variable}

Finally and to illustrate Lemma~\ref{lemma:bias-proxy}, the simulation is
extended to replace $X_1$ by 
a proxy variable, generated following \eqref{a:proxy} with a varying
$\sigma_{p r o x}$. The generalized ATE is estimated with the G-formula. The experiments is repeated 20 times per $\sigma_{p r o x}$ values. Results are presented on Figure~\ref{fig:simulation-linear-proxy}. Whenever $\sigma_{p r o x}$ is small compared to $\sigma_{m i s}$ (which is equal to one in this simulation), therefore the bias is small. 

\begin{figure}[H]
    \centering
    \includegraphics[width=0.75\textwidth]{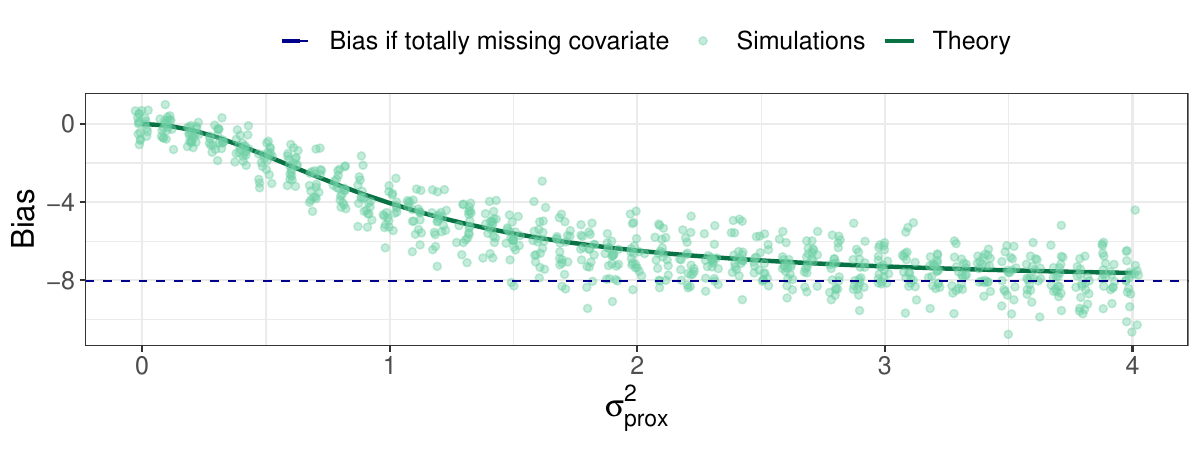}
    \caption{\textbf{Simulation results for proxy variable}
(procedure~\ref{algo:proxy}) Simulation when a key covariate is replaced by a proxy following the proxy-framework (see Assumption~\ref{a:proxy}). The theoretical bias \eqref{lemma:bias-proxy} is represented along with the empirical values obtained when generalizing the ATE with the plugged-in G-formula estimator.}
    \label{fig:simulation-linear-proxy}
\end{figure}

\section{Homogeneity of the variance-covariance matrix \label{appendix:assumption-8}}

Recall that Assumption~\ref{a:trans-sigma} states that the covariance matrices in both data sets are identical. This assumption, which may appear to be very restrictive, can be partially tested on the set of observed covariates. In this section, we present such a test \citep[Box's M-test][]{box1949boxmtest}, which illustrates the validity of Assumption~\ref{a:trans-sigma} on some particular data set. Taking one step further, we study the impact of Assumption~\ref{a:trans-sigma} violation on the resutling estimate.


\subsection{Statistical test and visualizations}

\cite{friendly2020covariancematrixviz} detail available tests to assess if covariance matrices from two data sample are equal. Despite its sensitivity to violation, Box's M-test \citep{box1949boxmtest} can be used test the equality. In particular the package \texttt{heplots} contains the tests and visualizations in \texttt{R}. The command line to perform the test is detailed below.

\begin{lstlisting}
library(heplots)
boxM(data[, c("X1", "X2", "X3", "X4")], group = data$S)
\end{lstlisting}

\begin{wrapfigure}{r}{0.4\textwidth}
    \centering
    \includegraphics[width=6cm]{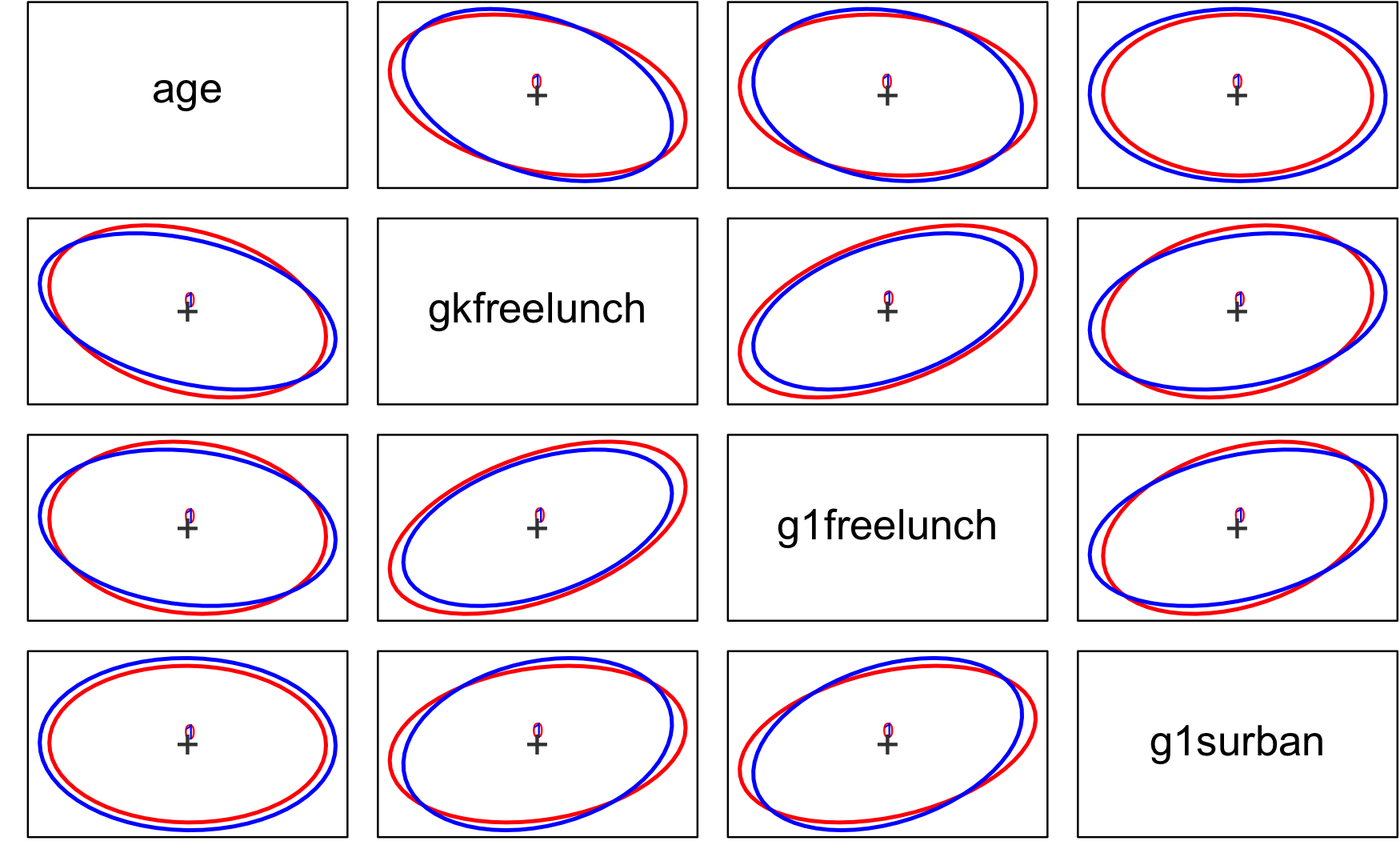}
    \caption{\textbf{Pairwise data ellipses for the STAR data}, centered at the origin. This view allows to compare the variances and covariances for all pairs of variables.}
    \label{fig:covariance-matrix-star}
\end{wrapfigure}

Even if we cannot bring a general rule to know if the covariance matrices are equal, we can display some examples in which Assumption~\ref{a:trans-sigma} holds. For instance, \cite{friendly2020covariancematrixviz} report that the \texttt{skull} data is an example of a real data set with multiple sources where there are substantial differences among the means of groups, but little evidence for heterogeneity of their covariance matrices.

\subsubsection{Semi-synthetic experiment: STAR}
While doing the semi-synthetic experiment on the STAR data set, the Box M-test rejects the null hypothesis when considering only numerical covariates (\texttt{age}, \texttt{g1freelunch}, \texttt{gkfreelunch}, and \texttt{g1surban}) with a p-value of 0.022. 
This indicates that the preservation of the variance-covariance structure between the two simulated sources does not hold.
To help support conclusions, one can visualize how the variance covariance matrix vary in between the two sources, as presented on Figure~\ref{fig:covariance-matrix-star}, supporting that the changes in the variance-covariance are not very strong.

\subsubsection{Traumabase and CRASH-3}

\textit{Note that this part's purpose is only to illustrate the principle as the application performed in Section~\ref{sec:traumabase} relies on the independence between the time to treatment and all other covariates, and not Assumption~\ref{a:trans-sigma}.}

One can inspect how far the variance and covariance change in between the two sources. Pairwise data ellipses are presented on Figure~\ref{fig:covariance-matrix-crash-TB} for CRASH-3 and Traumabase patients, suggesting rather strong difference in the variance-covariance matrix. As expected Box M-test largely rejects the null hypothesis.

\begin{figure}[!h]
    \centering
    \includegraphics[width=0.5\textwidth]{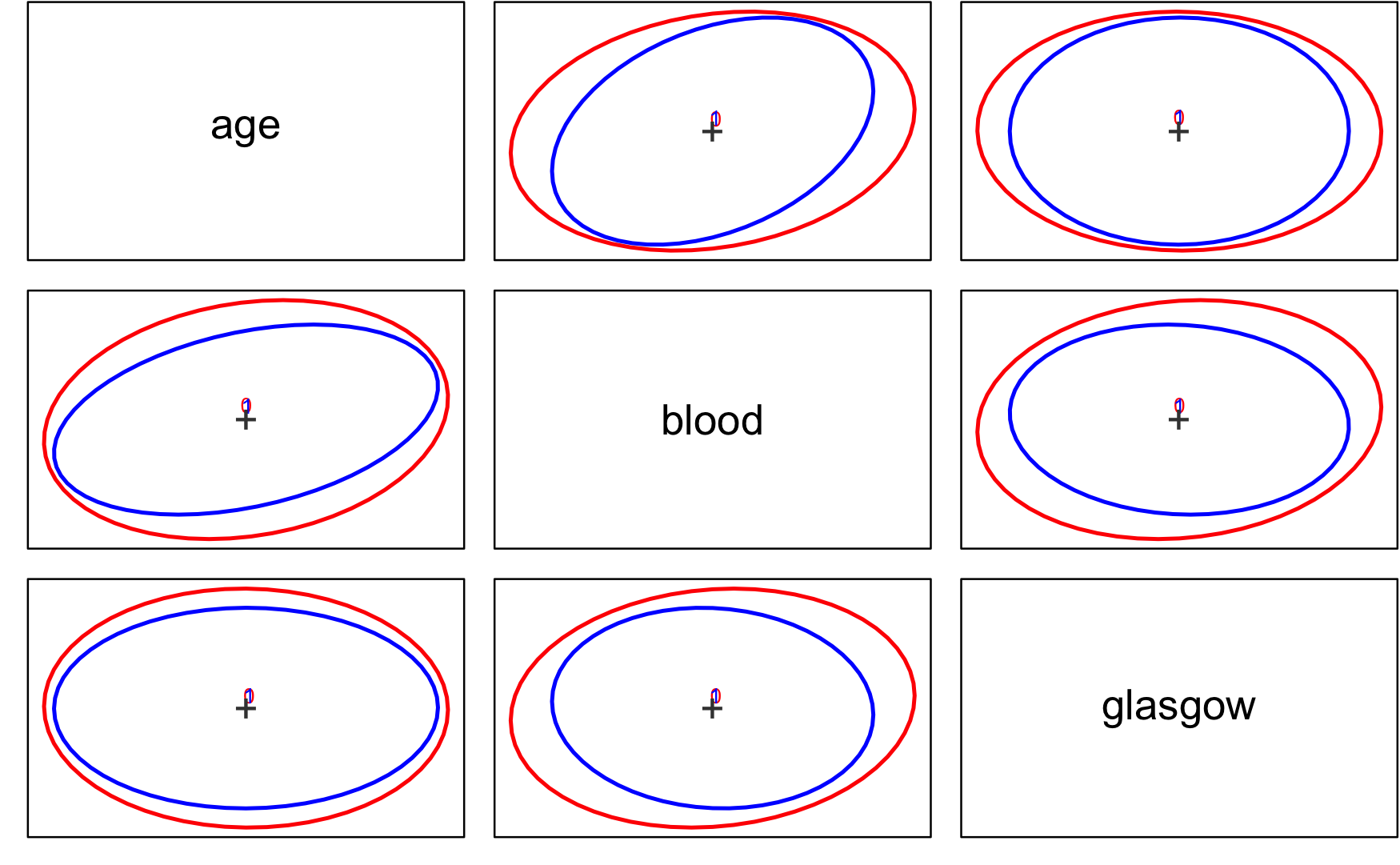}
    \caption{\textbf{Pairwise data ellipses for the CRASH-3 and Traumabase data}, centered at the origin. CRASH-3 data are in blue and Traumabase data in red. This view allows to compare the variances and covariances for all pairs of variables. While the mean are really different in the two sources, the variances and covariances are not so different.}
    \label{fig:covariance-matrix-crash-TB}
\end{figure}

It is interesting to note that in some cases the variance covariance matrix is identical in between two populations. For example we tested whether the two major trauma centers in France present heterogeneity in the variance-covariance matrix, and the Box M test does not reject the null hypothesis.

\subsection{Extension of the simulations \label{appendix:assumption-8-simulation}}

Simulations presented in Section~\ref{sec:simulation} can be extended to illustrate empirically the consequences of a poorly specified Assumption~\ref{a:trans-sigma}. 
Suppose $X_1$ is the unobserved covariate, and that the variance-covariance matrix is not the same in the randomized population ($S=1$) as in the target population. But the heterogeneities in between the two sources can be different in their nature, affecting covariates depending or not from $X_1$. We can imagine two situations, a situation (A) where the link in between $X_1$ and $X_5$ is different in the two sources, and another situation (B) where the link in between $X_2$ and $X_3$ is not the same. The situation is illustrated on Figures~\ref{fig:situation-A} and \ref{fig:situation-B} with pairwise data ellipses. Note that with $n=1000$ and $m=10000$ a Box-M test largely rejects the null-hypothesis with a similar statistic value for both situations. When computing the bias according to Theorem~\ref{lemma:linear-gformula-unbiased} and repeating the experiment $50$ times, empirical evidence is made that the localization of the heterogeneity impacts or not the bias computation. As presented on Figure~\ref{fig:covariance-matrix-simulation-ATE}, situation A affects the bias computation, when situation B keeps the bias estimation valid.


\begin{figure}[!h]
\begin{subfigure}{.5\linewidth}
\centering
\includegraphics[width=0.8\textwidth]{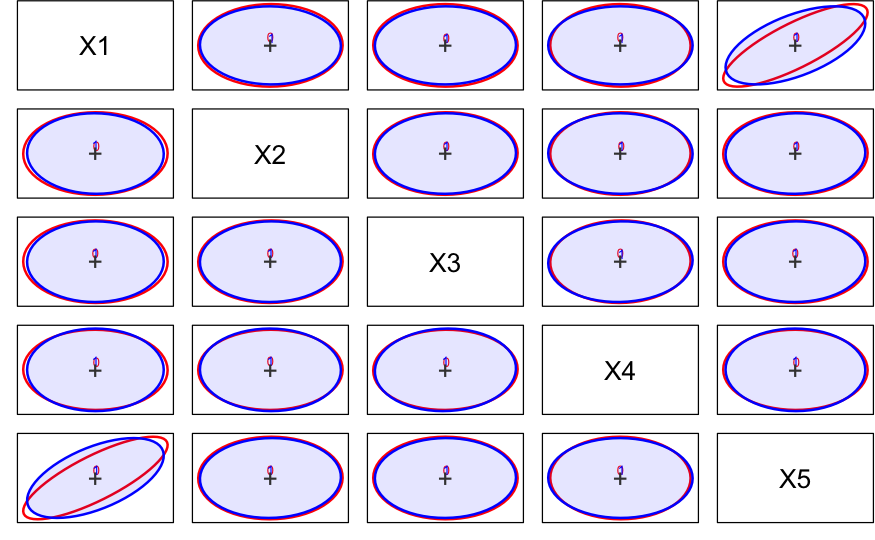}
\caption{\textbf{Situation A} - Centered pairwise data ellipses}
\label{fig:situation-A}
\end{subfigure}%
\begin{subfigure}{.5\linewidth}
\centering
\includegraphics[width=0.8\textwidth]{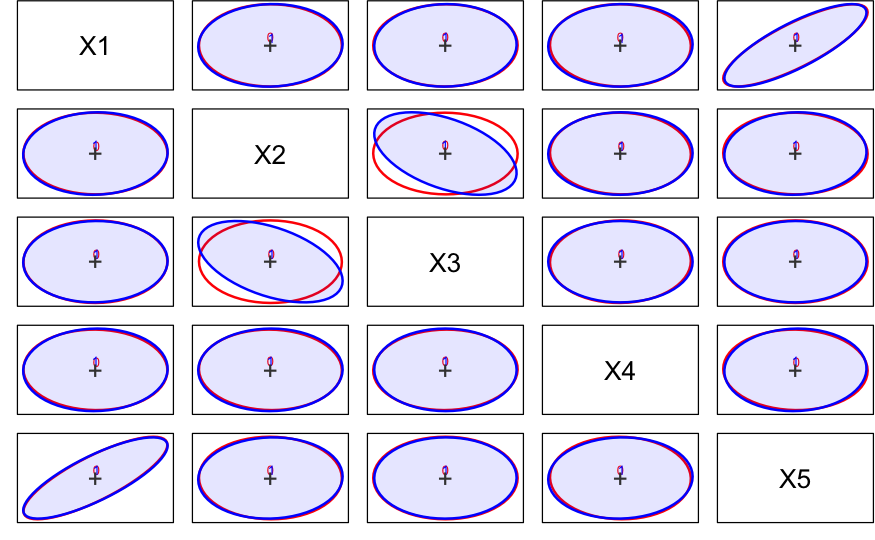}
\caption{\textbf{Situation B} - Centered pairwise data ellipses}
\label{fig:situation-B}
\end{subfigure}\\[1ex]
\begin{subfigure}{\linewidth}
\centering
\includegraphics[width=0.7\textwidth]{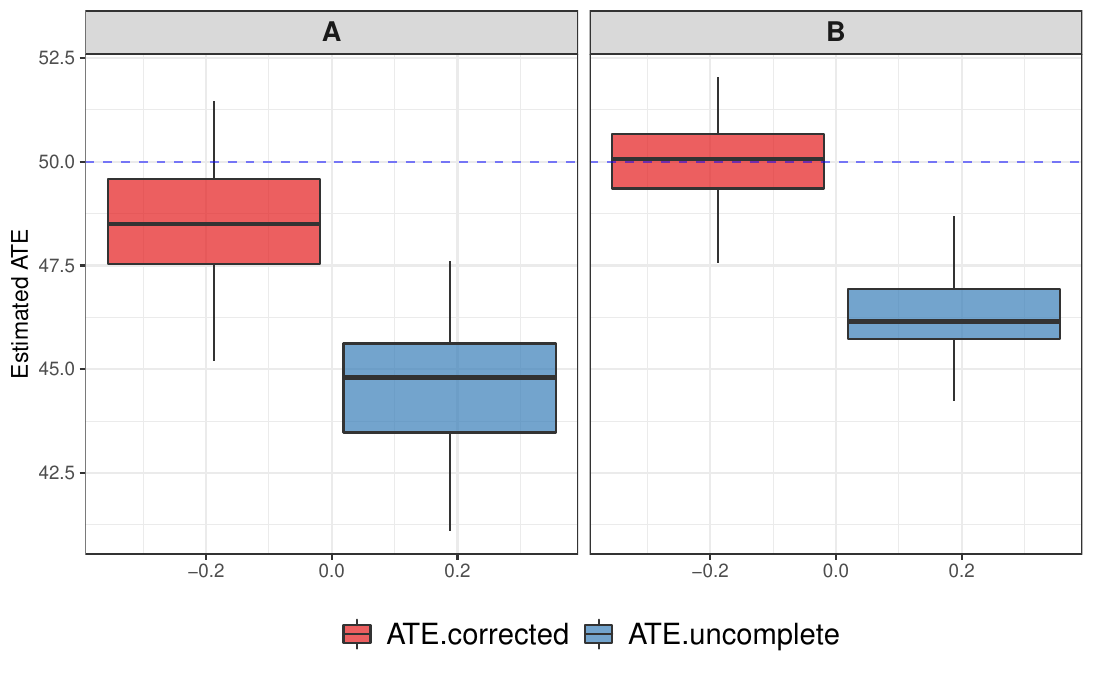}
\caption{\textbf{ATE estimation in the two situations} where $\hat \tau_{\text{\tiny G}, obs}$ is estimated considering $X_1$ is missing and denoted \texttt{ATE.uncomplete}, while the bias $B$ is estimated following Theorem~\ref{lemma:linear-gformula-unbiased} giving \texttt{ATE.corrected} ($\hat \tau_{\text{\tiny G}, obs} + \hat B$).}
\label{fig:covariance-matrix-simulation-ATE}
\end{subfigure}
\caption{\textbf{Effect of a different variance-covariance matrix on the ATE estimation}, where heterogeneity between the two variance-covariance matrix is introduced as presented in (a) and (b), and on (c) the impact on the estimated average treatment effect (ATE). Situations A and B result in a similar statistics when using a Box-M test, but leads to very different impact on the bias estimation as visible on (c). The simulation are repeated 50 times, with a similar outcome generative model as in \eqref{eq:Ymodel}, and $n=1000$ and $m=10000$.}
\label{fig:covariance-matrix-simulation}
\end{figure}

\subsection{Recommendations}

Our current recommendations when considering the Assumption~\ref{a:trans-sigma} is, first, to visualize the heterogeneity of variance-covariance matrix with pairwise data ellipses on $\Sigma_{obs,obs}$. A statistical test such as a Box-M test can be applied on $\Sigma_{obs,obs}$. 
We also want to emphasize that a statistical test depends on the size of the data sample, when what really matters in this assumption for the sensitivity analysis to be valid is the permanence of covariance structure of the missing covariates with the strongly correlated observed covariates. Simulations presented on Figure~\ref{fig:covariance-matrix-simulation-ATE} is somehow an empirical pathological case where the variance-covariance matrix are equivalently different when considering a statistical test, but leads to different consequences on the validity of Theorem~\ref{lemma:linear-gformula-unbiased}, and therefore the sensitivity analysis.

\modifbis{
\subsection{Comment about the notations}

The notations used in this work inherits from the generalization literature and reflects the idea of a plausibility to be sampled from a target superpopulation. The point of view of two population with support inclusion is equivalent for our purpose. Still, thinking to the problem of a sampling bias, then Assumption~\ref{a:trans-sigma} imposes unusual restrictions for $P(X \mid S =0)$, that is a subpopulation of the target population. As we do not do any inference on that population and as it has no practical interpretation, we do not discuss this in this work.}

\end{document}